%% file: main.tex
\documentclass[a4paper,11pt]{article}
\input{MathPhysDef}

\providecommand{\Tr}{\operatorname{Tr}}

\DeclareMathOperator{\Id}{id}

\DeclareMathOperator{\diag}{diag}
\DeclareMathOperator{\rank}{rank}

\providecommand{\I}{\mathds{1}}
\providecommand{\BB}{\mathbf{B}}

\providecommand{\JJ}{\mathbb{J}}

\providecommand{\MH}{\mathrm{MH}}

\providecommand{\lft}{\overleftarrow}
\providecommand{\rgt}{\overrightarrow}
\providecommand{\dbl}{\overleftrightarrow}
\providecommand{\overline}{\overline}
\newcommand{\vac}{\mathrm{vac}}
\newcommand{\Ups}{\Upsilon}
\newcommand{\Upsr}{\rgt{\Ups}}
\newcommand{\Upsl}{\lft{\Ups}}
\newcommand{\Upsd}{\dbl{\Ups}}
\newcommand{\Omegad}{\dbl{\Omega}}
\newcommand{\Omegal}{\lft{\Omega}}
\newcommand{\Omegar}{\rgt{\Omega}}
\newcommand{\Qr}{\rgt{Q}}
\newcommand{\Ql}{\lft{Q}}
\newcommand{\Qd}{\dbl{Q}}

\newcommand{\Un}{\overleftrightarrow{\Upsilon}_{t_n,\dots,t_0}}
\newcommand{\stlink}{\star}
\newcommand{\hs}{\widehat{\sigma}}
\providecommand{\Sep}{\operatorname{Sep}}

\makeatletter
\newsavebox{\@brx}
\newcommand{\llangle}[1][]{\savebox{\@brx}{\(\m@th{#1\langle}\)}%
  \mathopen{\copy\@brx\kern-0.5\wd\@brx\usebox{\@brx}}}
\newcommand{\rrangle}[1][]{\savebox{\@brx}{\(\m@th{#1\rangle}\)}%
  \mathclose{\copy\@brx\kern-0.5\wd\@brx\usebox{\@brx}}}
\makeatother
\usepackage{fontawesome5} % icon package

\usepackage[misc]{ifsym} %%%icon package

\usepackage{tikz-feynman}

\usetikzlibrary{decorations.markings}

\usetikzlibrary{decorations.pathmorphing,decorations.markings}

\usetikzlibrary{patterns}

\allowdisplaybreaks

\usetikzlibrary{decorations.markings}

\tikzset{
    partial ellipse/.style args={#1:#2:#3}{
        insert path={+ (#1:#3) arc (#1:#2:#3)}
    }
}

\tikzset{->-/.style={decoration={
  markings,
  mark=at position #1 with {\arrow{>}}},postaction={decorate}}}
\tikzset{-<-/.style={decoration={
  markings,
  mark=at position #1 with {\arrow{<}}},postaction={decorate}}}

\begin{document}

\flushbottom

\title{Unified spatiotemporal quantum states and spatiotemporal entanglement from Kirkwood-Dirac phase space}

\author[a,\star]{Zhian Jia\orcidlink{0000-0001-8588-173X}}
\author[b,\dagger]{Mei-Hui Xiao\orcidlink{0009-0004-8706-6156}}

\affiliation[a]{Institute of Quantum Physics, School of Physics, Central South University,
Changsha 410012, China}

\affiliation[b]{School of Physics and Astronomy, Sun Yat-Sen University, Guangzhou 510275, China}

\affiliation[\star]{Email: \href{mailto:giannjia@foxmail.com}{giannjia@foxmail.com}}
\affiliation[\dagger]{Email: \href{mailto: xiaomh25@mail2.sysu.edu.cn}{xiaomh25@mail2.sysu.edu.cn}}

\abstract{
The notion of a spatiotemporal quantum state extends the conventional concept of a spatial quantum state to the spatiotemporal domain. Such states are represented by unit-trace operators that encode correlations among quantum events distributed across space and time. In this work, we use the spatiotemporal Kirkwood-Dirac phase space to provide a unified characterization of spatiotemporal quantum states. Spatiotemporal states obtained from Kirkwood-Dirac distributions are generally non-Hermitian and nonnormal; whereas those constructed from Margenau-Hill distributions are Hermitian. We provide a unification via (quasi)probabilistic mixture of Kirkwood-Dirac spatiotemporal states which encompasses almost all existing formulations of spatiotemporal quantum states. We derive recursive expressions for  spatiotemporal states and elucidate the relation between Kirkwood-Dirac nonclassicality and the temporality of spatiotemporal states. We further extend the construction to many-fold correlation functions and establish its connection with out-of-time-ordered correlators (OTOCs). We also develop a more general unifying framework based on (quasi)probabilistic mixture of $s$-parametrized spatiotemporal states and establish their Petz time reversal and application in studying Kubo-Martin-Schwinger (KMS) condition in two-time setting. Finally, we apply this framework to characterize quantum entanglement in spacetime and analyze spatiotemporal entanglement using several complementary entropy measures.
}

\keywords{Non-Equilibrium Field Theory, Quantum
Dissipative Systems, Correlation Functions, Thermal Field Theory}

\maketitle

\section{Introduction}
\label{sec:intro}

Space and time are typically treated differently in Copenhagen interpretation of quantum mechanics: whereas observables are associated with spatial coordinates, no analogous observables exist for time, which instead serves as a parameter governing the evolution of spatial states. This asymmetry leads to conceptual difficulties known as \emph{problem of time} when attempting to reconcile quantum mechanics with relativity, and becomes especially significant in quantum field theory and quantum gravity \cite{Page1983evolution,Rovelli1991time,Muga2008Time,Anderson2012time}.
One possible approach to this problem is to introduce a notion of quantum state associated with the temporal direction, which we refer to as a \emph{temporal state}, or, more generally, a \emph{spatiotemporal state}. One may also associate a Hilbert space with each time slice and introduce a tensor-product structure among different spacetime points
$\bigotimes_{(x,t)} \cH_{(x,t)}$,
\footnote{We emphasize that this statement is only schematic. In theories such as quantum field theory, the tensor-product structure is generally more subtle. In such cases, it is more appropriate to associate observable algebras with spacetime regions, rather than only with spatial regions.}
thereby allowing temporal degrees of freedom to be characterized in a manner analogous to spatial ones and placing space and time on a more equal footing.

Over the past few decades, various temporal state formalisms have been developed, motivated by quantum Bayesian inference, quantum causality, and related considerations. Notable examples include consistent histories \cite{griffiths1984consistent}; Feynman--Vernon influence functional \cite{Feynman1963influenceMatrix} and influence matrix \cite{Lerose2021influencematrix}; multiple-time states \cite{Watanabe1955twostateQM,aharonov2007twostatevectorformalismqauntum,Aharonov2009multi}; pseudo-density operators \cite{fitzsimons2015quantum} (also known as states over time \cite{horsman2017can,fullwood2022quantum,Parzygnat2023pdo}); quantum-classical games \cite{gutoski2007toward}, quantum combs \cite{Chiribella2009comb},  process tensors \cite{Pollock2018processtensor}, process matrices \cite{oreshkov2012quantum}; superdensity operators \cite{cotler2018superdensity}; and doubled density operators \cite{jia2024spatiotemporal,Jia2025TemporalKirkwoodDirac,jia2026temporaltomography}, 
spacetime quantum mechanics \cite{diaz2026unifyingspacetimeapproachesquantum}, among others. Understanding the connections and distinctions among these various approaches remains an important direction for future investigation \cite{liu2023unification,Parzygnat2023pdo,fullwood2022quantum,Jia2025TemporalKirkwoodDirac,diaz2026unifyingspacetimeapproachesquantum}.
In this work, by generalizing the results of Ref.~\cite{Jia2025TemporalKirkwoodDirac}, we develop a spatiotemporal Kirkwood--Dirac phase-space framework that provides a unified description of spatiotemporal quantum states and their quantum correlations.

The spatiotemporal Kirkwood--Dirac phase-space framework can be viewed as a generalized probabilistic theory (GPT). In purely spatial settings, GPTs provide a broad operational framework for characterizing quantum theory and for distinguishing quantum correlations from those allowed by more general physical theories\footnote{In quantum foundations, one seeks to reconstruct quantum mechanics from more fundamental physical principles, such as no-signaling. The GPT framework is well suited for studying quantum theory from an external perspective and for delineating its boundaries. It  naturally extends classical probability theory to encompass a wide class of probabilistic models and offers a unified, operationally motivated approach to states, measurements, and transformations in physical systems, without relying on specific mathematical structures such as Hilbert spaces.  Prominent examples include the Bell--Kochen--Specker theorem\cite{bell1964,kochen1967problem}, where the Tsirelson's bound is the boundary for quantum theory.}.  
It also has many applications in high energy physics, e.g., in black hole information problem, where some generalizations beyond quantum theory are possible \cite{Muller2012GPTblackhole,Braunstein2007QIblackhole}.
A central object in this framework is the quasiprobability distribution, which originates from Wigner's seminal work \cite{Wigner1932on}. Recently, quasiprobability methods have been extended to spatiotemporal settings through the spatiotemporal Kirkwood--Dirac quasiprobability distribution, which itself defines a spatiotemporal state in the sense of a GPT. It has been shown that this framework provides a promising route toward unifying different spatiotemporal state formalisms: by averaging over the spatiotemporal phase space and subsequently performing Bloch tomography, one can recover and relate a variety of existing spatiotemporal state constructions within a common framework~\cite{Jia2025TemporalKirkwoodDirac,jia2024spatiotemporal}. In this work, we establish a more general unification, summarized in Theorem~\ref{theorem:UnificationSingle}, based on a (quasi)probabilistic mixture (a similar idea is discussed in \cite{jia2023quantumspace} for interpreting pseudo-density operator as a (quasi)probabilistic mixture of product states). This construction goes beyond existing unification schemes and also yields new classes of spatiotemporal states. The interrelations among temporal states, Kirkwood--Dirac distributions, and Wightman correlation functions are summarized in Figure~\ref{fig:KDequv}.

\begin{figure}
    \centering
    \includegraphics[width=0.7\linewidth]{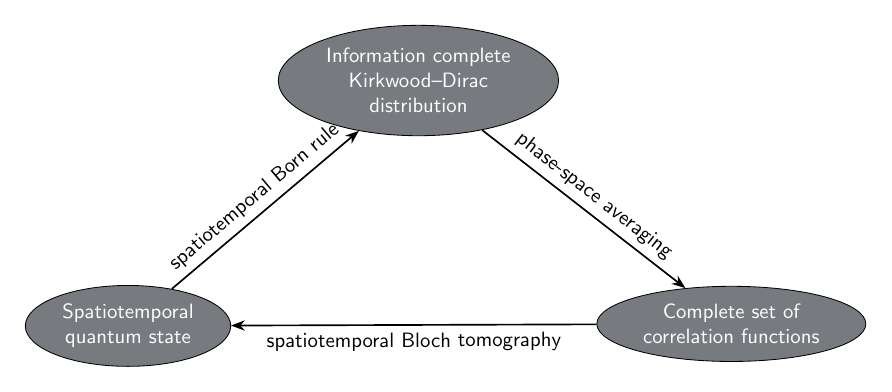}
    \caption{Equivalence between three core concepts for characterizing spatiotemporal quantum correlations: the spatiotemporal quantum state, the information-complete Kirkwood–Dirac distribution, and Wightman correlation functions. }
    \label{fig:KDequv}
\end{figure}

The Kirkwood--Dirac distribution \cite{Kirkwood1933quantum,Dirac1945on} was originally introduced as a complex-valued phase-space description of states $Q_{\rm KD}(a,b)=\langle \psi|a\rangle \langle a|b\rangle \langle b|\psi\rangle$. Like the Wigner function, it yields the correct probability marginals. Its real part is known as the Margenau--Hill quasiprobability distribution \cite{margenau1961correlation}. They have proven to be a useful tool in many contexts, see review \cite{ArvidssonShukur2024KDreview}. This framework was systematically extended to the (spatio)temporal setting in Ref.~\cite{Jia2025TemporalKirkwoodDirac}.
The spatiotemporal Wightman correlation function plays a key role in quantum field theory \cite{weinberg1995quantum}. A key observation is that all Wightman correlation functions are in fact phase-space averages over the spatiotemporal Kirkwood--Dirac distribution.

Roughly speaking, to construct a spatiotemporal state, we begin with a complete orthogonal set of observables $\{\mathcal{O}_{\mu}\}$. For a chosen collection of spacetime points, we consider the corresponding Wightman correlation functions
\begin{equation}
    T^{\mu,\nu,\cdots}
    =
    \big\langle
    \mathcal{O}_{\mu}(x,t)
    \mathcal{O}_{\nu}(y,t')
    \cdots
    \big\rangle.
\end{equation}
These correlation functions can then be assembled into a spatiotemporal state reads \cite{Jia2025TemporalKirkwoodDirac}
\begin{equation}
    \Upsilon
    =
    \frac{1}{\mathcal{Z}}
    \int d\mu\, d\nu \cdots\,
    T^{\mu,\nu,\cdots}\,
    \mathcal{O}_{\mu}(x,t)
    \otimes
    \mathcal{O}_{\nu}(y,t')
    \otimes \cdots,
\end{equation}
where $\mathcal{Z}$ is a normalization factor chosen such that $\Tr\Upsilon=1$. In finite-dimensional quantum systems, the integrals over the observable labels are replaced by discrete sums.
Notice that we need to assign a tensor product structure  $\otimes_{(x,t)}\mathcal{H}_{(x,t)}$, which in standard quantum theory only exists for spatial local degrees of freedom.
The spatiotemporal state is generally non-Hermitian; we therefore refer to it as a \emph{Kirkwood--Dirac spatiotemporal state} $\Ups$. Indeed, applying the spatiotemporal Born rule for a chosen set of projectors or positive operator-valued measures (POVMs) $F_{\alpha}(x,t), F_{\beta}(y,t')$, $\cdots$, yields a spatiotemporal Kirkwood--Dirac quasiprobability distribution $Q_{\rm KD}(\alpha,\beta,\cdots)$.
Hermitianizing the state gives the so-called Margenau--Hill spatiotemporal state
   $ \Upsilon^{\rm MH}=\frac{1}{2}(\Upsilon+\Upsilon^{\dagger})$,
for which the spatiotemporal Born rule yields the corresponding Margenau--Hill quasiprobability distribution, namely the real part of the Kirkwood--Dirac distribution, this reflects the fact  $\Tr  \Upsilon^{\rm MH}=1$.

The choice of correlation functions is therefore crucial. In this work, we focus primarily on finite-dimensional systems to illustrate the main ideas, while a systematic treatment of the infinite-dimensional case will be presented elsewhere. 
We will consider constructions with a single local space $\mathcal{H}_{(x,t)}=\mathcal{H}$, including left-, right-, and mixed-branch spatiotemporal states (the left and right bloom construction discussed in \cite{fullwood2023quantum,Lie2025stateovertime,lie2025probingquantumstatesspacetime} and the one defined in \cite{milekhin2025observable} are special cases of two-time left and right Kirkwood--Dirac temporal state), as well as constructions with a doubled local space $\mathcal{H}_{(x,t)}=\mathcal{H}^L\otimes\mathcal{H}^R$ \cite{jia2024spatiotemporal} (the doubled setting also underlies the process-matrix/process-tensor/quantum-comb formalisms~\cite{Chiribella2009comb,Pollock2018processtensor,gutoski2007toward}) and their many-fold generalizations. Hermitian spatiotemporal states are referred to as being of Margenau--Hill type, which includes the pseudo-density operator \cite{fitzsimons2015quantum} (we will call it L\"uders-von Neumann spatiotemporal state to emphasize the correlation function is averaging over L\"uders-von Neumann distributions) as a special case with a single local space. See Table~\ref{tab:spatiotemporal-state-summary} for a summary.
The spatiotemporal states satisfy the following properties: (i) $\Tr\Ups=1$; (ii) equal-time density operators are recovered as reduced states of the spatiotemporal state; and (iii) a spatiotemporal Born rule holds.
Since (quasi)probabilistic mixtures of spatiotemporal states also satisfy the conditions above, we use this observation to provide a unified description of several existing spatiotemporal-state frameworks.

\begin{table*}[t]
\centering
\caption{
Basic properties of the principal spatiotemporal-state constructions.
All states have unit trace, and every single-time marginal is an ordinary
positive density operator.
}
\label{tab:spatiotemporal-state-summary}
\renewcommand{\arraystretch}{1.45}
\setlength{\tabcolsep}{4pt}
\resizebox{\textwidth}{!}{%
\begin{tabular}{c|c|c|c|c|c}
\hline\hline
&
\textbf{Left}
&
\textbf{Right}
&
\textbf{Mixed-branch}
&
\textbf{Doubled}
&
\textbf{Many-fold}
\\
\hline

\shortstack{\textbf{Kirkwood--Dirac}}
&
\shortstack{
$\Upsl^{\rm KD}$\\
Fixed left ordering\\
Generally non-Hermitian\\
Generally nonnormal\\
$\Upsl^{\rm KD}=(\Upsr^{\rm KD})^\dagger$
}
&
\shortstack{
$\Upsr^{\rm KD}$\\
Fixed right ordering\\
Generally non-Hermitian\\
Generally nonnormal\\
$\Upsr^{\rm KD}=(\Upsl^{\rm KD})^\dagger$
}
&
\shortstack{
$\overline{\Upsilon}^{\rm KD}_{\ttB}$\\
Arbitrary branch assignment\\
Generally non-Hermitian, generally nonnormal\\
Adjoint branches occur in pairs
$\overline{\Upsilon}^{\rm KD}_{\ttB}=(\overline{\Upsilon}^{\rm KD}_{\ttB^c})^{\dagger}$
}
&
\shortstack{
$\Upsd^{\rm KD}$\\
Retains both contour branches\\
Generally non-Hermitian, generally nonnormal\\
All single-branch states are marginals
}
&
\shortstack{
$\Ups^{[2m],{\rm KD}}$\\
Defined on a $2m$-fold contour\\
Generally non-Hermitian, generally nonnormal\\
Encodes OTOCs and multi-fold\\
Wightman correlation functions
}
\\
\hline
\shortstack{\textbf{Margenau--Hill}}
&
\multicolumn{2}{c|}{
\shortstack{
$\displaystyle
 \Ups^{\rm MH}
 =\frac{1}{2}(\Upsl^{\rm KD}+\Upsr^{\rm KD})$\\
Hermitian and normal, generally nonpositive
}}
&
\shortstack{
$\overline{\Upsilon}^{\rm MH}_{\ttB}
=\frac{1}{2}
(\overline{\Upsilon}^{\rm KD}_{\ttB}
+\overline{\Upsilon}^{\rm KD}_{\ttB^c})$\\
Hermitianization of an adjoint\\
pair of mixed-branch KD states\\
Hermitian and normal, generally nonpositive
}
&
\shortstack{
$\displaystyle
 \Upsd^{\rm MH}
 =\frac{1}{2}(\Upsd^{\rm KD}
       +(\Upsd^{\rm KD})^\dagger)$\\
Hermitian and normal, generally nonpositive
\\
Produces real doubled correlations
}
&
\shortstack{
$\Ups^{[2m],{\rm MH}}
=\frac{1}{2}
(\Ups^{[2m],{\rm KD}}
+(\Ups^{[2m],{\rm KD}})^{\dagger})$\\
Hermitianization of
$\Ups^{[2m],{\rm KD}}$\\
Hermitian and normal, generally nonpositive
\\
Retains real multi-fold correlations
}
\\
\hline

\shortstack{\textbf{L\"uders--von Neumann}\\\textbf{(Pseudo-density operator)}}
&
\multicolumn{3}{c|}{
\shortstack{
$\Ups^{\rm LvN}$: Equal-weight probabilistic mixture of all mixed-branch KD states\\
Hermitian and normal, generally nonpositive\\
For two times:
$\Ups^{\rm LvN}=\Ups^{\rm MH}$
}}
&
\shortstack{
NA
}
&
\shortstack{
NA
}
\\
\hline\hline
\end{tabular}%
}
\end{table*}

Spatiotemporal states provide a general framework for understanding quantum entanglement and other correlations across spacetime. For spatially separated subsystems, the entanglement of a pure state is characterized by entanglement entropy, making entropy and mutual information fundamental quantities in the study of quantum correlations. 
Their spatiotemporal generalizations have been explored within several different formalisms, see e.g.~\cite{cotler2018superdensity,fullwood2026entropypseudodensitymatrix,Parzygnat2023SVD,Caputa2024SVD,jia2023quantumspace,Glorioso2024spacetimeMutual,milekhin2025observable,Doi2023Pseudoentropy,Doi2023Timelike,das2025temporalcorrelationschaosspacetime,das2026timelikeentanglementsignaturesergodicity}.
A further motivation comes from quantum gravity, where entropy plays a central role in the AdS/CFT correspondence. In particular, the Ryu--Takayanagi formula $S(\sA)=\operatorname{Area}(\Gamma_{\sA})/4G_N$ relates the entanglement entropy of a boundary conformal field theory to the area of a minimal surface in the bulk gravitational theory~\cite{Ryu2006}, providing a powerful geometric characterization of quantum entanglement. Related developments, such as pseudo-entropy~\cite{Doi2023Pseudoentropy}, extending entropic concepts to transition matrices, which originate from the two-state vector formalism~\cite{Watanabe1955twostateQM,aharonov2007twostatevectorformalismqauntum,Aharonov2009multi}; the resulting quantities are in general complex. This is the sense in which the timelike entanglement entropy of Ref.~\cite{Doi2023Timelike}, obtained by analytic continuation of the spacelike result and likewise complex-valued, can be interpreted as a pseudo-entropy. More recently, operator-level constructions have been proposed to give timelike entanglement entropy a microscopic meaning. Refs.~\cite{milekhin2025observable,Guo2026Spacetime}  (also see \cite{diaz2026unifyingspacetimeapproachesquantum}) introduce a spacetime density matrix for time-separated subsystems, which, as we will see, is two-time Kirkwood-Dirac state; while Ref.~\cite{Carignano:2024jxb} introduces reduced transition matrices arising in transverse tensor-network contractions. These reproduce the continued result in specific cases, but a single operator from which the timelike entanglement entropy follows in general has not yet been identified.  Our formalism provides an alternative framework that naturally extends to multi-time and general spatiotemporal settings. A deeper understanding of spatiotemporal entropy may therefore provide new insights into the interplay between quantum information, spacetime, and quantum gravity.

\subsection{Organization of the paper}

We use the temporal Kirkwood--Dirac quasiprobability distribution developed in
Ref.~\cite{Jia2025TemporalKirkwoodDirac} as the main tool of this work. We carefully show that: The Wightman correlation functions can be understood as quasiprobabilistic averages over temporal phase space.
Given an informationally complete set of Hermitian Hilbert--Schmidt operators, the complete Wightman correlation tensor determines a spatiotemporal state through a generalized Bloch formula. Conversely, the temporal Born rule recovers the corresponding Kirkwood--Dirac distribution. These three descriptions therefore encode the same information, while the equal-time restriction of the spatiotemporal state reduces to the ordinary density operator.

One of our principal technical results is a recursive construction of spatiotemporal states, including left, right, mixed branch and doubled ones. The left state is the adjoint of the right state, while mixed-branch states are obtained by changing the ordering of the temporal links or by taking appropriate marginals of the doubled state. The doubled construction coincides with the doubled density operator of Ref.~\cite{jia2024spatiotemporal}. The recursion preserves unit trace, reproduces all physical single-time marginals, and admits a matrix-product-operator representation. It also extends naturally to processes with memory through a system--environment dilation and
to many-fold contours, where a single spatiotemporal state encodes the corresponding OTOC and other many-fold correlators.

We quantify the departure of a spatiotemporal state from an ordinary density operator by its \emph{temporality}. For Hermitian states, it is the negative spectral weight and determines the trace-norm distance from the set of density operators. Many properties of the temporality are derived.  

For two-time states, we give a broader unification of existing temporal state constructions through a continuous family $\Ups^{(s)}$, whose endpoints are the left and right Kirkwood--Dirac states and whose symmetric point contains the Leifer--Spekkens construction. We establish an $s$-dependent Petz time-reversal relation for this family; at $s=1/2$ it reduces to the standard Petz recovery map. We also apply this to the thermal states to introduce the KMS family of temporal states and discuss their properties.

Finally, we develop complementary notions of spatiotemporal entanglement and entropy. Vectorizing a temporal operator produces a pure state on a doubled operator Hilbert space and leads to temporal Schmidt spectra, entanglement entropies, mutual information, and finite-memory area laws. For Hermitian Margenau--Hill spatiotemporal states, we define spatiotemporal separability and a corresponding robustness measure. We also discuss the spectral property of spatiotemporal state and show $\|\Ups\|_\infty\leq1$, so its eigenvalues lie in the closed unit disk. This allows us to distinguish spectral, spectral-magnitude, SVD, and normalized SVD entropies, and to separate temporality into spectral and nonnormal contributions. For Margenau--Hill states the spectral-magnitude and SVD hierarchies coincide, whereas they are inequivalent for general Kirkwood--Dirac states. We further show that a genuine nonnegative SVD mutual information is obtained by normalizing the global magnitude operator before taking its marginals.

The paper is organized as follows.
Section~\ref{sec:WightmanCor} reviews Wightman correlators and the
Schwinger--Keldysh formalism.
Section~\ref{sec:TQD} introduces temporal Kirkwood--Dirac distributions and their phase-space averaging.
Section~\ref{sec:tempstate} develops the spatiotemporal-state recursion, many-fold extensions, and temporality.
Section~\ref{sec:TwoTimeState} discusses two-time unification, Petz time reversal, and KMS symmetry.
Sections~\ref{sec:STentropy} and~\ref{sec:entropy} develop the vectorized, spectral, and SVD approaches to spatiotemporal entanglement and entropy.

\section{Preliminaries on Wightman correlation functions for open-system quantum field theory}
\label{sec:WightmanCor}

\subsection{Multi-time Wightman correlation functions}
\label{sec:keldysh}
In zero-temperature closed system quantum field theory, one commonly considers Wightman correlation functions of the form
\begin{equation}
    \Delta(x,y,\ldots,z)
    =
    \langle O_A(x) O_B(y) \cdots O_C(z) \rangle
    =
    \langle \mathrm{vac} |
    O_A(x) O_B(y) \cdots O_C(z)
    | \mathrm{vac} \rangle .
\end{equation}
More generally, for finite-temperature or dissipative quantum system, a state described by a density operator $\varrho$, one has the so-called \emph{left} Wightman correlation functions\footnote{Here, left and right refer to placing operators to the left and right of the density operator, respectively. The distinction among left, right, mixed-branch, doubled, and related constructions will be useful in the subsequent discussion.}
\begin{equation}
    \Delta(x,y,\ldots,z)
    =
    \operatorname{Tr}
    \left[
         O_A(x) O_B(y) \cdots O_C(z)\varrho
    \right].
\end{equation}
We can also introduce a doubled Wightman correlation functions (in many situations, we will also assume that the two sides are of equal time, namely, $t_x=t_{x'}$,$t_y=t_{y'}$, so on and so forth),
\begin{equation}
    \Delta(x,y,\ldots,z;x',y',\ldots,z')
    =
    \operatorname{Tr}
    \left[
        O_A(x) O_B(y) \cdots O_C(z)
        \varrho\,
        O'_A(x') O'_B(y') \cdots O'_C(z')
    \right].
\end{equation}
Note also that \emph{mixed-branch} Wightman correlation functions can be obtained from the doubled one by setting some observables on either the left or right branch to the identity operator; for example, $\Tr[\cdots O_{A}(t_3) O_{A}(t_1) \varrho O_{B}(t_0) O_{B}(t_2) \cdots ]$.

These correlation functions arise naturally in the description of quantum instruments and sequential measurements. At first glance, cyclicity of the trace seems to imply that there is little difference between the left/right/mixed-branch and doubled Wightman correlation functions. However, as we shall see, once evolutions are inserted between observables, the doubled version becomes the more essential one; to make every version a right one, one needs to perform some postprocessing.

To illustrate this point, we consider the following generalization of observables called quantum instrument \cite{davies1970operational}.
    A quantum instrument is a collection of completely positive, trace-non-increasing (CPNIP) maps
    $\{\mathcal{F}_{\alpha}\}_{\alpha}$ such that
$\sum_{\alpha} \mathcal{F}_{\alpha}$
    is a completely positive, trace-preserving (CPTP) map. For an input state $\varrho$, the probability of outcome $\alpha$ is
    \begin{equation}
        p(\alpha)
        =
        \operatorname{Tr}\left[\mathcal{F}_{\alpha}(\varrho)\right].
    \end{equation}
    Each map admits a Kraus decomposition,
    \begin{equation}
        \mathcal{F}_{\alpha}(\varrho)
        =
        \sum_k K_{\alpha,k} \varrho K_{\alpha,k}^{\dagger},
    \end{equation}
    and hence
    \begin{equation}
        p(\alpha)
        =
        \sum_k
        \operatorname{Tr}
        \left[
            K_{\alpha,k} \varrho K_{\alpha,k}^{\dagger}
        \right]
        =
        \sum_k
        \operatorname{Tr}
        \left[
             K_{\alpha,k}^{\dagger} K_{\alpha,k}\varrho
        \right].  
    \end{equation}
    For one-time case, the quantum instrument can always be transformed into the POVM $F_{\alpha}=\sum_k K_{\alpha,k}^{\dagger} K_{\alpha,k}$, however, as we will see, for multi-time case, this is not the case, we must\footnote{There is no natural way to transform multi-time quantum instruments to POVMs.} consider a doubled Kraus expression.
    The normal observable is just von Neumann projective quantum instrument with a chosen set of outputs, with $K_{\alpha}=\Pi_{\alpha}$ for projectors $\Pi_{\alpha}$ arising from $O_A=\sum_{\alpha}  \alpha \Pi_{\alpha}$.

If the Kraus operators can be expanded in a chosen operator basis $\sigma_{\mu}$ \footnote{For example, one may consider the Hilbert-Schmidt basis for the operator space $\mathbf{B}(\mathcal{H})$ equipped with the inner product $\langle J,K\rangle= \Tr [J^{\dagger}K]$. In this work, we always employ a Hermitian operator basis; such a basis exists (we restrict ourselves to Hilbert spaces with these nice properties), as one can Hermitianize any given set of basis elements by replacing $K_i$ with $\frac{K_i+K_{i}^{\dagger}}{2}$. For finite-dimensional spaces, we choose this basis so that $\sigma_{0}=\mathds{1}$ and $\Tr \sigma_{\mu}\sigma_{\nu}= \mathrm{dim}\mathcal{H}\, \delta_{\mu,\nu}$. 
For $\dim\cH=2$, they reduce to the Pauli operators, while for $\dim\cH=3$, they are the Gell--Mann matrices~\cite{Gell-mann1962symmetry}. For explicit matrix for for arbitrary dimensions, see e.g., \cite[Section~3]{wei2022antilinear} and reference therein.}, namely, we have $\Tr [ \sigma_{\mu}\varrho \sigma_{\nu}]$, the probabilities of instrument outcomes can be reconstructed from the corresponding Wightman correlation functions, since $ \operatorname{Tr}
        \left[
            K_{\alpha,k} \varrho K_{\alpha,k}^{\dagger}
        \right]=\sum_{\mu,\nu} c_{\alpha,k;\mu} c^*_{\alpha,k;\nu}\Tr [ \sigma_{\mu}\varrho \sigma_{\nu}]$ where $K_{\alpha,k}=\sum_{\mu}c_{\alpha,k;\mu}\sigma_{\mu}$.

We first consider sequential projective measurements. Let
\begin{equation}
    O_A = \sum_a a\, \Pi_a^A,
    \qquad
    O_B = \sum_b b\, \Pi_b^B,
\end{equation}
and suppose that $O_A$ is measured first, followed by $O_B$. The joint probability (von Neumann probability) of obtaining outcomes $a$ and $b$ is given by the L\"uders rule:
\begin{equation}
    p_{\rm vN}(a,b)
    =
    \operatorname{Tr}
    \left[
        \Pi_b^B \Pi_a^A
        \varrho
        \Pi_a^A \Pi_b^B
    \right].
\end{equation}
The corresponding sequential-measurement correlation function is
\begin{equation}
    \langle O_AO_B\rangle_{\rm vN}=\mathbb{E}_{\rm vN}(a,b)
    =
    \sum_{a,b} ab\, p_{\rm vN}(a,b).
\end{equation}
Using cyclicity of the trace and $(\Pi_b^B)^2=\Pi_b^B$, this can be rewritten as
\begin{equation}
    \langle O_AO_B\rangle_{\rm vN}
    =
    \sum_a a\,
    \operatorname{Tr}
    \left[
        \Pi_a^A O_B \Pi_a^A\varrho
    \right].
\end{equation}
This quantity is generally not determined solely by the three-point right Wightman correlation function
 $\operatorname{Tr}(\varrho \sigma_{\mu} \sigma_{\nu} \sigma_{\beta})$. 
For sequential quantum instruments, the joint probability is given by
\begin{equation}
    p(\alpha,\beta,\ldots,\mu)
    =
    \Tr\!\left[
    \cM_{\mu}\circ\cdots\circ\cK_{\beta}\circ\cF_{\alpha}(\varrho)
    \right].
\end{equation}
The corresponding correlation functions are then obtained by averaging over the measurement outcomes.

Note that sequential-measurement correlation function is not the same as Wightman correlation functions.
Our starting point is that the standard Wightman correlation function is in fact an averaging over the Kirkwood--Dirac phase-space distribution \cite{Kirkwood1933quantum,Dirac1945on}
\begin{equation}
    Q_{\rm KD}(a,b)=\Tr \left[ \Pi_a\Pi_b \varrho\right],
\end{equation}
namely
\begin{equation}
    \langle O_AO_B\rangle_{\rm KD} =\sum_{a,b}ab\, Q_{\rm KD}(a,b).
\end{equation}
This will serve as a crucial starting point for our consideration of the temporal phase-space Kirkwood--Dirac distribution and its applications in constructing spatiotemporal state and for calculating the Wightman correlation functions.

Notice that for commutative measurements (i.e., those that are jointly measurable), the expression simplifies to
\begin{equation}
    \langle O_AO_B\rangle_{\rm vN}
    =
    \operatorname{Tr}
    \left[
        \varrho\,
        O_A O_B
    \right]=:\langle O_AO_B\rangle_{\rm KD},
\end{equation}
since the observables $O_A$ and $O_B$ share a common set of eigenstates. In general, however, they are distinct, and Wightman correlation functions can be complex-valued.

We now consider the case with time steps \(t_0<\cdots<t_n\). The system is prepared in an initial state \(\varrho_{t_0}\), and the evolution between successive times is described by quantum channels \(\mathcal{E}_{t_i,t_{i+1}}\), corresponding to a memoryless multi-time quantum process. Since the general case is considerably more involved, we postpone its discussion to the next section and begin with the simpler setting of unitary evolution \(U_{t_0,t}= \exp (-i\int_{t_0}^t H dt)\) and a pure initial state \(|\mathrm{vac}\rangle\), which is the standard setup in quantum field theory.

In textbook quantum field theory, observables \(O_{A_k}(t_k)\) are defined in the Heisenberg picture as \(O_{A_n}(t_k)=U_{t_0,t_k}^{\dagger}O_{A_k}U_{t_0,t_k}\), and we have
\begin{equation}
\begin{aligned}
   &\langle\vac| O_{A_n}(t_n) \cdots O_{A_0}(t_0) |\vac\rangle\\
   =&\langle\vac| U_{t_0,t_n}^{\dagger} O_{A_n} U_{t_{n-1},t_{n}}O_{A_{n-1}} U_{t_{n-2},t_{n-1}} \cdots U_{t_1,t_2}O_{A_1}  U_{t_0,t_1} O_{A_0} |\vac\rangle,
\end{aligned}
\end{equation}
where for intermediate unitary evolution, we have apply $U_{t_0,t_{k}} U_{t_0,t_{k-1}}^{\dagger}=U_{t_{k-1},t_k}$.
For a mixed state \(\varrho_{t_0}\), we have
\begin{equation}
\begin{aligned}
   &\Tr \left[ O_{A_n}(t_n) \cdots O_{A_0}(t_0) \varrho_{t_0} \right]\\
   =&\Tr \left[ U_{t_0,t_n}^{\dagger} O_{A_n} U_{t_{n-1},t_{n}}O_{A_{n-1}} U_{t_{n-2},t_{n-1}} \cdots U_{t_1,t_2}O_{A_1}  U_{t_0,t_1} O_{A_0}
   \varrho_{t_0}\right].
\end{aligned}
\end{equation}
This can also be interpreted as an averaging over the \emph{temporal} Kirkwood--Dirac phase space.

Consider the complex-valued quasiprobability distribution,
\begin{equation}
    Q_{\rm KD}(a_n,\cdots,a_0)=\Tr \left[ \Pi_{a_n} U_{t_{n-1},t_n}{\Pi}_{a_{n-1}} U_{t_{n-2},t_{n-1}} \cdots U_{t_1,t_2}{\Pi}_{a_1} U_{t_0,t_1} \Pi_{a_0} \varrho_{t_0}  U_{t_0,t_n}^{\dagger} \right].
\end{equation}
Using the decomposition \(U_{t_0,t_n}^{\dagger}=U_{t_0,t_1}^{\dagger}\cdots U_{t_{n-1},t_n}^{\dagger}\), the above expression can be recognized as a left Kirkwood--Dirac phase-space distribution \cite{Jia2025TemporalKirkwoodDirac,jia2026temporaltomography}, since the identity operator is inserted on the right-hand side of the density operator at each time step. The evolution is now implemented by \(\mathcal{E}_{t_k,t_{k+1}}=U_{t_k,t_{k+1}}(\bullet)U_{t_k,t_{k+1}}^{\dagger}\), and at time \(t_k\) the phase-space operation takes the form \(\mathcal{P}_{a_k}=\Pi_{a_k}(\bullet)\mathds{1}\). We then have
\begin{equation}
\begin{aligned}
   \Tr \left[ O_{A_n}(t_n) \cdots O_{A_0}(t_0) \varrho_{t_0} \right]=\mathbb{E}_{\rm KD}(a_n,\cdots,a_0)
   =\sum_{a_n,\cdots,a_0} a_n\cdots a_0\, Q_{\rm KD}(a_n,\cdots,a_0).
\end{aligned}
\end{equation}
This naturally motivates the use of the Kirkwood--Dirac phase-space distribution as a tool for characterizing temporal quantum correlations. A similar analysis applies to the right and doubled Wightman correlation functions.

\begin{remark}
    The Kirkwood--Dirac phase-space distributions are complex-valued; their real parts define the real-valued Margenau--Hill quasiprobability distribution \cite{margenau1961correlation}, while the imaginary part sums to zero. Since the Wightman correlation function is obtained by averaging over the Kirkwood--Dirac phase-space distributions, it can likewise be decomposed into real and imaginary contributions:
    \begin{equation}
    \begin{aligned}
         \Tr \left[ O_{A_n}(t_n) \cdots O_{A_0}(t_0) \varrho_{t_0} \right]=&\mathbb{E}_{\rm KD}(a_n,\cdots,a_0)\\
         =&\mathbb{E}_{\rm MH}(a_n,\cdots,a_0)+i\mathbb{E}_{\rm Im-KD}(a_n,\cdots,a_0).
    \end{aligned}
    \end{equation}
    The imaginary part of the Kirkwood--Dirac phase-space distribution plays a key role in characterizing nonclassicality of the system \cite{ArvidssonShukur2024KDreview}.
\end{remark}

%=====================================================================
\subsection{Schwinger--Keldysh formalism}
\label{subsec:keldysh}
Let us briefly review the Schwinger--Keldysh formalism for non-equilibrium quantum field theories \cite{kamenev2023field,sieberer2016keldysh}, in a form suitable for the phase-space constructions that follow. The two-time correlation functions have a structure closely analogous to the averages taken over Kirkwood--Dirac distributions.

For a system prepared in $\varrho_{t_0}$ and evolving under a (possibly time-dependent, possibly dissipative) generator, expectation values of Heisenberg operators are computed on the closed time path $\mathcal{C}=\mathcal{C}_+\cup\mathcal{C}_-$, running forward from $t_0$ to $t_f$ along $\mathcal{C}_+$ and backward along $\mathcal{C}_-$. Fields acquire a contour index, $\phi_\pm$, and the generating functional is
\begin{equation}
  Z[j_+,j_-]=\Tr\!\left[\;\mathcal{T}\,e^{-i\int (H-j_+\phi)dt}\;\varrho_{t_0}\;
  \bar{\mathcal{T}}\,e^{+i\int (H-j_-\phi)dt}\right],
\end{equation}
with $\mathcal{T}$ ($\bar{\mathcal{T}}$) denoting time (anti-time) ordering. The two contour branches correspond precisely to the two sides on which operators can be inserted relative to the density operator: $\phi_+$ multiplies $\varrho$ from the left, while $\phi_-$ multiplies it from the right. This is the field-theoretic origin of the \emph{left/right/doubled/mixed branch} distinction that organizes the entire paper.

The left and right Wightman correlation functions are
$G^{>}(t,t')=-i\langle\phi(t)\phi(t')\rangle$ and
$G^{<}(t,t')=-i\langle\phi(t')\phi(t)\rangle$.
Consider the Keldysh rotation
\begin{equation}
  \phi_{\rm cl}=\tfrac{1}{2}(\phi_++\phi_-),\qquad
  \phi_{\rm q}=\phi_+-\phi_-,
\end{equation}
where $\phi_{\rm cl}$ is the classical field that survives in the classical limit, while $\phi_{\rm q}$ is the quantum field, which vanishes in the classical limit. Computing the correlation functions among them yields
\begin{equation}
  G^{K}(t,t')=-i\langle\{\phi(t),\phi(t')\}\rangle,\qquad
  G^{R/A}(t,t')=\mp i\theta(\pm(t-t'))\langle[\phi(t),\phi(t')]\rangle,
\end{equation}
where $G^K$ is the Keldysh (statistical) component and $G^{R/A}$ are the retarded/advanced (spectral) components.
Note that $G^K=G^{>}+G^{<}$  and $G^R-G^A=G^{>}-G^{<}$, thus
\begin{equation}\label{eq:keldysh-decomp}
  G^{>}=\frac{1}{2}\left(G^{K}+ G^{R}-G^{A}\right), \quad  G^{<}=\frac{1}{2}\left(G^{K}- G^{R}+G^{A}\right).
\end{equation}
Equation~\eqref{eq:keldysh-decomp} is the field-theoretic avatar of the Kirkwood--Dirac/Margenau--Hill decomposition used above: the real (statistical) part of the temporal Kirkwood--Dirac averaging corresponds to the Keldysh function, while its imaginary (response) part corresponds to the spectral function.

In thermal equilibrium, they are related by the Kubo--Martin--Schwinger (KMS) condition (in frequency space) $G^{>}(\omega)=e^{\beta\omega}G^{<}(\omega)$, or equivalently by the fluctuation-dissipation theorem
$G^K(\omega)=\coth\left(\frac{\beta\omega}{2}\right)\bigl(G^R(\omega)-G^A(\omega)\bigr)
$. In Section~\ref{subsec:KMS}, we show that KMS symmetry selects a distinguished member of the Kirkwood--Dirac family, namely, the one that is covariant under Petz time reversal.

Finally, an open system in the Keldysh formalism is obtained by integrating out the environment for a unitary system-environment evolution. When the initial system-environment state is not of product form, the reduced dynamics is no longer a CPTP map which makes it difficult to deal with in the traditional framework; nonetheless, the temporal Kirkwood--Dirac quasiprobability formalism remains applicable to this scenario. We treat this case explicitly in Section~\ref{subsec:withMemo}.

\section{Temporal Kirkwood--Dirac quasiprobability distributions and temporal phase space averaging}
\label{sec:TQD}

Let us now turn to the general framework of temporal quasiprobability distributions (we will focus on temporal setting, generalization to spatiotemporal setting is straightforward) for closed- and open-system quantum field theory. In fact, we will consider an even more intricate scenario where the system (S) and environment (E) may be initially entangled and then evolve under a joint unitary evolution $U^{SE}_{0,t}$, after which we trace out the environment. This goes beyond the standard quantum channel description, since obtaining a quantum channel from the joint evolution $U^{SE}_{0,t}$ requires the initial state to be of product form \cite{Nielsen2010}. We thus distinguish between memoryless and memory cases in open quantum field theory. Closed-system quantum field theory, of course, always falls into the memoryless case.
We will consider a discrete series of time slices, $t_0<\cdots<t_n$, introducing a phase space $\alpha_k$ for each time step $t_k$. The temporal quasiprobability distributions are thus a quasiprobability distribution for the trajectories in temporal phase space.

\subsection{Temporal Kirkwood--Dirac quasiprobability distributions}

\subsubsection{Without memory}
 
We first consider the memoryless case. For a quantum multi-time process with initial state $\varrho_{t_0}$ and evolution CPTP maps $\mathcal{E}_{t_i,t_{i+1}}$ (for closed-system quantum field theory, $\mathcal{E}_{t_k,t_{k+1}}=U_{t_k,t_{k+1}}(\bullet)U_{t_k,t_{k+1}}^{\dagger}$ and $\varrho_{t_0}=|\vac\rangle \langle \vac|$), we define the Kirkwood--Dirac TQD as
\begin{equation}
Q_{\rm KD}(\alpha_n, \dots, \alpha_0)
=
\operatorname{Tr} \Big[ \mathcal{P}^{t_n}_{\alpha_n} \circ \mathcal{E}_{t_{n-1},t_n} \circ \cdots \circ \mathcal{P}^{t_1}_{\alpha_1} \circ \mathcal{E}_{t_0,t_1} \circ \mathcal{P}^{t_0}_{\alpha_0} (\varrho_{t_0}) \Big],
\label{eq:temporal_quasiprobability}
\end{equation}
where $\alpha_k$ denotes a phase-space point at time $t_k$, and $\mathcal{P}^{t_k}_{\alpha_k}$ is the corresponding phase-space operation, which is generally a linear map but not completely positive or trace-preserving. Typical choices~\cite{Jia2025TemporalKirkwoodDirac} include the right projection $\mathcal{P}^{t_k}_{\beta_k}(\bullet) = \bullet\, \Pi_{\beta_k}$, the left projection $\mathcal{P}^{t_k}_{\alpha_k}(\bullet) = \Pi_{\alpha_k}\, \bullet$, and the doubled projection $\mathcal{P}^{t_k}_{\alpha_k,\beta_k}(\bullet) = \Pi_{\alpha_k}\, \bullet\, \Pi_{\beta_k}$, where $\{\ket{\alpha_k}\}$ and $\{\ket{\beta_k}\}$ are orthonormal bases of the Hilbert space.

To render the temporal Kirkwood--Dirac quasiprobability distribution informationally complete, these projectors can be generalized to an informationally complete POVM (IC-POVM) $\{F_{\alpha}\}$ (the collection of IC-POVMs in different time steps is called spatiotemporal quantum frame). We refer to the resulting constructions as the left, right, and doubled informationally complete temporal Kirkwood--Dirac quasiprobability distributions, corresponding to the left, right, and doubled insertions
\begin{equation}
    \mathcal{P}_{\alpha}=F_{\alpha}(\bullet)\mathds{1},
    \qquad
    \mathcal{P}_{\beta}=\mathds{1}(\bullet)F'_{\beta},
    \qquad
    \mathcal{P}_{\alpha,\beta}=F_{\alpha}(\bullet)F'_{\beta},
\end{equation}
respectively (also notice projectors are special cases of POVM). We denote these distributions by
\begin{equation}
    \Ql_{\rm KD}(\alpha_n,\cdots,\alpha_0),
    \qquad
    \Qr_{\rm KD}(\beta_n,\cdots,\beta_0),\qquad
    \Qd_{\rm KD}
    (\alpha_n,\cdots,\alpha_0;\beta_n,\cdots,\beta_0),
\end{equation}
where $\Ql_{\rm KD}$, $\Qr_{\rm KD}$, and $\Qd_{\rm KD}$ correspond to the left, right, and doubled constructions, respectively. Explicitly, for IC-POVM insertions,
\begin{align}
  \Ql_{\rm KD}(\alpha_n,\dots,\alpha_0)&=\Tr\!\left[F_{\alpha_n}\cE_{t_{n-1},t_n}\!\left(\cdots F_{\alpha_1}\cE_{t_0,t_1}(F_{\alpha_0}\varrho_{t_0})\cdots\right)\right],\label{eq:QL}\\
  \Qr_{\rm KD}(\beta_n,\dots,\beta_0)&=\Tr\!\left[\cE_{t_{n-1},t_n}\!\left(\cdots \cE_{t_0,t_1}(\varrho_{t_0}F'_{\beta_0})F'_{\beta_1}\cdots\right)F'_{\beta_n}\right],\label{eq:QR}\\
  \Qd_{\rm KD}(\alpha_n,\dots,\alpha_0;\beta_n,\dots,\beta_0)&=
  \Tr\left[F_{\alpha_n}\cE_{t_{n-1},t_n}\left(\cdots F_{\alpha_1}\cE_{t_0,t_1}(F_{\alpha_0}\varrho_{t_0}F'_{\beta_0})F'_{\beta_1}\cdots\right)F'_{\beta_n}\right].\label{eq:QD}
\end{align}
We may also choose a mixture of left and right temporal slots by introducing a branch, for example,
$\ttB=(\cdots,t_1^R,t_0^L)$, which gives rise to a mixed-branch distribution
$\overline{Q}^{\ttB}_{\rm KD}(\cdots,\alpha_1,\beta_0)$.
Each branch $\ttB$ has a dual branch $\ttB'$, obtained by replacing every temporal slot with its opposite chirality, namely $L\leftrightarrow R$.
Some elementary but important consistency relations follow immediately and will be used repeatedly:
\begin{align}\label{eq:marginals}
  \sum_{\beta_n,\cdots, \beta_0}\Qd_{\rm KD}(\alpha;\beta)=\Ql_{\rm KD}(\alpha),\,\,
  \sum_{\alpha_n,\cdots ,\alpha_0}\Qd_{\rm KD}(\alpha;\beta)=\Qr_{\rm KD}(\beta),\,\,\\
  \sum_{\ttB'} \Qd_{\rm KD} =\overline{Q}^{\ttB}_{\rm KD},\quad 
  \Ql_{\rm KD}(\alpha)^{*}=\Qr_{\rm KD}(\alpha), \quad \overline{Q}^{\ttB}_{\rm KD}=(\overline{Q}^{\ttB'}_{\rm KD})^*
\end{align}
and, when the two branches use the same projectors, the diagonal of the doubled distribution is the L\"uders--von Neumann distribution,
\begin{equation}\label{eq:vNdiag}
  Q_{\rm LvN}(\alpha_n,\dots,\alpha_0)=\Qd_{\rm KD}(\alpha_n,\dots,\alpha_0;\alpha_n,\dots,\alpha_0)\;\ge 0 .
\end{equation}
Equation~\eqref{eq:vNdiag} is the precise sense in which the doubled distribution is the fundamental object: the operationally realizable (sequential L\"uders) statistics is its diagonal, whereas the left and right distributions are its marginals. The off-diagonal entries $\alpha\neq \beta$ carry the coherences between measurement branches that are destroyed by an actual sequential measurement, and are accessible only interferometrically or via quantum snapshotting \cite{Jia2025TemporalKirkwoodDirac,jia2026temporaltomography}.

These temporal Kirkwood--Dirac distributions  are generally complex-valued quasiprobability distribution and are therefore referred to as temporal Kirkwood--Dirac distributions \cite{Kirkwood1933quantum,Dirac1945on,Jia2025TemporalKirkwoodDirac}. 
They can be regarded as quasiprobabilistic temporal states in the generalized probabilistic theory , equivalent to the temporal state in operator form~\cite{Jia2025TemporalKirkwoodDirac}.

\subsubsection{With memory}
\label{subsec:withMemo}
Consider a quantum system $S$ coupled to an environment $E$, with the system and environment initially prepared in a joint state $\varrho_{SE}$, which may or may not be a product state. We consider the scenario illustrated in Figure~\ref{fig:QuantumComb}, where the joint system and environment undergo a unitary evolution $U^{SE}_{t_0,t}$. More generally, one may consider a CPTP evolution $\mathcal{E}_{t_0,t}$. By a multi-time version of the Stinespring dilation theorem, any CPTP evolution can be represented as a unitary evolution on an enlarged system. Thus, in many situations, it suffices to consider the unitary channel
\begin{equation}
\mathcal{U}_{t_0,t}(\bullet)
=
U^{SE}_{t_0,t}\,\bullet\,\big(U^{SE}_{t_0,t}\big)^\dagger.
\end{equation}
At each of the $n+1$ discrete time steps $t_0,\ldots,t_n$, we leave the system's slot open; after the final time step, we trace out the environment. The resulting object is commonly referred to as a process tensor or quantum comb \cite{Chiribella2009comb,Milz2021quantum}.

\begin{figure}
    \centering
    \includegraphics[width=0.8\linewidth]{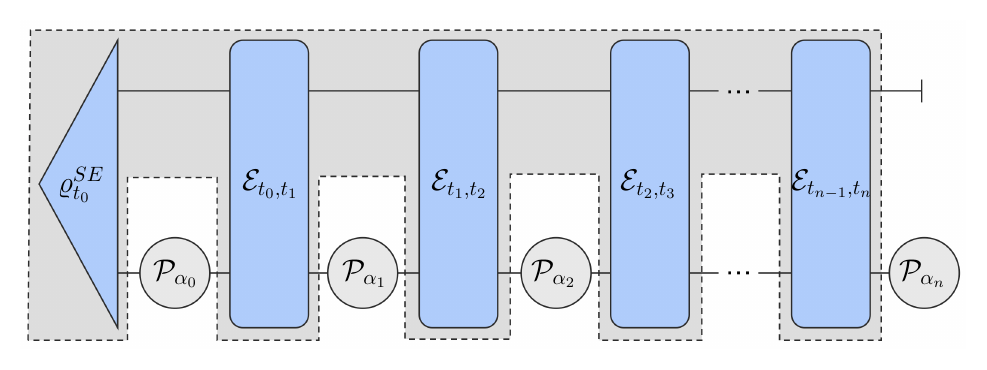}
    \caption{The general multi-time quantum process with memory, which can be represented as a quantum comb.}
    \label{fig:QuantumComb}
\end{figure}

At each time $t_k$, we insert a phase-space operation $\mathcal{P}_{\alpha_k}^{t_k}$ and then apply the unitary evolution $U^{SE}_{t_k,t_{k+1}}$ to the joint system and environment.
The resulting temporal Kirkwood--Dirac distribution is
\begin{equation}
Q_{\rm KD}(\alpha_n, \dots, \alpha_0)
=
\operatorname{Tr}_{SE} \Big[ \mathcal{P}^{t_n}_{\alpha_n} \circ \mathcal{U}_{t_{n-1},t_n} \circ \cdots \circ \mathcal{P}^{t_1}_{\alpha_1} \circ \mathcal{U}_{t_0,t_1} \circ \mathcal{P}^{t_0}_{\alpha_0} (\varrho^{SE}_{t_0}) \Big],
\label{eq:temporal_quasiprobability_memory}
\end{equation}
where the phase-space operations act as $\mathcal{P}\otimes \Id^E$ on the joint space. When the inserted operations are chosen from an informationally complete set, the resulting multi-time quasiprobability distribution contains the full information about the underlying quantum process, including its temporal correlations and memory effects.

%%%%%%%%%

\subsubsection{Temporal Margenau--Hill quasiprobability distributions}

The real part of the temporal Kirkwood--Dirac distribution defines a real-valued quasiprobability distribution, known as the temporal Margenau--Hill distribution \cite{margenau1961correlation,Jia2025TemporalKirkwoodDirac,jia2026temporaltomography}. The imaginary part of the temporal Kirkwood--Dirac distribution, in contrast, is not a quasiprobability distribution, since its sum over all outcomes vanishes. Nevertheless, we refer to it as a ``distribution'' and consider its average over the corresponding outcomes.

Several properties make this decomposition particularly useful. First, the left and right distributions are complex conjugates of each other, as shown in Eq.~\eqref{eq:marginals}. Hence,
\begin{equation}
  Q_{\MH}
  =
  \tfrac12\left(\Ql_{\rm KD}+\Qr_{\rm KD}\right)
  =
  \mathrm{Re}\,\Ql_{\rm KD},
\end{equation}
so Hermitianization at the level of spatiotemporal states corresponds precisely to taking the real part at the level of distributions. Second, the spatiotemporal state associated with the Margenau--Hill distribution is Hermitian, which makes its analysis more tractable. Moreover, as we shall see, the corresponding correlation functions admit an analogous decomposition.

%---------------------------------------------------------------------
\subsection{Wightman correlation functions as averaging over temporal phase space}

As discussed in Section~\ref{sec:WightmanCor}, for closed-system quantum field theory, the Wightman correlation function can be regarded as an average over the Kirkwood--Dirac distribution. We naturally generalize this construction to the open-system setting by defining
\begin{equation}\label{eq:KDaveWightman}
    \begin{aligned}
    \left\langle \{O_{A_n}(t_n), \cdots, O_{A_0}(t_0)\} \right\rangle_{W}
    &= \mathbb{E}_{\rm KD}(a_n,\cdots,a_0) \\
    &= \mathbb{E}_{\rm MH}(a_n,\cdots,a_0)
    + i\mathbb{E}_{\rm Im\text{-}KD}(a_n,\cdots,a_0).
    \end{aligned}
\end{equation}
This can be understood as inserting observables into the slots of the quantum comb. More generally, by inserting observables $O_{A_k}$ and $O_{B_k}$ on the left and right of each slot, respectively, we obtain
\begin{equation}
\begin{aligned}
\operatorname{Tr}_{SE} \Big[
    &(O_{A_n}\otimes \mathds{1}^E)
    \mathcal{U}_{t_{n-1},t_n}
    \Big(
        \cdots
        \mathcal{U}_{t_0,t_1}
        \Big(
            (O_{A_0}\otimes \mathds{1}^E)
            \varrho^{SE}_{t_0}
            (O_{B_0}\otimes \mathds{1}^E)
        \Big)
       \cdots
    \Big)
    (O_{B_n}\otimes \mathds{1}^E)
\Big].
\end{aligned}\label{eq:TraceWightman}
\end{equation}
The left and right Wightman correlation functions are obtained as special cases of this expression by replacing all $O_{B_k}$ (respectively, all $O_{A_k}$) with identity operators.
By decomposing the observables in the IC-POVM basis $O_{X}=\sum_{\alpha} c_{\alpha}F_{\alpha}$, it is clear that Eq.~\eqref{eq:TraceWightman} coincides with Eq.~\eqref{eq:KDaveWightman}.

We can further generalize this construction by replacing the observables with quantum instruments. In particular, we replace the left (resp. right) observable
$O_{A_k}\otimes\mathds{1}^E(\bullet)\mathds{1}^{SE}$ at each time step $t_k$ by the corresponding quantum instrument $\{(\mathcal{F}_{\alpha}^{(k)} , \alpha^{(k)})\}_{\alpha^{(k)}}$:
\begin{equation}
    \mathcal{F}_{\alpha}(\bullet)
    =
    \sum_k K_{\alpha,k}(\bullet)K_{\alpha,k}^{\dagger},
\end{equation}
and then average over the outcome $\alpha^{(k)}$ with the corresponding weight. This provides, arguably, the most general class of Wightman correlation functions accessible in open-system quantum field theory. Note that a genuinely \emph{doubled} insertion $F_\alpha(\bullet)F'_\beta$ with $\alpha\neq\beta$ is not of instrument form: it is not completely positive, and this is exactly why the doubled distribution is not directly measurable by a sequential measurement but requires the interferometric protocols or quantum snapshotting \cite{Jia2025TemporalKirkwoodDirac,jia2026temporaltomography}.

%================
\section{Spatiotemporal states, temporal link products and recursive expression}
\label{sec:tempstate}

The key idea underlying the spatiotemporal-state construction is to introduce a tensor-product structure across different times and define the temporal space as $\mathcal{H}=\mathcal{H}_{t_n}\otimes \cdots \otimes \mathcal{H}_{t_0}$.
Temporal states can be defined in a unified manner through a temporal generalization of the Bloch representation for spatial density operators \cite{Jia2025TemporalKirkwoodDirac}, in the same spirit as the construction of the pseudo-density operator \cite{fitzsimons2015quantum}. An $n$-qudit spatial state can be written as
\begin{equation}
    \varrho_{x_n,\cdots,x_1} = \frac{1}{d^n} \sum_{\mu_1,\cdots,\mu_n} \langle \{ \sigma_{\mu_n}, \cdots  ,\sigma_{\mu_1} \}\rangle_{\rm sp}\,   \sigma_{\mu_n}\otimes \cdots \otimes \sigma_{\mu_1},
\end{equation}
where $\sigma_{\mu}$ are Hilbert--Schmidt operators and $\langle \{ \sigma_{\mu_n}, \cdots  ,\sigma_{\mu_1} \}\rangle_{\rm sp}$ denote the corresponding spatial correlation functions.
We then extend this construction naturally to the temporal setting by replacing the spatial correlation functions with temporal ones,
\begin{equation}\label{eq:temporalbloch}
    \Ups_{t_n,\cdots,t_0}= \frac{1}{d^{\,n+1}} \sum_{\mu_0,\cdots,\mu_n} \langle \{ \sigma_{\mu_n}, \cdots  ,\sigma_{\mu_0} \} \rangle_{\rm tp}\; \sigma_{\mu_n}\otimes \cdots \otimes \sigma_{\mu_0},
\end{equation}
where $\langle \{ \sigma_{\mu_n}, \cdots  ,\sigma_{\mu_0} \} \rangle_{\rm tp}$ denote the corresponding temporal correlation functions defined in the temporal phase space.
More generally, for a spatiotemporal state, we first evaluate the relevant spatiotemporal correlation functions and then reconstruct the corresponding state through the same procedure.

If we take doubled, left, right and mixed-branch temporal Kirkwood--Dirac; doubled, left/right Margenau--Hill; or L\"uders--von Neumann distributions, we obtain the corresponding temporal states (see \cite{Jia2025TemporalKirkwoodDirac} for details):
\begin{itemize}
    \item The doubled Kirkwood--Dirac temporal state $\Upsd^{\rm KD}_{t_n,\cdots,t_0}$ coincides with the doubled density operator defined in \cite{jia2024spatiotemporal}. It is also closely related to the superdensity operator \cite{cotler2018superdensity}: both yield the same doubled Wightman correlation functions for Hilbert-Schmidt operators. However, while the superdensity operator constructs a positive semidefinite matrix by regarding these correlation functions as matrix elements, the doubled Kirkwood--Dirac temporal state defined here is generally non-Hermitian. This non-Hermiticity offers advantages in detecting temporality, as discussed in \cite{jia2024spatiotemporal}.
    \item Left and right Kirkwood--Dirac temporal states $\Upsl^{\rm KD}_{t_n,\cdots,t_0}$ and $\Upsr^{\rm KD}_{t_n,\cdots,t_0}$; they are adjoint to each other,
    \begin{equation}\label{eq:adjointpair}
        \Upsl^{\rm KD}_{t_n,\cdots,t_0}=\big(\Upsr^{\rm KD}_{t_n,\cdots,t_0}\big)^{\dagger},
    \end{equation}
    and they are the left and right reduced states of the doubled Kirkwood--Dirac temporal state (Corollary~\ref{cor:marginals}). The mixed-branch Kirkwood--Dirac temporal states $\overline{\Upsilon}^{\rm KD}_{t_n,\cdots,t_0}$ is also a reduced state of the doubled Kirkwood--Dirac temporal state (Corollary~\ref{cor:marginals}).
    The formalism discussed in \cite{milekhin2025observable,lie2025probingquantumstatesspacetime} is actually a two-time left (or right) Kirkwood--Dirac temporal state.
    \item Doubled Margenau--Hill temporal state $\Upsd^{\MH}$, defined by
    \begin{equation}
        \Upsd^{\MH}_{t_n,\cdots,t_0} =\frac{1}{2}\left[\Upsd^{\rm KD}_{t_n,\cdots,t_0}+(\Upsd^{\rm KD}_{t_n,\cdots,t_0})^{\dagger} \right];
    \end{equation}
    this is the Hermitianization of the doubled Kirkwood--Dirac temporal state.
    It is also convenient to introduce the another Hermitianization of the Kirkwood--Dirac temporal state (which we call dual Margenau--Hill temporal state):
    \begin{equation}
        \Omegad^{\rm MH}_{t_n,\cdots,t_0} =\frac{1}{2i}\left[\Upsd^{\rm KD}_{t_n,\cdots,t_0}-(\Upsd^{\rm KD}_{t_n,\cdots,t_0})^{\dagger} \right];
    \end{equation}
    Note that $\Upsd^{\rm KD}_{t_n,\cdots,t_0}= \Upsd^{\MH}_{t_n,\cdots,t_0} + i\,\Omegad^{\rm MH}_{t_n,\cdots,t_0}$.
    When applying the temporal Born rule, $\Upsd^{\MH}_{t_n,\cdots,t_0}$ yields the real-part averaging, while $\Omegad^{\rm MH}_{t_n,\cdots,t_0}$ gives the imaginary-part averaging.
\item The left and right Margenau--Hill temporal states coincide:
    \begin{equation}
        \Upsl^{\MH}_{t_n,\cdots,t_0}=\Upsr^{\MH}_{t_n,\cdots,t_0} =\frac{1}{2} \left[ \Upsl^{\rm KD}_{t_n,\cdots,t_0}+  \Upsr^{\rm KD}_{t_n,\cdots,t_0}\right].
    \end{equation}
    Similarly, we can introduce the anti-Hermitian counterparts
    \begin{equation}
        \Omegal^{\MH}_{t_n,\cdots,t_0}= \Omegar^{\MH}_{t_n,\cdots,t_0}=\frac{1}{2i} \left[ \Upsl^{\rm KD}_{t_n,\cdots,t_0}-  \Upsr^{\rm KD}_{t_n,\cdots,t_0}\right].
    \end{equation}
    We see that $\Upsl^{\rm KD}_{t_n,\cdots,t_0}=\Upsr^{\MH}_{t_n,\cdots,t_0}+i\Omegar^{\MH}_{t_n,\cdots,t_0}$ and $\Upsr^{\rm KD}_{t_n,\cdots,t_0}=\Upsl^{\MH}_{t_n,\cdots,t_0}-i\Omegal^{\MH}_{t_n,\cdots,t_0}$.
    These two Hermitianizations provide the real and imaginary parts of the averaging when computing the Wightman correlation function in the temporal phase space.
    \item L\"uders--von Neumann temporal state $ \Ups^{\rm LvN}_{t_n,\cdots,t_0}$ coincides with the pseudo-density operator of \cite{fitzsimons2015quantum,Jia2025TemporalKirkwoodDirac,fullwood2025spatiotemporalbornrule,Lie2025stateovertime,lie2025probingquantumstatesspacetime}. For the two-time case the left/right Margenau--Hill temporal state coincides with the von Neumann temporal state; we have re-derived this identity numerically for random qubit channels and random measurement bases. This means that Born rule for L\"uders--von Neumann temporal state can only reproduce the correct correlation functions, but cannot reproduce the L\"uders--von Neumann distributions. For this reason, this kind of temporal state is not well suited for QFT in some situations.
\end{itemize}

\begin{definition}[Spatiotemporal state through spatiotemporal Bloch tomography~\cite{Jia2025TemporalKirkwoodDirac}]
    Temporal states are temporal Bloch expressions constructed from temporal Wightman correlation functions, i.e., Eq.~\eqref{eq:temporalbloch}, with $\langle\cdots\rangle_{\rm tp}$ given by Eq.~\eqref{eq:KDaveWightman}. Equivalently, a temporal state $\Ups$ is the unique operator on $\cH_{t_n}\otimes\cdots\otimes\cH_{t_0}$ that reproduces the chosen temporal quasiprobability distribution through the temporal Born rule
    \begin{equation}\label{eq:tempborn}
        Q(\alpha_n,\dots,\alpha_0)
        =
        \Tr\!\left[
            (F_{\alpha_n}\otimes\cdots\otimes F_{\alpha_0})\,\Ups
        \right].
    \end{equation}
More generally, in the spatiotemporal setting $\mathcal{H}_{\rm tot}=\otimes_{(x,t)}\mathcal{H}_{x,t}$, spatiotemporal Wightman correlation functions similarly define a corresponding spatiotemporal state, from which the spatiotemporal quasiprobability distribution is obtained via the spatiotemporal Born rule. For example, in closed-system QFT, we can choose local operators $\sigma_{\mu}(x,t),\sigma_{\nu}(y,t'),\cdots,\sigma_{\gamma}(z,t'')$ at different spacetime points and evaluate their Wightman correlation functions in the vacuum state. These correlation functions can then be used to construct the corresponding doubled, left or right Kirkwood--Dirac spatiotemporal state.
\end{definition}

\begin{remark}[Equivalence of spatiotemporal quasiprobability distribution and spatiotemporal  state]
\label{remk:equivalence}
An informationally complete temporal (spatiotemporal) quasiprobability distribution is equivalent to the temporal  (spatiotemporal)  state. 
Indeed, from such a distribution one can reconstruct all temporal correlation functions, and therefore uniquely determine the corresponding temporal state $\Upsilon$. 
Conversely, given a temporal state $\Upsilon$, the temporal Born rule yields a quasiprobability distribution for any chosen operators; when the chosen operator basis is informationally complete, the mapping from $\Upsilon$ to the quasiprobability distribution $Q$ is bijective. Thus, up to the choice of tomographic frame, the temporal state and the informational complete temporal quasiprobability distribution carry precisely the same physical information.
\end{remark}

\subsection{Recursive structure of spatiotemporal states}
\label{sec:recursion}

%\begin{definition}[Jamio\l kowski operator]\label{def:choi}

We now derive a recursive expression for spatiotemporal states, providing a complementary description to the spatiotemporal Bloch tomography~\cite{Jia2025TemporalKirkwoodDirac}. More importantly, this formulation leads to a broader unification in terms of probabilistic mixtures over all possible temporal link products of the initial state and the Jamio{\l}kowski operators associated with the intervening evolutions.

Let $\cE:\BB(\cH_{\rm in})\to\BB(\cH_{\rm out})$ be a CPTP map, its Jamio\l kowski operator is defined by
$ J[\mathcal{E}]
    =
    \sum_{i,j}\mathcal{E}(|i\rangle\langle j|)
    \otimes |j\rangle\langle i|$, which is an operator on $\cH_{\rm out}\otimes\cH_{\rm in}$.
If $\{\sigma_\mu\}$ is a Hilbert--Schmidt basis with $\Tr[\sigma_\mu\sigma_\nu]=d\,\delta_{\mu\nu}$, then the operator satisfying Eq.~\eqref{eq:choipairing} is $J[\cE]=\frac{1}{d}\sum_\mu \cE(\sigma_\mu)\otimes\sigma_\mu$.
It is easy to verify
\begin{equation}\label{eq:choipairing}
   \Tr\!\left[(A\otimes B)\,J[\cE]\right]=\Tr\!\left[A\,\cE(B)\right]
   \qquad\text{for all }A\in\BB(\cH_{\rm out}),\;B\in\BB(\cH_{\rm in}),
\end{equation}
which is crucial for proving the recursive expression of spatiotemporal state.

%\begin{definition}[Padded temporal link product]\label{def:link}
Let $N\in\BB(\cH_C\otimes\cH_B)$ and $M\in\BB(\cH_B\otimes\cH_A)$, where $\cH_A$ may itself be a tensor product of earlier time slots. The padded \emph{temporal link product} is defined by
\begin{equation}\label{eq:link}
   N\stlink M:=\left(N_{CB}\otimes\mathds{1}_A\right)\left(\mathds{1}_C\otimes M_{BA}\right)\;\in\;\BB(\cH_C\otimes\cH_B\otimes\cH_A).
\end{equation}
%\end{definition}
Unlike the comb link product of Ref.~\cite{Chiribella2009comb}, Eq.~\eqref{eq:link} contains no partial trace and no partial transpose, the intermediate space is kept, because in a spatiotemporal state every time slot must survive as a tensor factor. This is exactly the structural difference between a process tensor (where intermediate slots are contracted with the inserted operations) and a temporal state (where they remain open).

We first consider how to express the memoryless temporal state as a temporal link product of initial states and evolution maps. In the spatiotemporal setting, one can write down all equal-time many-body spatial states and combine them with the temporal recursive expression to obtain the general form.

% \begin{theorem}[Left and right Kirkwood-Dirac temporal states \cite{Jia2025TemporalKirkwoodDirac}]\label{thm:rec-single}
It was proven in Ref.~\cite{Jia2025TemporalKirkwoodDirac} that, for a memoryless process $(\varrho_{t_0},\cE_1,\dots,\cE_n)$ with $\cE_k=\cE_{ t_{k-1},t_k}$, we have the recursive expression
\begin{equation}\label{eq:recright}
   \Upsr_{t_k,\cdots ,t_0}=J[\cE_k]\stlink\Upsr_{t_{k-1},\cdots, t_0},
   \qquad \Upsr_{t_0}=\varrho_{t_0}.
\end{equation}
For all $X_k\in\BB(\cH_{t_k})$,
\begin{equation}
  \Tr\!\left[\Upsr_{t_n,\cdots, t_0}\,(X_n\otimes\cdots\otimes X_0)\right]
  =\Tr\!\left[\cE_n\!\left(\cdots\cE_1(\varrho_{t_0}X_0)X_1\cdots\right)X_n\right],
\end{equation}
i.e.\ $\Upsr$ is the right Kirkwood-Dirac temporal  state defined via temporal Bloch tomography. Moreover $\Upsl=\Upsr^{\dagger}$ is the left temporal KD state,  and the single-time marginals are the physical states,
\begin{equation}
\Tr_{t_n,\dots,\widehat{t_k},\dots,t_0}\Upsr=\varrho_{t_k}=\cE_k\circ\cdots\circ\cE_1(\varrho_{t_0}).
\end{equation}
This implies that these states satisfy the quantum Kolmogorov consistency condition and  $\Tr\Upsr=1$.
For left  Kirkwood-Dirac temporal  state, $\Upsl=\Upsr^{\dagger}$, similar results hold.

Since the left- and right-branch spatiotemporal states are special cases of the mixed-branch construction, we have the following more general result:

\begin{theorem}[Mixed-branch Kirkwood--Dirac temporal state]
    For the mixed-branch case, one must specify, for each time step, whether the left or right time slot is chosen as the free (i.e., uncontracted) slot; for example, one may have the choose branch as $\ttB=(t_n^R,\cdots,t_3^L,t_2^R,t_1^L,t_0^R)$ \footnote{We focus primarily on mixed-branch temporal states with a single left or right slot at each time step, although both slots may be left open at selected time steps.}. For each consecutive pair of time steps $t_{k-1}$ and $t_k$, there are four possible configurations of free time slots\footnote{This should not be confused with the mixed-branch time slots of the full temporal state.} of $J[\cE_k]$. 
    \begin{itemize}
        \item $(t_{k-1}^L,t_k^L)$: this means that $(t_{k-1}^R,t_k^R)$ are the contracted time slots, and we have $J[\cE_k]\stlink J[\cE_{k-1}]\stlink J[\cE_{k+1}]$.
        \item $(t_{k-1}^L,t_k^R)$: this means that $(t_{k-1}^R,t_k^L)$ are the contracted time slots, and we have $J[\cE_{k+1}]\stlink J[\cE_{k}]\stlink J[\cE_{k-1}]$.
        \item $(t_{k-1}^R,t_k^L)$: this means that $(t_{k-1}^L,t_k^R)$ are the contracted time slots, and we have $J[\cE_{k-1}]\stlink J[\cE_{k}]\stlink J[\cE_{k+1}]$.
        \item $(t_{k-1}^R,t_k^R)$: this means that $(t_{k-1}^L,t_k^L)$ are the contracted time slots, and we have $J[\cE_{k-1}]\stlink J[\cE_{k+1}] \stlink J[\cE_k]$.
    \end{itemize}
    The contracted slots determine the ordering of the temporal link product; note that $J[\cE_{m}]\stlink J[\cE_{l}]= J[\cE_{l}]\stlink J[\cE_{m}]$ whenever $|l-m|\geq 2$. 
    A tensor-network representation, in which these structures become much more transparent, is provided in Section~\ref{sec:tensornet}.
\end{theorem}

The doubled state requires one additional ingredient: the swap operator $\mathbb{S}=\sum_{i,j}|i\rangle \langle j|\otimes |j\rangle \langle i|=\frac{1}{d}\sum_{\alpha}\sigma_{\alpha}\otimes \sigma_{\alpha}$ on $\cH^{L}\otimes\cH^{R}$, by definition $\mathbb{S}\ket{\psi}\otimes\ket{\varphi}=\ket{\varphi}\otimes\ket{\psi}$, and we have $\Tr[(X\otimes Y)\mathbb{S}]=\Tr[XY]$.  Geometrically, $\mathbb{S}$ plays the role of the ``cap'' that closes the Schwinger--Keldysh contour.

\begin{theorem}[Doubled  Kirkwood-Dirac temporal state]\label{thm:rec-doubled}
Let $\cH_{t_k}^{L}\otimes\cH_{t_k}^{R}$ be the doubled slot at time $t_k$ and define the \emph{doubled Jamio\l kowski operator}
\begin{equation}\label{eq:doubledchoi}
   \JJ[\cE]\;:=\;\left(J[\cE]^{\,L_k,\,R_{k-1}}\otimes\mathds{1}^{\,R_k,\,L_{k-1}}\right)
   \left(\mathbb{S}_{t_k}\otimes\mathds{1}_{t_{k-1}}\right)
\;=\;\sum_{i,j}\Big[\big(\cE(\ket{i}\!\bra{j})\otimes\mathds{1}\big)\mathbb{S}\Big]_{t_k}\otimes\Big[\mathds{1}\otimes\ket{j}\!\bra{i}\Big]_{t_{k-1}} .
\end{equation}
Then the doubled temporal KD state obeys
\begin{equation}\label{eq:recdoubled}
   \Upsd_{t_k\cdots t_0}=\JJ[\cE_k]\stlink \Upsd_{t_{k-1}\cdots t_0},
   \qquad \Upsd_{t_0}=(\varrho_{t_0}\otimes\mathds{1})\,\mathbb{S},
\end{equation}
and reproduces Eq.~\eqref{eq:QD} through the Born rule
\begin{equation}
   \Tr\!\left[\bigotimes_{k=0}^{n}\left(A_k\otimes B_k\right)\Upsd_{t_n\cdots t_0}\right]
   =\Tr\!\left[A_n\cE_n\!\left(\cdots A_1\cE_1(A_0\varrho_{t_0}B_0)B_1\cdots\right)B_n\right].
\end{equation}
\end{theorem}

\begin{proof}
We prove that the functional realized by $\Upsd_{t_k,\cdots , t_0}$ is
\begin{equation}\label{eq:functional}
  \Tr\!\left[(A_k\otimes B_k\otimes\cdots\otimes A_0\otimes B_0)\Upsd_{t_k\cdots t_0}\right]=\Tr\!\left[A_kG_kB_k\right],
\end{equation}
where $G_k=\cE_k\!\left(A_{k-1}G_{k-1}B_{k-1}\right)$, $G_{-1}:=\varrho_{t_0}$, 
with the convention $G_0=A_0\varrho_{t_0}B_0$ read off from the base case.

\emph{Base case.} $\Tr[(A_0\otimes B_0)(\varrho\otimes\mathds{1})\mathbb{S}]=\Tr[(A_0\varrho\otimes B_0)\mathbb{S}]=\Tr[A_0\varrho B_0]$, using $\Tr[(X\otimes Y)\mathbb{S}]=\Tr[XY]$.

\emph{Induction step.} Assume Eq.~\eqref{eq:functional} at level $k-1$, we have
\begin{equation}
\begin{aligned}
     &\Tr\!\left[\bigotimes_{k=0}^{n}\left(A_k\otimes B_k\right)\Upsd_{t_n,\cdots ,t_0}\right]\\
     = & \Tr_{t_k} \sum_{i,j} (A_k\otimes B_k) \Big[\big(\cE(\ket{i}\!\bra{j})\otimes\mathds{1}\big)\mathbb{S}\Big]_{t_k}\\
    & \otimes \Tr_{t_{k-1},\cdots,t_0}
     \Big[ (A_{k-1}\otimes B_{k-1}\ket{j}\!\bra{i})_{t_{k-1}}  \otimes A_{k-2}\otimes B_{k-2}\cdots )\Big]  \Upsd_{t_{k-1},\cdots ,t_0}\\
     =&\Tr_{t_k} \sum_{i,j} (A_k\otimes B_k) \Big[\big(\cE(\ket{i}\!\bra{j})\otimes\mathds{1}\big)\mathbb{S}\Big]_{t_k}  \Tr\!\left[A_{k-1}G_{k-1}B_{k-1}\ket{j}\!\bra{i}\right]
\end{aligned}
\end{equation}
Summing over $i,j$ and using $\sum_{ij}\bra{i}Y\ket{j}\ket{i}\!\bra{j}=Y$ with $Y=A_{k-1}G_{k-1}B_{k-1}$ gives
$\Tr[A_k\,\cE_k(A_{k-1}G_{k-1}B_{k-1})\,B_k]=\Tr[A_kG_kB_k]$, as claimed.
\end{proof}

\begin{corollary}[Mixed-branch marginals]\label{cor:marginals}
The relations
$\Tr_{R}\Upsd=\Upsl$ and $\Tr_{L}\Upsd=\Upsr$
hold, where $\Tr_{L/R}$ denotes the partial trace over the left or right copy of every time slot. Consequently,
\begin{equation}
    \Tr_R\Upsd=(\Tr_L\Upsd)^\dagger.
\end{equation}
More generally, one may choose an independent branch assignment at each time slot, for example,
$\ttB=(t_0^L,t_1^R,t_2^R,t_3^L,\ldots,t_n^L)$,
giving a mixed-branch spatiotemporal state $\overline{\Upsilon}_B$. There are $2^{n+1}$ such branch assignments and corresponding spatiotemporal states. The $B$-marginal of the doubled Kirkwood--Dirac temporal state yields the corresponding mixed-branch state:
\begin{equation}
    \Tr_{\ttB^c}\Upsd=\overline{\Upsilon}_{\ttB}.
\end{equation}
and $\overline{\Upsilon}_{\ttB}^{\dagger}=\overline{\Upsilon}_{\ttB^c}$.
\end{corollary}

%\begin{corollary}[Margenau--Hill temporal states]
The recursive expression for Kirkwood--Dirac spatiotemporal states naturally yields a temporal link product expression for the corresponding Margenau--Hill spatiotemporal states, since these are obtained via Hermitianization of the Kirkwood--Dirac temporal states.
%\end{corollary}

%\begin{corollary}[Recursive expression for LvN temporal state]
    The L\"uders-von Neumann temporal state is known to have the recursive expression \cite{Liu2025PDO,fullwood2022quantum}
    \begin{equation}
        \Ups^{\rm LvN}_{t_n,\cdots,t_0}=\frac{1}{2}\{J[\cE_k],\Ups^{\rm LvN}_{t_{k-1},\cdots,t_0}\}, \quad \Ups^{\rm LvN}_{t_0}=\varrho_{t_0}.
    \end{equation}
   where $\{\bullet,\bullet\}$ is a anticommutator. 
  It is an equal-weight mixture of all single-space mixed-branch Kirkwood-Dirac temporal states.
  For doubled Kirkwood-Dirac temporal state, we take all possible marginals and then mix them together to obtain the L\"uders-von Neumann temporal state.

\begin{example}
For a time evolution map $\mathcal{E}_{t_0,t}$ (with $t\in \Rbb$ and $t\geq t_0$) and initial state $\rho_{t_0}$, consider a three-step quantum process (which suffices to illustrate the main results). The left and right Kirkwood--Dirac temporal states are given by
\begin{equation}
     \Upsl^{\rm KD}_{t_2,t_1,t_0}= J[\mathcal{E}_{t_1,t_2}] \star J[\mathcal{E}_{t_0,t_1}] \star \rho_{t_0},\quad \Upsr^{\rm KD}_{t_2,t_1,t_0}=\rho_{t_0} \star J[\mathcal{E}_{t_0,t_1}] \star J[\mathcal{E}_{t_1,t_2}].
\end{equation}
The other two mixed-branch Kirkwood--Dirac temporal states take the form
\begin{equation}
    \overline{\Upsilon}^{\rm KD}_{t_2,t_1,t_0}=J[\mathcal{E}_{t_1,t_2}] \star \rho_{t_0} \star J[\mathcal{E}_{t_0,t_1}], \quad {\overline{\Upsilon'}}^{\rm KD}_{t_2,t_1,t_0}=J[\mathcal{E}_{t_0,t_1}] \star \rho_{t_0} \star J[\mathcal{E}_{t_1,t_2}].
\end{equation}
The L\"uders--von Neumann temporal state takes the form
\begin{equation}
\begin{aligned}
     \Upsilon^{\mathrm{LvN}}_{t_2,t_1, t_0} = &\frac{1}{4} \Big[
         J[\mathcal{E}_{t_1,t_2}] \star J[\mathcal{E}_{t_0,t_1}] \star \rho_{t_0} 
         + J[\mathcal{E}_{t_1,t_2}] \star \rho_{t_0} \star J[\mathcal{E}_{t_0,t_1}] \\
        & + J[\mathcal{E}_{t_0,t_1}] \star \rho_{t_0} \star J[\mathcal{E}_{t_1,t_2}] 
         + \rho_{t_0} \star J[\mathcal{E}_{t_0,t_1}] \star J[\mathcal{E}_{t_1,t_2}]
    \Big].
\end{aligned}
\end{equation}
As we will argue later, this should be regarded as an equal-weight probabilistic mixture of all mixed-branch Kirkwood--Dirac temporal states.
In contrast, the left and right Margenau--Hill states contain only a sum of two terms with a fixed product order:
\begin{equation}
    \Upsilon^{\mathrm{MH}}_{t_2,t_1, t_0} = \frac{1}{2} \left( 
        J[\mathcal{E}_{t_1,t_2}] \star J[\mathcal{E}_{t_0,t_1}] \star \rho_{t_0} 
        + \rho_{t_0} \star J[\mathcal{E}_{t_0,t_1}] \star J[\mathcal{E}_{t_1,t_2}]
    \right),
\end{equation}
We observe that the time-ordered link product yields the left or right Kirkwood--Dirac temporal state, while the non-time-ordered product gives the mixed-branch Kirkwood--Dirac temporal state. The L\"uders--von Neumann temporal state is then an equal-weighted mixture of all possible temporal link products with all possible time orderings.
\end{example}

For the single local space case (namely, when the temporal local space is not doubled), all the notions of spatiotemporal states introduced above can actually be unified into a single object as follows.

\begin{theorem}[Unification of spatiotemporal state through probabilistic mixture]\label{theorem:UnificationSingle}
   From the above discussion, it follows that for a given initial state $\varrho_{t_0}$ and evolution operators $J[\mathcal{E}_{t_{k-1},t_k}]$, one may form arbitrary temporal link products to obtain all mixed-branch spatiotemporal state, denoted by $\overline{\Upsilon}^{\mathrm{KD}, 1}_{t_n,\cdots,t_0},\cdots,  \overline{\Upsilon}^{\mathrm{KD}, 2m}_{t_n,\cdots,t_0}$ (Note that there are always a total of $2m=2^{n+1}$ such states, as each state and its Hermitian conjugate are both included. We label 1 as the left order, 2 as the right order (which is the Hermitian conjugate of 1), and so on.) These can then be mixed according to a probability distribution $\mathbb{P}(i)$ to yield the state
   \begin{equation}
       \Upsilon^{\mathbb{P}}_{t_n,\cdots,t_0}=\sum_i\mathbb{P}(i)\,\overline{\Upsilon}^{\mathrm{KD}, i}_{t_n,\cdots,t_0}.
   \end{equation}
   Evidently, all equal-time reduced states $\varrho_{t_k}$ can be recovered from this construction. The left, right, and mixed-branch Kirkwood--Dirac states, the left/right Margenau--Hill spatiotemporal states, and the L\"uders--von Neumann spatiotemporal states (pseudo-density operator) are all recovered as special cases corresponding to extreme choices of the probability distribution $\mathbb{P}$. Note that replacing $\mathbb{P}$ with a (even complex-valued) quasiprobability distribution $\mathbb{Q}$, we will also obtain an operator $\Ups^{\mathbb{Q}}_{t_n,\cdots,t_0}$ that has correct equal-time marginals $\rho_{t_k}$.
\end{theorem}

\begin{proof}
    The proof is straightforward from the recursive expression derived above. For example, the left Kirkwood--Dirac temporal state corresponds to the probability distribution $\mathbb{P}(i)=(1,0,\cdots,0)$; the right state corresponds to $\mathbb{P}(i)=(0,1,0,\cdots,0)$. The Margenau--Hill state corresponds to $\mathbb{P}(i)=(1/2,1/2,0,\cdots,0)$, while the L\"uders--von Neumann temporal state corresponds to $\mathbb{P}(i)=(1/2m,1/2m,\cdots,1/2m)$.
\end{proof}

This unifies almost all existing single-local-space constructions of spatiotemporal states. For the two-time case, an even more general unification is possible; see Section~\ref{sec:TwoTimeState}.

We emphasize that general Kirkwood--Dirac spatiotemporal states are non-Hermitian and nonnormal operators. As a simple example, consider the two-time state $\Upsr_{t_1,t_0}$ associated with the identity channel $\mathcal{E}=\id$ and the initial state $\varrho_{t_0}$:
\begin{equation}
    \Upsr_{t_1,t_0}=\mathbb S (\I\otimes \varrho_{t_0} )
\end{equation}
One finds
\begin{equation}
    \Upsr_{t_1,t_0}\Upsr_{t_1,t_0}^{\dagger}
    =\varrho_{t_0}^2\otimes \mathds{1},
    \qquad
    \Upsr_{t_1,t_0}^{\dagger}\Upsr_{t_1,t_0}
    =\mathds{1}\otimes\varrho_{t_0}^2 .
\end{equation}
Since normality requires these two operators to coincide, $\Upsr_{t_1,t_0}$ is normal if and only if $\varrho_{t_0}^2\otimes\mathds{1}=\mathds{1}\otimes\varrho_{t_0}^2$, which holds precisely for the maximally mixed state $\varrho_{t_0}=\mathds{1}/d$.
The doubled spatiotemporal state may be nonnormal even on a fixed time slice $\Upsd_{t_k}$, one must then take the left or right marginals to obtain the corresponding density operator $\rho_{t_k}=\Tr_L\Upsd_{t_k}=\Tr_R\Upsd_{t_k}$.
A simplest example is
\begin{equation}\label{eq:one-time-doubled-state}
 \Upsd_{t_0} =(\varrho_{t_0}\otimes\mathds1)\mathbb S,
\end{equation}
which is nonnormal.

\subsubsection{Backward (Heisenberg) recursion}

Every forward (Schr\"odinger) recursion admits a dual backward (Heisenberg) recursion. Repeatedly applying
$\Tr[Y\cE(X)]=\Tr[\cE^{\dagger}(Y)X]$,
where $\cE^{\dagger}$ is the unital CP map describing the backward evolution of observables, to Eq.~\eqref{eq:QD} gives the following result.

\begin{proposition}[Backward (Heisenberg) recursion]\label{prop:backward}
For left- and right-branch operators
$A_{\alpha_0},B_{\beta_0}$, $\ldots$, $A_{\alpha_n},B_{\beta_n}$,
define the spatiotemporal effect operators recursively by
\begin{equation}
  M_{\alpha_n,\beta_n}=B_{\beta_n}A_{\alpha_n},
  \qquad
  M_{\alpha_n,\beta_n,\ldots,\alpha_{k-1},\beta_{k-1}}
  =
  B_{\beta_{k-1}}\,
  \cE_k^{\dagger}\!\left(
  M_{\alpha_n,\beta_n,\ldots,\alpha_k,\beta_k}
  \right)
  A_{\alpha_{k-1}},
  \quad k=n,\ldots,1.
\end{equation}
Then
\begin{equation}
  \Tr\!\left[
  \Upsd
  \bigotimes_{k=0}^n
  (A_{\alpha_k}\otimes B_{\beta_k})
  \right]
  =
  \Tr\!\left[
  M_{\alpha_n,\beta_n,\ldots,\alpha_0,\beta_0}
  \varrho_{t_0}
  \right].
\end{equation}
When $A_{\alpha_k}$ and $B_{\beta_k}$ are chosen as spatiotemporal frame operators, this yields the doubled spatiotemporal Kirkwood--Dirac distribution $\Qd_{\rm KD}(\alpha_n,\beta_n,\cdots,\alpha_0,\beta_0)$. The left-, right-, and mixed-branch cases are obtained by setting the appropriate $A_{\alpha_k}$ or $B_{\beta_k}$ to the identity. The spatiotemporal effect operators thus encode the temporal phase-space information.
\end{proposition}

Thus, the doubled spatiotemporal Kirkwood--Dirac distribution can be computed either by propagating the state forward with $\cE_{t_0,t}$ or by propagating the effect backward with $\cE_{t_0,t}^{\dagger}$. For a given spatiotemporal frame, the effect operator
$M_{\alpha_n,\beta_n,\ldots,\alpha_0,\beta_0}$
encodes the temporal correlations of the quantum process. Its marginals are obtained by summing over the corresponding indices $\alpha_k$ or $\beta_k$, yielding lower-order spatiotemporal effect operators. These properties were discussed in detail in Ref.~\cite{Jia2025TemporalKirkwoodDirac}. They are closely related to the nonclassicality of spatiotemporal Kirkwood--Dirac distributions, which will be discussed in Section~\ref{sec:quantumness}.

% when  $ M_{\alpha_n,\beta_n,\ldots,\alpha_0,\beta_0}\geq 0$, 

\subsubsection{Processes with memory}
\label{sec:recmemory}

For a process with memory, the same recursion holds on the dilated space. Since the spatiotemporal partial trace is well-defined, one may naturally trace out the environment to obtain the corresponding temporal states.

\begin{proposition}[Dilated recursion]\label{prop:recmemory}
Let the process be generated by $\varrho^{SE}_{t_0}$ and joint unitaries $\mathcal{U}_k$ on $\cH_S\otimes\cH_E$, with insertions acting on $S$ only. Then
\begin{equation}
   \Upsr^{SE}_{t_k\cdots t_0}=J[\mathcal{U}_k]\stlink\Upsr^{SE}_{t_{k-1}\cdots t_0},\qquad
   \Upsr^{S}_{t_n\cdots t_0}=\Tr_{E}\,\Upsr^{SE}_{t_n\cdots t_0},
\end{equation}
where $\Tr_E$ acts on the environment. The environment therefore enters only through the bond space of the recursion; the algebraic structure is unchanged.
\end{proposition}

Proposition~\ref{prop:recmemory} is what makes the temporal-entanglement bounds of Section~\ref{sec:STentropy} sharp: Markovianity is exactly the statement that the bond space is $\BB(\cH_S)$; memory enlarges it to $\BB(\cH_S\otimes\cH_E)$.

\subsubsection{Temporal tensor network representation}
\label{sec:tensornet}
It is convenient to represent temporal states using tensor network notation. A matrix can be depicted as 
\begin{equation}
    A_{ij}= \begin{aligned}
        \begin{tikzpicture}
\node[draw,line width=1pt, rectangle, minimum width=0.6cm, minimum height=0.6cm] (A) at (0,0) {$A$};
\draw[line width=1pt] (A.west) -- ++(-0.8,0) node[left] {$i$};
\draw[line width=1pt] (A.east) -- ++(0.8,0) node[right] {$j$};
\end{tikzpicture}
    \end{aligned}
\end{equation}
The identity operator is represented by a line without a box (occasionally, we may add a dot on the line for emphasis). Multiplication is represented by gluing two tensors along a shared leg. For operators carrying an index, such as $\sigma_{\mu}$ (or $F_{\alpha}$), we represent the additional index with an extra leg:
\begin{equation}
       \sigma_{\mu}= \begin{aligned}
        \begin{tikzpicture}
\node[draw,line width=1pt, rectangle, minimum width=0.6cm, minimum height=0.6cm] (A) at (0,0) {$A$};
\draw[line width=1pt] (A.west) -- ++(-0.8,0) node[left] {$i$};
\draw[line width=1pt] (A.east) -- ++(0.8,0) node[right] {$j$};
\draw[line width=1pt, red] (A.south) -- ++(0,-0.5) node[right] {$\mu$};
\end{tikzpicture}
    \end{aligned}.
\end{equation}
The swap operator $\langle ij|\mathbb{S}|kl\rangle=\delta_{il}\delta_{jk}$ is represented as 
\begin{equation}
    \mathbb{S}_{ij,kl}=\begin{aligned}
\begin{tikzpicture}
\draw[line width=1pt] node[left] {$j$} (0,0) -- ++(1,1)  node[right] {$k$};
\draw[line width=1pt]  (0,1) -- ++(0.47,-0.47) ;
\draw[line width=1pt]  (1,0) -- ++(-0.47,0.47);
\node at (-0.2,1) {$i$};
\node at (1.2,0) {$l$};
\end{tikzpicture}
    \end{aligned} \label{eq:swaptensor}
\end{equation}
The identities $\Tr_B (A\otimes B) \mathbb{S} =AB$ and $\Tr_A (A\otimes B) \mathbb{S} =BA$ can be expressed diagrammatically as
\begin{equation}
 \Tr_B[ (A\otimes B) \mathbb{S}]=  \begin{aligned}
        \begin{tikzpicture}
\node[draw,line width=1pt, rectangle, minimum width=0.6cm, minimum height=0.6cm] (A) at (0,0) {$A$};
\draw[line width=1pt] (A.west) -- ++(-0.8,0) ;
\draw[line width=1pt] (A.east) -- ++(0.8,0) ;
\node[draw,line width=1pt, rectangle, minimum width=0.6cm, minimum height=0.6cm] (B) at (0,1) {$B$};
\draw[line width=1pt] (B.west) -- ++(-0.8,0) ;
\draw[line width=1pt] (B.east) -- ++(0.8,0) ;
\draw[line width=1pt] (B.east)++(0.8,0) --  ++(0.47,-0.47);
\draw[line width=1pt] (B.east)++(0.8,0)++(0.47,-0.47)++(0.06,-0.06) -- ++(0.5,-0.5)  ;
\draw[line width=1pt] (A.east)++(0.8,0) -- ++(0.98,0.98) ;
\draw[line width=1pt] (B.west) ++(-0.8,0)  -- ++(0,0.4) -- ++(3.2,0)-- ++ (0,-0.42);
\end{tikzpicture}
    \end{aligned},
\quad
 \Tr_A [(A\otimes B) \mathbb{S}]=  \begin{aligned}
        \begin{tikzpicture}
\node[draw,line width=1pt, rectangle, minimum width=0.6cm, minimum height=0.6cm] (A) at (0,0) {$A$};
\draw[line width=1pt] (A.west) -- ++(-0.8,0) ;
\draw[line width=1pt] (A.east) -- ++(0.8,0) ;
\node[draw,line width=1pt, rectangle, minimum width=0.6cm, minimum height=0.6cm] (B) at (0,1) {$B$};
\draw[line width=1pt] (B.west) -- ++(-0.8,0) ;
\draw[line width=1pt] (B.east) -- ++(0.8,0) ;
\draw[line width=1pt] (B.east)++(0.8,0) --  ++(0.47,-0.47);
\draw[line width=1pt] (B.east)++(0.8,0)++(0.47,-0.47)++(0.06,-0.06) -- ++(0.5,-0.5)  ;
\draw[line width=1pt] (A.east)++(0.8,0) -- ++(0.98,0.98) ;
\draw[line width=1pt] (A.west) ++(-0.8,0)  -- ++(0,-0.4) -- ++(3.26,0)-- ++ (0,0.4);
\end{tikzpicture}
    \end{aligned}.
\end{equation}
This tensor network representation simplifies the structure of temporal states and renders it more intuitive.

We adopt the convention that time flows upwards. The tensor network representation of temporal quasiprobability distributions can be given as follows (example of doubled and mixed-branch Kirkwood-Dirac quasiprobability distributions)
\begin{align}
\begin{aligned}
     &\,\,\,\,\,\quad    \begin{aligned}
        \begin{tikzpicture}
\node[draw,line width=1pt, rectangle, minimum width=0.6cm, minimum height=0.6cm] (A) at (0,0) {$\varrho_{t_0}$};
\draw[line width=1pt] (A.west) -- ++(-0.5,0) -- ++(0,4);
\draw[line width=1pt] (A.east) -- ++(0.5,0) -- ++(0,4);
            \node[ line width=1pt, dashed, draw opacity=0.5] (a) at (-0.9,4.4){$\vdots$};
            \node[ line width=1pt, dashed, draw opacity=0.5] (a) at (0.9,4.4){$\vdots$};
\draw[line width=1pt] (-0.9,4.4) -- ++(0,2);
\draw[line width=1pt] (0.9,4.4) -- ++(0,2)--++(-1.8,0);
%%%%%%%%%%%F1%%%%
\node[draw, line width=1pt, circle, minimum size=0.6cm, inner sep=0pt,fill=white] (F1) at (-.9,0.6) {$F_{\alpha_0}$};
\draw[line width=1pt,red] (F1.west) -- ++(-0.5,0) ;
%%%%%%E1%%%%%%
 \node[draw, line width=1pt, rectangle, 
        minimum width=3cm, minimum height=0.6cm, 
        fill=white] (E1) at (0,1.6) {$\cE_{t_0,t_1}$};
%%%%%%%%%%%F2%%%%%%%%
\node[draw, line width=1pt, circle, minimum size=0.6cm, inner sep=0pt,fill=white] (F2) at (-.9,2.5) {$F_{\alpha_1}$};
\draw[line width=1pt,red] (F2.west) -- ++(-0.5,0) ;
%%%%%%E2%%%%%%
 \node[draw, line width=1pt, rectangle, 
        minimum width=3cm, minimum height=0.6cm, 
        fill=white] (E2) at (0,3.5) {$\cE_{t_1,t_2}$};
%%%%%%En%%%%%%
 \node[draw, line width=1pt, rectangle, 
        minimum width=3cm, minimum height=0.6cm, 
        fill=white] (E2) at (0,4.9) {$\cE_{t_1,t_2}$};
%%%%%%%%%%%Fn%%%%%%%%
\node[draw, line width=1pt, circle, minimum size=0.6cm, inner sep=0pt,fill=white] (Fn) at (-.9,5.7) {$F_{\alpha_n}$};
\draw[line width=1pt,red] (Fn.west) -- ++(-0.5,0) ;
%%%%%%%%%%%Fnr%%%%%%%%
\node[draw, line width=1pt, circle, minimum size=0.6cm, inner sep=0pt,fill=white] (Fnr) at (.9,5.7) {$F_{\beta_n}$};
\draw[line width=1pt,red] (Fnr.east) -- ++(0.5,0) ;
%%%%%%%%%%%F2r%%%%%%%%
\node[draw, line width=1pt, circle, minimum size=0.6cm, inner sep=0pt,fill=white] (F2r) at (.9,2.5) {$F_{\beta_1}$};
\draw[line width=1pt,red] (F2r.east) -- ++(0.5,0) ;
%%%%%%%%%%%F1r%%%%
\node[draw, line width=1pt, circle, minimum size=0.6cm, inner sep=0pt,fill=white] (F1r) at (.9,0.6) {$F_{\beta_0}$};
\draw[line width=1pt,red] (F1r.east) -- ++(0.5,0) ;
\end{tikzpicture} 
\end{aligned}\\
&\Qd_{\rm KD}(\alpha_n,\cdots,\alpha_0;\beta_n,\cdots,\beta_0)
\end{aligned}
\quad 
\begin{aligned}
     &    \begin{aligned}
        \begin{tikzpicture}
\node[draw,line width=1pt, rectangle, minimum width=0.6cm, minimum height=0.6cm] (A) at (0,0) {$\varrho_{t_0}$};
\draw[line width=1pt] (A.west) -- ++(-0.5,0) -- ++(0,4);
\draw[line width=1pt] (A.east) -- ++(0.5,0) -- ++(0,4);
            \node[ line width=1pt, dashed, draw opacity=0.5] (a) at (-0.9,4.4){$\vdots$};
            \node[ line width=1pt, dashed, draw opacity=0.5] (a) at (0.9,4.4){$\vdots$};
\draw[line width=1pt] (-0.9,4.4) -- ++(0,2);
\draw[line width=1pt] (0.9,4.4) -- ++(0,2)--++(-1.8,0);
%%%%%%%%%%%F1%%%%
\node[draw, line width=1pt, circle, minimum size=0.6cm, inner sep=0pt,fill=white] (F1) at (-.9,0.6) {$F_{\alpha_0}$};
\draw[line width=1pt,red] (F1.west) -- ++(-0.5,0) ;
%%%%%%E1%%%%%%
 \node[draw, line width=1pt, rectangle, 
        minimum width=3cm, minimum height=0.6cm, 
        fill=white] (E1) at (0,1.6) {$\cE_{t_0,t_1}$};
% %%%%%%%%%%%F2%%%%%%%%
% \node[draw, line width=1pt, circle, minimum size=0.6cm, inner sep=0pt,fill=white] (F2) at (-.9,2.5) {$F_{\alpha_1}$};
% \draw[line width=1pt,red] (F2.west) -- ++(-0.5,0) ;
%%%%%%E2%%%%%%
 \node[draw, line width=1pt, rectangle, 
        minimum width=3cm, minimum height=0.6cm, 
        fill=white] (E2) at (0,3.5) {$\cE_{t_1,t_2}$};
%%%%%%En%%%%%%
 \node[draw, line width=1pt, rectangle, 
        minimum width=3cm, minimum height=0.6cm, 
        fill=white] (E2) at (0,4.9) {$\cE_{t_1,t_2}$};
%%%%%%%%%%%Fn%%%%%%%%
\node[draw, line width=1pt, circle, minimum size=0.6cm, inner sep=0pt,fill=white] (Fn) at (-.9,5.7) {$F_{\alpha_n}$};
\draw[line width=1pt,red] (Fn.west) -- ++(-0.5,0) ;
% %%%%%%%%%%%Fnr%%%%%%%%
% \node[draw, line width=1pt, circle, minimum size=0.6cm, inner sep=0pt,fill=white] (Fnr) at (.9,5.7) {$F_{\beta_n}$};
% \draw[line width=1pt,red] (Fnr.east) -- ++(0.5,0) ;
%%%%%%%%%%%F2r%%%%%%%%
\node[draw, line width=1pt, circle, minimum size=0.6cm, inner sep=0pt,fill=white] (F2r) at (.9,2.5) {$F_{\beta_1}$};
\draw[line width=1pt,red] (F2r.east) -- ++(0.5,0) ;
%%%%%%%%%%%F1r%%%%
% \node[draw, line width=1pt, circle, minimum size=0.6cm, inner sep=0pt,fill=white] (F1r) at (.9,0.6) {$F_{\beta_0}$};
% \draw[line width=1pt,red] (F1r.east) -- ++(0.5,0) ;
\end{tikzpicture} 
\end{aligned}\\
&\quad  \overline{Q}_{\rm KD}(\alpha_n,\cdots,\beta_1,\alpha_0)
\end{aligned}
\end{align}
The left Kirkwood--Dirac temporal state $\Upsl^{\rm KD}_{t_n,\cdots,t_0}$ can be represented diagrammatically as follows (recall that $J[\mathcal{E}]=(\mathcal{E}\otimes \mathrm{id})(\mathbb{S})$, which admits a more explicit tensor network depiction based on Eq.~\eqref{eq:swaptensor})
\begin{equation} \label{eq:MPOdoub}
    \begin{aligned}
     &    \begin{aligned}
        \begin{tikzpicture}
\draw[line width=1pt] (0,0) -- ++(9,0) ; 
\node at (-.2,0) {$t_0$};
\draw[line width=1pt] (0,1) -- ++(9,0) ; 
\node at (-.2,1) {$t_1$};
\draw[line width=1pt] (0,2) -- ++(9,0) ;  
\node at (-.2,2) {$t_2$};
\draw[line width=1pt] (0,4) -- ++(9,0) ;  
\node at (-.2,4) {$t_n$};
%%%Jn
 \node[draw, line width=1pt, rectangle, 
        minimum width=0.6cm, minimum height=1.5cm, 
        fill=white] (Jn) at (1.8,3.6) {$J[\cE_{n}]$};   
%%%J2
 \node[draw, line width=1pt, rectangle, 
        minimum width=0.6cm, minimum height=1.5cm, 
        fill=white] (Jn) at (3.8,1.6) {$J[\cE_{2}]$};  
        %%%J1
 \node[draw, line width=1pt, rectangle, 
        minimum width=0.6cm, minimum height=1.5cm, 
        fill=white] (Jn) at (5.8,0.5) {$J[\cE_{1}]$}; 
\node[draw,line width=1pt, rectangle, minimum width=0.6cm, minimum height=0.6cm,fill= white] (Rho) at (7.3,0) {$\varrho_{t_0}$};
\node at (3.8,3) {$\vdots$};
\end{tikzpicture} 
\end{aligned}
\end{aligned}
\end{equation}
The mixed-branch and right Kirkwood-Dirac temporal state can be represented similarly by changing the horizontal ordering of the boxes in the tensor network representations.

The doubled  Kirkwood-Dirac temporal state $\Upsd$ can be represented as follows
\begin{equation}
\begin{aligned}
\begin{tikzpicture}[x=1cm,y=1cm]
%%%%%%%%%%%%%%%%%%%%%%%%%%%%%%%%%%%%%%%%%%%%%%%%%%%%%%%%%%%%
% Upper piece
%%%%%%%%%%%%%%%%%%%%%%%%%%%%%%%%%%%%%%%%%%%%%%%%%%%%%%%%%%%%
% four horizontal wires

\draw[line width=1pt] (0,7.5) -- (1.5,7.5);
\node at (-.3,7.5) {$t_n^L$};

\draw[line width=1pt] (0,6.8) -- (1.37,6.8);
\node at (-.3,6.8) {$t_n^R$};

\draw[line width=1pt] (0,6.1) -- (1.37,6.1);
\node at (-.35,6.1) {$t_{n-1}^L$};

\draw[line width=1pt] (0,5.4) -- (1.5,5.4);
\node at (-.35,5.4) {$t_{n-1}^R$};
% J[E_n]
\node[
    draw,
    line width=1pt,
    rectangle,
    minimum width=0.8cm,
    minimum height=2.5cm,
    fill=white
] (Jn-top) at (2,6.45)
{$J[\cE_n]$};

% outgoing wires
\draw[line width=1pt] (2.55,7.5) -- (4.8,7.5);
\draw[line width=1pt] (2.62,6.8) -- (4.8,6.8);
\draw[line width=1pt] (2.62,6.1) -- (4,6.1);
\node at (4.4,6.1) {$\cdots$};
\draw[line width=1pt] (4.8,6.1) -- (8.7,6.1);
\draw[line width=1pt] (2.55,5.4) -- (4,5.4);
\node at (4.4,5.4) {$\cdots$};
\draw[line width=1pt] (4.8,5.4) -- (8.7,5.4);

% swap on upper two lines
\draw[line width=1pt] (4.8,6.8) -- (5.55,7.5);
\draw[line width=1pt] (4.8,7.5) -- ++(0.35,-0.33);
\draw[line width=1pt] (5.55,6.8) -- ++(-0.35,0.33);

% continue wires after swap
\draw[line width=1pt] (5.55,7.5) -- (8.7,7.5);
\draw[line width=1pt] (5.55,6.8) -- (8.7,6.8);

%%%%%%%%%%%%%%%%%%%%%%%%%%%%%%%%%%%%%%%%%%%%%%%%%%%%%%%%%%%%
% vertical dots
%%%%%%%%%%%%%%%%%%%%%%%%%%%%%%%%%%%%%%%%%%%%%%%%%%%%%%%%%%%%

\node at (2.0,4.75) {$\vdots$};

%%%%%%%%%%%%%%%%%%%%%%%%%%%%%%%%%%%%%%%%%%%%%%%%%%%%%%%%%%%%
% Lower/main piece
%%%%%%%%%%%%%%%%%%%%%%%%%%%%%%%%%%%%%%%%%%%%%%%%%%%%%%%%%%%%

% six incoming wires with vertical spacing 0.7
\draw[line width=1pt] (0,4.2) -- (3.8,4.2);
\node at (-.3,4.2) {$t_2^L$};

\draw[line width=1pt] (0,3.5) -- (3.64,3.5);
\node at (-.3,3.5) {$t_2^R$};

\draw[line width=1pt] (0,2.8) -- (3.64,2.8);
\node at (-.3,2.8) {$t_1^L$};

\draw[line width=1pt] (0,2.1) -- (5.1,2.1);
\node at (-.3,2.1) {$t_1^R$};

\draw[line width=1pt] (0,1.4) -- (5.1,1.4);
\node at (-.3,1.4) {$t_0^L$};

\draw[line width=1pt] (0,0.7) -- (6.0,0.7);
\node at (-.3,0.7) {$t_0^R$};
%%%%%%%%%%%%%%%%%%%%%%%%%%%%%%%%%%%%%%%%%%%%%%%%%%%%%%%%%%%%
% J[E_2]
%%%%%%%%%%%%%%%%%%%%%%%%%%%%%%%%%%%%%%%%%%%%%%%%%%%%%%%%%%%%
\node[
    draw,
    line width=1pt,
    rectangle,
    minimum width=0.8cm,
    minimum height=2.5cm,
    fill=white
] (Jn) at (4.25,3.18)
{$J[\cE_2]$};

% wires leaving J[E_2]
\draw[line width=1pt] (4.8,4.2) -- (6.2,4.2);
\draw[line width=1pt] (4.85,3.5) -- (6.2,3.5);

\draw[line width=1pt] (4.85,2.8) -- (5.5,2.8);
%\draw[line width=1pt] (4.8,2.1) -- (5.5,2.1);

% swap after J[E_2]
% broken strand
\draw[line width=1pt] (6.2,4.2) -- ++(0.33,-0.33);
\draw[line width=1pt] (6.9,3.5) -- ++(-0.33,0.33);

% continuous strand
\draw[line width=1pt] (6.2,3.5) -- (6.9,4.2);

% wires after swap
\draw[line width=1pt] (6.9,4.2) -- (8.7,4.2);
\draw[line width=1pt] (6.9,3.5) -- (8.7,3.5);
%%%%%%%%%%%%%%%%%%%%%%%%%%%%%%%%%%%%%%%%%%%%%%%%%%%%%%%%%%%%
% J[E_1]
%%%%%%%%%%%%%%%%%%%%%%%%%%%%%%%%%%%%%%%%%%%%%%%%%%%%%%%%%%%%

\node[
    draw,
    line width=1pt,
    rectangle,
    minimum width=0.8cm,
    minimum height=2.5cm,
    fill=white
] (J1) at (5.7,1.7)
{$J[\cE_1]$};

% wires from J[E_1]
\draw[line width=1pt] (6.22,2.8) -- (8.7,2.8);
\draw[line width=1pt] (6.3,2.1) -- (8.7,2.1);

\draw[line width=1pt] (6.3,1.4) -- (8,1.4);
\draw[line width=1pt] (6.22,0.7) -- (8,0.7);

%%%%%%%%%%%%%%%%%%%%%%%%%%%%%%%%%%%%%%%%%%%%%%%%%%%%%%%%%%%%
% swap after J[E_1]
%%%%%%%%%%%%%%%%%%%%%%%%%%%%%%%%%%%%%%%%%%%%%%%%%%%%%%%%%%%%

%%%%%%%%%%%%%%%%%%%%%%%%%%%%%%%%%%%%%%%%%%%%%%%%%%%%%%%%%%%%
% rho_{t_0}
%%%%%%%%%%%%%%%%%%%%%%%%%%%%%%%%%%%%%%%%%%%%%%%%%%%%%%%%%%%%

\node[
    draw,
    line width=1pt,
    rectangle,
    minimum width=0.7cm,
    minimum height=0.55cm,
    fill=white
] (rho) at (7,1.4)
{$\varrho_{t_0}$};

%%%%%%%%%%%%%%%%%%%%%%%%%%%%%%%%%%%%%%%%%%%%%%%%%%%%%%%%%%%%
% bottom swap
%%%%%%%%%%%%%%%%%%%%%%%%%%%%%%%%%%%%%%%%%%%%%%%%%%%%%%%%%%%%
% swap on bottom two lines
% broken strand
\draw[line width=1pt] (8.0,1.4) -- ++(0.33,-0.33);
\draw[line width=1pt] (8.7,0.7) -- ++(-0.33,0.33);

% continuous strand
\draw[line width=1pt] (8.0,0.7) -- (8.7,1.4);
\end{tikzpicture}
\end{aligned}
\end{equation}
Based on this, it is easy to see $\Tr_L \Upsd=\Upsr$ and $\Tr_R \Upsd=\Upsl$.

The recursion in Eq.~\eqref{eq:recright} represents the temporal state as a matrix product operator (MPO) whose sites correspond to time slices and whose bond space is the operator algebra $\BB(\cH_{t_k})$ transported through the slice. This establishes a precise duality between a $(0+1)$-dimensional temporal problem and a $1$-dimensional spatial one: the temporal transfer operator $J[\cE]$ acts as a spatial MPO tensor, the temporal state as a one-dimensional chain wavefunction, and the temporal entanglement is precisely the entanglement entropy of that chain across a cut. For Floquet or time-translation-invariant processes, the tensor is identical at every slice, and the temporal state becomes a uniform MPO (up to the initial state, or in the long-time limit).

\subsection{Out-of-time-ordered correlators and folded contours: $2m$-branch temporal states}
\label{sec:mfold}

Out-of-time-ordered correlators (OTOCs) \cite{larkin1969otoc,Roberts2015otoc} play a crucial role in quantum gravity, information scrambling, and related areas. It has been shown that the Kirkwood--Dirac distribution is closely related to OTOCs, as the correlation functions can be expressed as through phase-space distributions \cite{Halpern2018otoc,Alonso2019otocKD}. 
A characterization of quantum chaos at the level of temporal states has also attracted considerable attention \cite{Dowling2023scrabling,Dowling2024chaosEnt}; it has been argued that quantum chaos is related to the entanglement structure of temporal states.
The doubled temporal state is associated with a two-branch contour.  More
generally, The OTOC require contours with
additional folds.  This inspired us to introduce $2m$-fold correlators and 
$2m$-branch temporal states.

Let $m\geq1$ and consider the $2m$-branch contour (see Figure~\ref{fig:2mcontour})
correlator
\begin{equation}\label{eq:mfoldcorr}
  C^{[2m]}\!\left(\{A_k^{(s)}\}\right)
  =\Tr\!\left[
    \varrho_{t_0}\,
    \mathfrak A_1\mathfrak A_2\cdots\mathfrak A_{2m}
  \right],
  \quad
  \mathfrak A_s=
  \begin{cases}
    A_0^{(s)}(t_0)A_1^{(s)}(t_1)\cdots A_n^{(s)}(t_n),
      &s\ \text{odd},\\[2pt]
    A_n^{(s)}(t_n)\cdots A_1^{(s)}(t_1)A_0^{(s)}(t_0),
      &s\ \text{even}.
  \end{cases}
\end{equation}
where for unitary dynamics,
$\mathcal U_{t_0,t}(\bullet)=U_{t_0,t}\bullet U_{t_0,t}^{\dagger}$ and
$A(t)=U_{t_0,t}^{\dagger}AU_{t_0,t}$.  Thus the contour consists of $m$
forward--backward pairs between $t_0$ and $t_n$.  The case $m=1$ is the
ordinary Schwinger--Keldysh contour and reproduces the doubled Kirkwood-Dirac
construction.  The correlators in Eq.~\eqref{eq:mfoldcorr} determine a
$2m$-branch temporal Bloch tensor $T^{\bm \mu, \bm \nu, \cdots}$ and hence a temporal state $\Ups^{[2m]}$.

\begin{figure}
    \centering
    \includegraphics[width=0.8\linewidth]{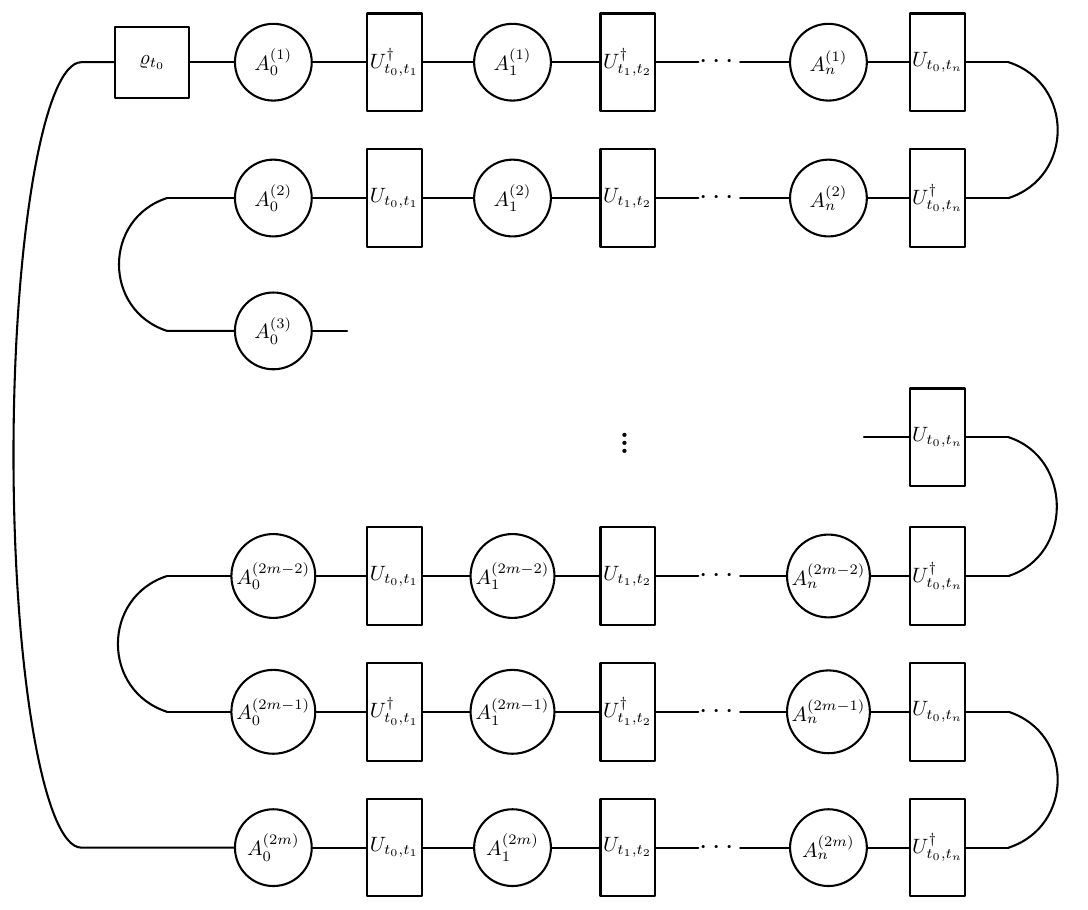}
    \caption{The $2m$-fold zigzag contour for correlation functions.}
    \label{fig:2mcontour}
\end{figure}

The $2m$-branch correlators can be transformed into two-branch correlators with out-of-time-ordered evolutions,
\begin{equation}
    C^{[2m]}\!\left(\{A_k^{(s)}\}\right)
  =\Tr\!\left[
 \mathfrak A_{m+1} \cdots  \mathfrak A_{2m}  \varrho_{t_0}\,
    \mathfrak A_1\cdots \fA_m
  \right],
\end{equation}
namely, we first evolve forward for one round, then backward for one round, then forward again, then backward again, and so on (this can also be generalized to quantum channel evolutions). Using this, we can obtain a recursive expression based on that of doubled Kirkwood--Dirac temporal state but with $n\times m$ time slots.
From the MPO representation of this temporal state (Eq.~\eqref{eq:MPOdoub}), folding the tensors for different values of $s$ yields a temporal state with effectively $n$ evolutions with each a package of evolution between $t_{k}$ to $t_{k+1}$ as different fold branch (this is also clear from Figure~\ref{fig:2mcontour}).
The corresponding temporal quasiprobability distribution is an out-of-time-ordered Kirkwood-Dirac distribution. All information about the processes, whether integrable or chaotic, is encoded in the informationally complete out-of-time-ordered distributions or equivalently in the corresponding temporal state.

For two times $t_0$ and $t_1=t$, let $V$ and $W$ be Hermitian operators.  Set
\begin{equation}
  A_0^{(1)}=A_0^{(3)}=V,
  \qquad
  A_1^{(1)}=A_1^{(3)}=W,
\end{equation}
and set all remaining insertions in Eq.~\eqref{eq:mfoldcorr} equal to
$\mathds{1}$.  Then
\begin{equation}
  C^{[4]}
  =\Tr\!\left[\varrho_{t_0}V W(t)V W(t)\right]
  =F(t),
\end{equation}
the OTOC in the ordering used here.  Thus $F(t)$ is a linear functional of
$\Ups^{[4]}$. This can naturally be generalized to the case with more fold branches.

\subsection{Quantumness of spatiotemporal quasiprobability distributions and temporality of spatiotemporal states}
\label{sec:quantumness}

The quantumness of a quasiprobability distribution can be characterized by its deviation from a genuine probability distribution. In the classical world, there are no uncertainty constraints on observables, so the phase-space distribution is a well-defined probability distribution. In the quantum setting, however, if one attempts to define a joint distribution for incompatible observables such as $x$ and $p$, negativity or complex values inevitably arise, yielding a quasiprobability distribution.
For the temporal Kirkwood--Dirac distribution
$Q_{\rm KD}(\alpha_n,\cdots,\alpha_0)$, we define its \emph{quantumness} (or \emph{temporal negativity}) as \cite{Jia2025TemporalKirkwoodDirac}
\begin{equation}\label{eq:quantumness}
    \mathfrak{N}\!\left[Q_{\rm KD}\right]
    =\frac{1}{2}\left(
    \sum_{\alpha_0,\cdots,\alpha_n}
    \left|Q_{\rm KD}(\alpha_n,\cdots,\alpha_0)\right|-1\right).
\end{equation}
When $Q_{\rm KD}$ takes values in $[0,1]$, so that it becomes a genuine real-valued probability distribution, the quantumness vanishes. For the temporal Margenau--Hill distribution, this quantity reduces to (twice) the negativity of the distribution. More generally, the quantumness of the temporal Kirkwood--Dirac distribution contains two contributions: the negativity of its real part, and the total absolute weight of its imaginary part.

From Remark~\ref{remk:equivalence}, it is clear that, for a quantum system, the quantumness of a temporal quasiprobability distribution can be characterized via the temporal state. We define the temporality of $\Upsilon$ as\footnote{Note that the equal-time reduced states are spatial density operators, which are positive semidefinite; hence, deviation from positive semidefiniteness can serve as a measure of temporality. See \cite{Song2024causal}.}
\begin{equation}\label{eq:temporality}
 \mathfrak{T}(\Ups):=\frac{1}{2}( \|\Ups\|_1-1)
 =\frac{1}{2}(\Tr\sqrt{\Ups^\dagger\Ups}-1),
 \qquad
 \mathfrak L_T(\Ups):=\log\|\Ups\|_1
 =\log\!\bigl(1+2\mathfrak{T}(\Ups)\bigr).
\end{equation}
We call $\mathfrak{T}$ the \emph{temporality}, and
$\mathfrak L_T$ the \emph{logarithmic temporality}.  The name refers specifically
to failure of positivity on the tensor product of time slots.  It should not be
confused with temporal entanglement, which concerns operator Schmidt
correlations across a temporal cut and can be nonzero even for a positive
history state.\footnote{Every one-time marginal of a valid temporal state is an
ordinary positive density operator.  Thus nonpositivity can occur only in the
multi-time assembly and supplies a natural witness of temporality; compare the
causal classification of Ref.~\cite{Song2024causal}.}

\begin{theorem}[Quantumness is bounded by the temporality]
\label{thm:quantumness}
Let $\Ups$ be the temporal state  generating $Q_{\rm KD}$ through the Born rule $Q_{\rm KD}(a_n,\dots,a_0)=\Tr[(\Pi_{a_n}\otimes\cdots\otimes\Pi_{a_0})\Ups]$ for rank-one projective insertions. Then for every choice of local bases
\begin{equation}
    \mathfrak{N}[Q_{\rm KD}]\le \mathfrak{T}(\Upsilon).
\end{equation} 
Moreover, for any observables with $\|O_{A_k}\|_\infty\le 1$,
\begin{equation}\label{eq:holderbound}
  \left|\left\langle \{O_{A_n}(t_n),\dots,O_{A_0}(t_0)\}\right\rangle_W\right|
  \;\le\; (2\mathfrak{T}(\Upsilon)+1)\prod_k\|O_{A_k}\|_{\infty}.
\end{equation}
\end{theorem}

\begin{proof}
Choose phases $e^{i\theta_{a}}$ with $e^{-i\theta_a}Q_{\rm KD}(a)=|Q_{\rm KD}(a)|$ for each multi-index $a=(a_n,\dots,a_0)$, and set $X=\sum_a e^{-i\theta_a}\Pi_{a_n}\otimes\cdots\otimes\Pi_{a_0}$. Since the $\Pi$'s are rank-one projectors onto product bases, $X$ is a unitary (diagonal in that product basis), hence $\|X\|_\infty=1$. Then
$\sum_a|Q_{\rm KD}(a)|=\Tr[X\Ups]\le\|X\|_\infty\|\Ups\|_1=\|\Ups\|_1$ by H\"older's inequality for Schatten norms. The normalization $\sum_aQ_{\rm KD}(a)=\Tr\Ups=1$ gives the second statement. Equation~\eqref{eq:holderbound} is H\"older's inequality applied to $X=\bigotimes_kO_{A_k}$, using $\|\bigotimes_kO_{A_k}\|_\infty=\prod_k\|O_{A_k}\|_\infty$.
\end{proof}

Theorem~\ref{thm:quantumness} corrects and sharpens the naive bound $\langle\cdots\rangle_W/\prod_k\|O_{A_k}\|\le\|\Ups\|$: the correct norm on the right-hand side is the \emph{trace} (Schatten-$1$) norm of the temporal state, which for non-Hermitian $\Ups$ is the sum of singular values. Since $\|\Ups\|_1=1$ if and only if $\Ups$ is a genuine density operator on the temporal Hilbert space, the excess $\|\Ups\|_1-1$ is an ordering-independent measure of temporal nonclassicality --- the exact temporal analogue of Wigner negativity, and the quantity that upper-bounds the quantum advantage obtainable from the process in metrological and thermodynamic tasks.

\begin{proposition}
\label{prop:temporality-properties}
The temporality satisfies the following properties:
\begin{enumerate}
\item \emph{Non-negativity and faithfulness to positivity.}
\begin{equation}\label{eq:temporality-faithful}
 \mathfrak{T}(\Ups)\ge0,
 \qquad
 \mathfrak{T}(\Ups)=0\quad\Longleftrightarrow\quad \Ups\ge0.
\end{equation}
Thus temporality vanishes exactly when the normalized temporal operator is a
genuine density operator.  Positivity implies nonnegative Born probabilities
for every history test.  The converse need not follow from one restricted
family of local product tests, because such tests need not detect every
nonpositive direction of a multipartite operator.

\item \emph{Convexity and stability.}
\begin{align}
 \mathfrak{T}\!\left(\lambda\Ups+(1-\lambda)\Omega\right)
 &\le \lambda\mathfrak{T}(\Ups)+(1-\lambda)\mathfrak{T}(\Omega),
 \label{eq:temporality-convexity}
\end{align}
In particular, the temporality is convex in the initial state
and in each channel separately.  With all other data fixed,
\begin{align}
 \mathfrak{T}\!\left(\Ups[\lambda\varrho_{t_0}+(1-\lambda)\omega_{t_0}]\right)
 &\le \lambda\mathfrak{T}(\Ups[\varrho_{t_0}])
 +(1-\lambda)\mathfrak{T}(\Ups[\omega_{t_0}]),\label{eq:temp-convex-state}\\
 \mathfrak{T}\!\left(\Ups[\ldots,\lambda\cE_j+(1-\lambda)\mathcal K_j,\ldots]\right)
 &\le \lambda\mathfrak{T}(\Ups[\ldots,\cE_j,\ldots])
 +(1-\lambda)\mathfrak{T}(\Ups[\ldots,\mathcal K_j,\ldots]).
 \label{eq:temp-convex-channel}
\end{align}
The qualification ``separately'' is essential: dependence on several time-step
channels is multilinear.  A probabilistic mixture of complete processes is,
of course, convex and also satisfies Eq.~\eqref{eq:temporality-convexity}.

\item \emph{Tensor products and logarithmic additivity.}
For independent processes whose initial states, dynamics, and temporal
insertions factorize $\Ups_1\otimes\Ups_2$, and
\begin{align}
 \mathfrak{T}(\Ups_1\otimes\Ups_2)
 &=\mathfrak{T}(\Ups_1)+\mathfrak{T}(\Ups_2)
 +2\mathfrak{T}(\Ups_1)\mathfrak{T}(\Ups_2),
 \label{eq:temporality-product}\\
 \mathfrak L_T(\Ups_1\otimes\Ups_2)
 &=\mathfrak L_T(\Ups_1)+\mathfrak L_T(\Ups_2).
 \label{eq:log-temporality-additive}
\end{align}

\item \emph{Branch and basis symmetries.}
Temporality is invariant under unitary conjugation, Hermitian conjugate, and full transpose:
\begin{equation}\label{eq:temporality-symmetries}
 \mathfrak{T}(U\Ups U^\dagger)
 =\mathfrak{T}(\Ups^\dagger)
 =\mathfrak{T}(\Ups^{\mathsf T})
 =\mathfrak{T}(\Ups).
\end{equation}
Consequently the left and right (also for mixed-branch and its Hermitian conjugate) Kirkwood--Dirac states have equal temporality.  Moreover,
for the Margenau--Hill state
$\Ups^{\rm MH}=(\Upsr+\Upsl)/2$,
\begin{equation}\label{eq:mh-temporality}
 \mathfrak{T}(\Ups^{\rm MH})\le\mathfrak{T}(\Upsr)
 =\mathfrak{T}(\Upsl),
\end{equation}
so Hermitianization cannot increase temporality.

\item \emph{Temporality and distance to the positive semidefinite operator space.}
If $\Ups=\Ups^\dagger$ and
$\Ups=\Ups_+-\Ups_-$ with $\Ups_+=(|\Ups|+\Ups)/2$ and $\Ups_-=(|\Ups|-\Ups)/2$, then
\begin{equation}\label{eq:temporality-jordan}
 \mathfrak{T}(\Ups)
 =\Tr\Ups_-
 =\frac{1}{2}
 \inf_{\rho\ge0,\,\Tr\rho=1}\|\Ups-\rho\|_1.
\end{equation}
For a general non-Hermitian temporal state,
\begin{equation}\label{eq:temporality-distance}
 \mathfrak{T}(\Ups)
 \le \frac{1}{2} \inf_{\rho\ge0,\,\Tr\rho=1}\|\Ups-\rho\|_1,
\end{equation}
but equality need not hold.  Thus for Hermitian temporal states temporality is
exactly the trace-distance from the positive semidefinite operators; for non-Hermitian states it
is a computable lower bound on that violation.

\item \emph{Monotonicity under loss of temporal resolution.}
For every completely positive trace-preserving map $\Phi$ on the temporal
operator space,
\begin{equation}\label{eq:temporality-monotonicity}
 \mathfrak{T}\!\left(\Phi(\Ups)\right)\le\mathfrak{T}(\Ups).
\end{equation}
This includes local noise on any collection of time slots, dephasing, classical
relabeling, etc.  When applied to partial trace,  we see every marginal temporal state obeys
\begin{equation}\label{eq:temporality-marginal}
 \mathfrak{T}\!\left(\Tr_{\mathsf A^c}\Ups\right)\le\mathfrak{T}(\Ups).
\end{equation}

\item \emph{Operational witnesses.}
For every local temporal frame $\mathcal{B}=\{F_{\alpha_k}\}_{k=0}^n$,
\begin{equation}\label{eq:temporality-envelope}
 \mathfrak N[Q_{\mathcal B}]\le\mathfrak{T}(\Ups),
 \qquad
 \sup_{\mathcal B\in{\rm product\ bases}}\mathfrak N[Q_{\mathcal B}]
 \le\mathfrak{T}(\Ups).
\end{equation}
If arbitrary global temporal frame  on the full temporal Hilbert space are
allowed, the bound is tight:
\begin{equation}\label{eq:temporality-global-witness}
 \sup_{\mathcal B\in{\rm global\ bases}}\mathfrak N[Q_{\mathcal B}]
 =\mathfrak{T}(\Ups).
\end{equation}
The globally optimal basis can be entangled across time slots and therefore
need not correspond to an implementable sequence of local measurements.
\end{enumerate}
\end{proposition}

The proofs are provided in Section~\ref{app:proofs} in Appendix.

%%%%%%%%%%%%%%%===================
\section{Two-time temporal state}
\label{sec:TwoTimeState}
The study of two-time temporal states has a longer history in quantum information theory and exhibits a rich structure; see, e.g., \cite{Leifer2013toward,Parzygnat2023pdo,fitzsimons2015quantum,Song2025twotime,zhao2018geometry}. In quantum field theory, the study of two-time temporal states is naturally related to two-time correlation functions (Green functions), which serve as the fundamental building blocks for other observables. Thus, it is worthwhile to discuss the two-time temporal state in greater detail.
For a two-time process $0\rightarrow t$, we consider the temporal state $\Upsilon_{t,0}$ and define the correlation function on $\mathcal{H}_t\otimes\mathcal{H}_0$ as
\begin{equation}
    \langle A(t)B(0)\rangle_{\Upsilon_{t,0}}:=\Tr[ (A\otimes B)\Upsilon_{t,0}].
\end{equation}
The corresponding Green function is then $G_{AB}^{\Upsilon_{t,0}}(t)=-i\langle A(t)B(0)\rangle_{\Upsilon_{t,0}}$. The temporal kernel is the Fourier transform of this Green function. In this section, we give a more general expression for two-time state and discuss its applications.

\subsection{A more general unification of two-time temporal state}

Recall that the two-time left and right Kirkwood--Dirac temporal states take the form
\begin{equation}
    \Upsl_{t_1,t_0}^{\rm KD}=\sum_{i,j}\mathcal{E}_{t_0,t_1}(|i\rangle \langle j|)\otimes |j\rangle \langle i|\varrho_{t_0},\quad \Upsr_{t_1,t_0}^{\rm KD}=\sum_{i,j}\mathcal{E}_{t_0,t_1}(|i\rangle \langle j|)\otimes \varrho_{t_0} |j\rangle \langle i|.
\end{equation}
They are also referred to as the left and right Bloom states in \cite{Parzygnat2023pdo,fullwood2023quantum}; see also \cite{lie2025probingquantumstatesspacetime,Jia2025TemporalKirkwoodDirac}.
The two-time left/right Margenau–Hill state coincides with pseudo-density operators \cite{fitzsimons2015quantum,Parzygnat2023pdo,fullwood2023quantum,Jia2025TemporalKirkwoodDirac}.

The left and right Kirkwood--Dirac temporal states can be generalized to a one-parameter family involving both left and right parts \cite{Parzygnat2023pdo}:
\begin{equation}\label{eq:TemporalTwos}
    \Ups^{(s)}_{t_1,t_0} =\sum_{i,j}\mathcal{E}_{t_0,t_1}(|i\rangle \langle j|)\otimes \varrho_{t_0}^{s}|j\rangle \langle i|\varrho_{t_0}^{1-s}=(\I\otimes \varrho^{s})J[\cE_{t_0,t_1}](\I\otimes \varrho^{1-s}) ,  
\end{equation}
where $s\in [0,1]$. It is clear that $\Tr_{t_0}\Ups^{(s)}_{t_1,t_0}=\mathcal{E}_{t_0,t_1}(\varrho_{t_0})=\varrho_{t_1}$ and $\Tr_{t_1}\Ups^{(s)}_{t_1,t_0}=\varrho_{t_0}$.
Note that we can equivalently write
\begin{equation}\label{eq:TemStas}
    \Ups^{(s)}_{t_1,t_0} =\sum_{i,j}\mathcal{E}_{t_0,t_1}(\varrho_{t_0}^{1-s}|i\rangle \langle j|\varrho_{t_0}^{s})\otimes |j\rangle \langle i|.
\end{equation}
This family proves useful in the study of quantum conditional expectations and quantum Bayesian inference \cite{Parzygnat2023pdo}. In the special case $s=1/2$, it reduces to the Leifer--Spekkens causal state \cite{Leifer2013toward}.

Given a (quasi)probability\footnote{We emphasize that using (even complex-valued) quasiprobability distribution, we can also obtain temporal state that gives the correct marginals.} distribution $\mathbb{P}(s)$, we may form the mixture
\begin{equation}\label{eq:TwoTimePs}
  \Ups^{\mathbb{P}}_{t_1,t_0}=  \int_{[0,1]}\Ups^{(s)}_{t_1,t_0} d\mathbb{P}(s) ,
\end{equation}
which provides a more general unified definition of a two-time temporal state than the one in Theorem~\ref{theorem:UnificationSingle}, where only a mixture of the $s=0$ and $s=1$ states was considered.

\begin{example}[closed-system two-time temporal state]\label{exp:pure}
    For unitary evolution $U$ which maps the initial state $|\psi\rangle$ to $U|\psi\rangle=|\varphi\rangle$, we have
    \begin{equation}
        \Upsilon^{(0)}=\sum_j |\varphi\rangle \langle j|U^{\dagger} \otimes |j\rangle \langle \psi|,\quad  \Upsilon^{(1)} =\sum_i U|i\rangle \langle \varphi | \otimes |\psi\rangle \langle i|.
    \end{equation}
    These are simply coupled transition matrices.
    For $r\neq 0,1$, we have $\Upsilon^{(s)}=|\varphi\rangle \langle \varphi|\otimes |\psi\rangle \langle \psi|$, which exhibits no temporal correlation. Thus, a general two-time state for a closed system is a probabilistic mixture of the form
    \begin{equation}
        \Upsilon^{\mathbb{P}}=\mathbb{P}(0) \Upsilon^{(0)}
        +\mathbb{P}(1)\Upsilon^{(1)}
        +\mathbb{P}(1/2)\Upsilon^{( 1/2)}.
    \end{equation}
\end{example}

\begin{remark}[Multi-time $s$-parameterized temporal state]
For a fixed mixed branch
$\ttB=(\ldots,t_1^R,t_0^L)$,
we define the temporal link product of the corresponding Jamio{\l}kowski operators by
\begin{equation}
    \mathcal{J}_{\ttB}
    :=
    J[\cE_{t_{k-1},t_k}]
    \star_{\ttB}\cdots\star_{\ttB}
    J[\cE_{t_0,t_1}].
\end{equation}
The associated $s$-parameterized temporal state is then defined as
\begin{equation}
    \Ups^{(s)}_{\ttB}
    :=
    \bigl(\I\otimes\cdots\otimes\I\otimes\varrho^s\bigr)
    \mathcal{J}_{\ttB}
    \bigl(\I\otimes\cdots\otimes\I\otimes\varrho^{1-s}\bigr).
\end{equation}
More generally, given a probability measure $\mathbb{P}$ on the parameter $r$, we may consider the probabilistic mixture
\begin{equation}
    \Ups^{\mathbb{P}}_{\ttB}
    :=
    \int \Ups^{(s)}_{\ttB}\,d\mathbb{P}(s).
\end{equation}
Here the mixed branch $\ttB$ is kept fixed throughout the construction.
We can then introduce a probabilistic distribution $p(\ttB)$ for branches, then obtain 
\begin{equation}
    \Ups=\sum_{\ttB} p(\ttB)  \Ups^{\mathbb{P}}_{\ttB}
\end{equation}
which gives a very general unification of all spatiotemporal state.
\end{remark}

\subsection{Petz time reversal of temporal states}

For a quantum channel $\mathcal{E}$ and an input state $\varrho$, let $\omega=\mathcal{E}(\varrho)$. Define the two-time state $\Upsilon^{(s)}\big[\varrho,\cE\big]$ for the pair $(\cE,\varrho)$. One may naturally regard time reversal as an interchange of the two time steps, which is implemented by the swap operator
\begin{equation}
  \mathbb{S}\,\Ups^{(s)}\big[\varrho,\cE\big]\,\mathbb{S}.  
\end{equation}
The corresponding initial state for this reversed process should be $\omega$; the question is then to find a channel that realizes this as a temporal state. For $s=1/2$, it turns out that the answer is given by the Petz recovery map \cite{Petz1986petzmap}, which for a fixed state $\varrho$ and quantum channel $\mathcal{E}$ is defined as 
\begin{equation}
    \mathcal{R}_{\varrho,\cE}(\bullet)=\varrho^{1/2}\cE^{\dagger}({\mathcal{E}(\varrho)}^{-1/2}\bullet\,{\mathcal{E}(\varrho)}^{-1/2})\varrho^{1/2}.
\end{equation}
We observe that $\mathcal{R}_{\varrho,\cE}(\mathcal{E}(\varrho))=\varrho$. We can then construct a temporal state for which the initial state is $\omega$ and
\begin{equation}
\Ups^{(1/2)}\big[\omega,\mathcal{R}_{\varrho,\cE}\big]=\mathbb{S}\,\Ups^{(1/2)}\big[\varrho,\cE\big]\,\mathbb{S}.
\end{equation}
This construction can be generalized to more general temporal states, yielding the following.

\begin{theorem}[Petz time reversal]\label{thm:petz}
For $s \in [0,1]$, let \( \Ups^{(s)} \) be the temporal state such that
\begin{equation}
  \Tr\!\left[(A\otimes B)\,\Ups^{(s)}\right]
  =
  \Tr\!\left[A\,\cE\bigl(\varrho^{1-s}B\varrho^{s}\bigr)\right].
\end{equation}
Let \( \omega = \cE(\varrho) \) and define the \(s\)-weighted Petz recovery map 
\begin{equation}
  \mathcal{R}^{(s)}_{\varrho,\cE}(\bullet)
  =
  \varrho^{s}\,
  \cE^{\dagger}\!\left(
    \omega^{-(1-s)}\,\bullet\,\omega^{-s}
  \right)
  \varrho^{1-s},
\end{equation}
with inverses taken as Moore–Penrose pseudo-inverses on the support of \(\omega\).
Then, for every \( s \in [0,1] \),
\begin{equation}
    \Ups^{(s)}\bigl[\omega,\mathcal{R}^{(s)}_{\varrho,\cE}\bigr]
    =
    \mathbb{S}\,
    \Ups^{(s)}\bigl[\varrho,\cE\bigr]\,
    \mathbb{S}
\end{equation}
i.e., every member of the one-parameter temporal state family is covariant under its corresponding \(s\)-weighted Petz time reversal, up to exchanging the two temporal slots.
Also note that $\mathcal{R}^{(s)}_{\varrho,\cE}$ is a CPTP map only when $s=1/2$, meaning $\Ups^{(1/2)}$ admits a physical Petz reversal.
\end{theorem}

\begin{proof}
It suffices to show that both sides yield the same expectation value for arbitrary Hermitian \(A\) and \(B\). 
For the left-hand side, using the trace formula,
\begin{align}
  \Tr\!\bigl[(A\otimes B)\,\Ups^{(s)}[\omega,\mathcal{R}^{(s)}]\bigr]
  &=
  \Tr\!\bigl[A\,\mathcal{R}^{(s)}\!\left(\omega^{1-s}B\omega^{s}\right)\bigr] \\
  &=
  \Tr\!\Bigl[
    A\,\varrho^{s}\,
    \cE^{\dagger}\!\left(
      \omega^{-(1-s)}
      \bigl(\omega^{1-s}B\omega^{s}\bigr)
      \omega^{-s}
    \right)
    \varrho^{1-s}
  \Bigr] \\
  &=
  \Tr\!\bigl[
    A\,\varrho^{s}\,
    \cE^{\dagger}(B)\,
    \varrho^{1-s}
  \bigr] \\
  &=
  \Tr\!\bigl[
    \varrho^{1-s} A \varrho^{s}\,
    \cE^{\dagger}(B)
  \bigr] \\
  &=
  \Tr\!\bigl[
    B\,\cE\!\left(\varrho^{1-s} A \varrho^{s}\right)
  \bigr].
\end{align}
For the right-hand side, the swap operator \(\mathbb{S}\) exchanges the two Hilbert-space slots, so
\begin{align}
  \Tr\!\bigl[(A\otimes B)\,\mathbb{S}\,\Ups^{(s)}[\varrho,\cE]\,\mathbb{S}\bigr]
  &=
  \Tr\!\bigl[(B\otimes A)\,\Ups^{(s)}[\varrho,\cE]\bigr] \\
  &=
  \Tr\!\bigl[
    B\,\cE\!\left(\varrho^{1-s} A \varrho^{s}\right)
  \bigr].
\end{align}
The two expressions are identical. Since this holds for all \(A\) and \(B\), the operator equality follows.
\end{proof}

\subsection{KMS temporal states}
\label{subsec:KMS}

Now let us consider an application in thermal states.
Recall that the KMS condition provides a rigorous mathematical characterization of thermal equilibrium in terms of correlation functions \cite{kubo1957statistical,martin1959KMStheory,Haag1967KMS}. 
In the Heisenberg picture, for a thermal state \(\varrho = e^{-\beta H}/Z\), it reads
\[
\langle A(t)B(0)\rangle_{\beta} = \langle B(0) A(t+i\beta)\rangle_{\beta},
\]
where \(A(t)=e^{iHt}Ae^{-iHt}\).
This can be extended to $s$-parametrized temporal states.

For thermal state $\varrho=e^{-\beta H}/Z$ and unitary evolution $e^{-iHt}(\bullet)e^{iHt}$, consider the \(s\)-parameterized temporal state $ \Ups^{(s)}_{t_1,t_0} =\sum_{i,j}\mathcal{E}_{t_0,t_1}(|i\rangle \langle j|)\otimes \varrho_{t_0}^{s}|j\rangle \langle i|\varrho_{t_0}^{1-s}$.
We call $\Ups^{(s)}$ \emph{KMS temporal states}.
The family satisfies
\begin{equation}\label{eq:kms-adjoint-reflection}
\bigl(\Ups^{(s)}_{t,0}\bigr)^\dagger
=
\Ups^{(1-s)}_{t,0}.
\end{equation}
Hence the midpoint $\Ups^{(1/2)}_{t,0}$ is Hermitian.
The resulting temporal correlation functions is
\begin{equation}
\begin{aligned}
    F^{(s)}_{AB}(t)=&\langle A(t)B(0)\rangle_{\Ups^{(s)}} \\
    =  & \Tr [(A\otimes B) {\Upsilon^{(s)}}] \\
    =&
\Tr\!\left[ A\, e^{-iHt}\varrho^{1-s}B\varrho^{s}e^{iHt}\right]=\langle B(0)A(t+i s\beta)\rangle_{\beta}.
\end{aligned}
\label{eq:ims}
\end{equation}
The endpoints are consequently
\begin{equation}
F^{(0)}_{AB}(t)=\langle B(0) A(t)\rangle_\beta,
\qquad
F^{(1)}_{AB}(t)=\langle A(t)B(0)\rangle_\beta.
\end{equation}
Let $H|n\rangle=E_n|n\rangle$ and
$p_n=e^{-\beta E_n}/Z(\beta)$. Then
\begin{equation}\label{eq:kms-lehmann-s}
F^{(s)}_{AB}(t)
=
\sum_{m,n}
p_m^s p_n^{1-s}
A_{mn}B_{nm}\,
e^{i(E_m-E_n)t}.
\end{equation}
With the Fourier convention
$\widetilde F^{(s)}_{AB}(\omega)
=\int_{\mathbb R}dt\,e^{i\omega t}F^{(s)}_{AB}(t)$,
this yields, as an identity of spectral distributions,
\begin{equation}\label{eq:kms-s-frequency}
\widetilde F^{(s)}_{AB}(\omega)
=
e^{s\beta\omega}\widetilde F^{(0)}_{AB}(\omega)
=
e^{-(1-s)\beta\omega}\widetilde F^{(1)}_{AB}(\omega).
\end{equation}
In particular, we obtain the KMS condition
\begin{equation}\label{eq:kms-detailed-balance}
\widetilde F^{(1)}_{AB}(\omega)
=
e^{\beta\omega}\widetilde F^{(0)}_{AB}(\omega).
\end{equation}
For $B=A^\dagger$, every coefficient in
Eq.~\eqref{eq:kms-lehmann-s} is nonnegative, so
$\widetilde F^{(s)}_{A A^\dagger}(\omega)$ is a positive measure.

Using the notation introduced in Section~\ref{sec:keldysh}
\begin{align}
       G^{>}(\omega)
= \int_{-\infty}^{\infty} dt \, e^{i\omega t} \, G^{>}(t),
\qquad 
G^{>}(t) = -i\langle A(t)B(0)\rangle_{\beta}=-i F^{(1)}_{AB}(t),\\
    G^{<}(\omega)
= \int_{-\infty}^{\infty} dt \, e^{i\omega t} \, G^{<}(t),
\qquad 
G^{<}(t) = -i\langle B(0)A(t)\rangle_{\beta}=-i F^{(0)}_{AB}(t).
\end{align}
With this convention, the KMS condition for $\Ups^{(s)}$ reads \(G^{<}(\omega) = e^{-s\beta\omega}G^{>}(\omega)\).
Let \(G^{\rm sp}(\omega)=G^{>}(\omega)-G^{<}(\omega)\) be the spectral function and \(G^{K}(\omega)=G^{>}(\omega)+G^{<}(\omega)\) the Keldysh function. Then by KMS condition, we have
\begin{align}\label{eq:kmskernel}
  G^{\rm sp}(\omega) =(1-e^{\beta s \omega}) G^{>}(\omega),\\
   G^{\rm K}(\omega) =(1+e^{\beta s \omega}) G^{>}(\omega).
\end{align}

The thermal unitary setting permits an explicit characterization
of the temporal spectrum.
Let $\mathbb S$ be swap operator, we have
\begin{equation}\label{eq:kms-swap-form}
\Ups^{(s)}_{t,0}
=
(U_t\otimes I)
(\varrho^{1-s}\otimes\varrho^s)\mathbb S
(U_t^\dagger\otimes I).
\end{equation}
In particular, its trace is one and both one-time marginals are
$\varrho$.
The entire family is similar to its
Hermitian midpoint:
\begin{equation}\label{eq:kms-midpoint-similarity}
\Ups^{(s)}_{t,0}
=
L_s\Ups^{(1/2)}_{t,0}L_s^{-1},
\qquad
L_s=\varrho^{1/2-s}\otimes I.
\end{equation}
Consequently, every member is diagonalizable and has a real
spectrum, even though it need not be normal.
Explicitly,
\begin{equation}\label{eq:kms-temporal-spectrum}
\operatorname{spec}\Ups^{(s)}_{t,0}
=
\{p_i\}_{i=1}^{d}
\;\cup\;
\bigl\{
+\sqrt{p_i p_j},-\sqrt{p_i p_j}
\bigr\}_{i<j},
\end{equation}
counting multiplicities.
To see this directly, remove the unitary conjugation in
Eq.~\eqref{eq:kms-swap-form}. Each vector $|ii\rangle$ has
eigenvalue $p_i$, while on
$\operatorname{span}\{|ij\rangle,|ji\rangle\}$ the operator has
the matrix
\begin{equation}
\begin{pmatrix}
0 & p_i^{1-s}p_j^s\\
p_j^{1-s}p_i^s & 0
\end{pmatrix},
\end{equation}
whose eigenvalues are $\pm\sqrt{p_i p_j}$.
Thus the eigenvalues are independent of both $s$ and $t$.
For a faithful state with $d>1$, negative eigenvalues are
unavoidable, including at the Hermitian midpoint.

In contrast, the singular values depend on $s$.
Equation~\eqref{eq:kms-swap-form} gives
\begin{equation}\label{eq:kms-magnitude}
\bigl|\Ups^{(s)}_{t,0}\bigr|
=
\varrho^s\otimes\varrho^{1-s},
\end{equation}
and hence
\begin{equation}\label{eq:kms-trace-norm}
\bigl\|\Ups^{(s)}_{t,0}\bigr\|_1
=
\Tr\varrho^s\,\Tr\varrho^{1-s}
=
\frac{Z(s\beta)Z((1-s)\beta)}{Z(\beta)}.
\end{equation}
For a non-maximally-mixed faithful $\varrho$, the operator is
normal precisely at $s=1/2$.
Indeed, normality of each two-dimensional block above requires
$p_i^{1-s}p_j^s=p_j^{1-s}p_i^s$ for every pair $i,j$.
Cauchy--Schwarz inequality (applied to spectral eigenvalues of $\varrho^{1-s}$ and $\varrho^s$) further implies
\begin{equation}\label{eq:kms-minimal-temporality}
\bigl\|\Ups^{(s)}_{t,0}\bigr\|_1
\ge
\bigl(\Tr\sqrt{\varrho}\bigr)^2
=
\bigl\|\Ups^{(1/2)}_{t,0}\bigr\|_1.
\end{equation}
Equality holds only at $s=1/2$, unless
$\varrho=I/d$, in which case the family is independent of $s$.
Thus the midpoint minimizes the temporality within the KMS family.
The spectral absolute weight is independent of $s$:
\begin{equation}
\sum_j|\lambda_j|
=
\bigl(\Tr\sqrt{\varrho}\bigr)^2.
\end{equation}
Accordingly, the excess trace norm away from the midpoint is
entirely associated with nonnormality (see Eq.~\eqref{eq:temporality-decomposition} for definition of $\mathfrak T_{\rm nn}$):
\begin{equation}\label{eq:kms-nonnormal-excess}
\mathfrak T_{\rm nn}\bigl(\Ups^{(s)}_{t,0}\bigr)
=
\frac{
\Tr\varrho^s\,\Tr\varrho^{1-s}
-
(\Tr\sqrt{\varrho})^2
}{2}.
\end{equation}
The KMS family therefore separates a fixed signed spectrum from
an $s$-dependent nonnormal contribution to temporality.

When $s=1/2$, the ordinary Petz reverse of
$\cE_t=U_t(\bullet)U_t^{\dagger}$ is
\begin{align}
\mathcal R_{\varrho,\cE_t}(X)
&=
\varrho^{1/2}
\cE_t^\dagger\bigl(
\varrho^{-1/2}X\varrho^{-1/2}
\bigr)
\varrho^{1/2}
\nonumber\\
&=
U_t^\dagger XU_t
=
\cE_{-t}(X)
\label{eq:kms-petz-inverse}
\end{align}
which is nothing but the backward evolution.
The corresponding temporal operators obey (see Section~\ref{app:proofs} in Appendix for proof.)
\begin{equation}\label{eq:kms-petz-reflection}
\mathbb S\,
\Ups^{(s)}[\varrho,\cE_t]\,
\mathbb S
=
\Ups^{(1-s)}[\varrho,\cE_{-t}].
\end{equation}
Petz reversal therefore exchanges the temporal slots and reflects
$s\mapsto1-s$ across the midpoint of the KMS strip.
The ordinary Petz map implements this relation for every $s$;
the midpoint is distinguished because its ordering parameter is
unchanged by reflection.

\section{Vectorized spatiotemporal state, entanglement and entropy}
\label{sec:STentropy}

There are two complementary ways to understand spatiotemporal entanglement through the entropy of a temporal state. (i) One may regard the spatiotemporal state of interest as a reduced state of a larger spatiotemporal state. From this perspective, the entropy of the reduced state characterizes the spatiotemporal entanglement between the chosen spatiotemporal region and its complement. One may further investigate the entanglement structure within the chosen region itself. (ii) Alternatively, one may take the temporal state associated with a chosen spatiotemporal region as the given state and focus directly on the entanglement structure within that region. In this section, we adopt the second point of view and postpone the more general discussion of the first perspective to the next section.
We will focus on the temporal case, as the generalization to the spatiotemporal setting is straightforward.

\subsection{Vectorized spatiotemporal state and exact normalization}

A spatiotemporal state supplies a natural way to define entanglement and entropy across arbitrary regions of spacetime. A  spatiotemporal state is an operator rather than, in general, a positive density matrix. Its direct partial trace gives the physical state at the retained times, not a measure of correlations with the discarded
times, and its eigenvalue entropy can be complex. Vectorization resolves both problems\footnote{The general theory of operator entanglement is discussed, e.g., in \cite{Zanardi2001operatorEnt,Prosen2007operatorEnt}.}: every nonzero spatiotemporal operator defines a normalized pure state on a doubled operator Hilbert space. Ordinary entanglement and mutual information
can then be applied without any positivity assumption on the original temporal operator. 

Fix the column-vectorization convention
\begin{equation}
 |X\rangle\!\rangle=\sum_{ij}X_{ij}|i\rangle\otimes|j\rangle,
 \qquad
 \langle\!\langle X|Y\rangle\!\rangle=\Tr(X^\dagger Y).
\end{equation}
For the Hilbert--Schmidt basis, introduce the normalized vectorization
\begin{equation}
 |\widehat{\sigma}_\mu\rangle\!\rangle
 :=\frac{1}{\sqrt d}|\sigma_\mu\rangle\!\rangle .
\end{equation}
The normalized vectorization of the left Kirkwood--Dirac temporal state is
\begin{align}
 |\widehat{\Upsl}_{t_n,\cdots ,t_0}\rangle\!\rangle
 &:=
 \mathcal N_L\frac{1}{d^{\,n+1}}
 \sum_{\mu_n,\ldots,\mu_0}T_L^{\mu_n\cdots\mu_0}
 |\sigma_{\mu_n}\rangle\!\rangle\otimes\cdots\otimes
 |\sigma_{\mu_0}\rangle\!\rangle \label{eq:normalized-vectorization}\\
 &=\frac{1}{\sqrt{\sum_{\boldsymbol\mu}|T_L^{\boldsymbol\mu}|^2}}
 \sum_{\boldsymbol\mu}T_L^{\boldsymbol\mu}
 |\widehat\sigma_{\mu_n}\rangle\!\rangle\otimes\cdots\otimes
 |\widehat\sigma_{\mu_0}\rangle\!\rangle ,
\end{align}
where
\begin{equation}\label{eq:vector-normalization}
 \mathcal N_L=\frac{1}{\|\Upsl\|_2}
 =\left(\frac{d^{\,n+1}}
 {\sum_{\boldsymbol\mu}|T_L^{\boldsymbol\mu}|^2}\right)^{1/2}.
\end{equation}
Similarly, for any spatiotemporal state \(\Ups\), define
\begin{equation}\label{eq:vectorized-density}
 |\widehat\Ups\rangle\!\rangle
 \langle\!\langle\widehat\Ups|,
 \qquad
 |\widehat\Ups\rangle\!\rangle=\frac{|\Ups\rangle\!\rangle}{\|\Ups\|_2}.
\end{equation}
For a spacetime region \(\mathsf A\subseteq\{(x_n,t_n),\ldots,(x_0,t_0)\}\), the subsystem
is the doubled operator space
\begin{equation}
 \widetilde\cH_{\mathsf A}
 =\bigotimes_{(x_k,t_k)\in\mathsf A}
 \left(\cH_{(x_k,t_k)}\otimes\cH_{(x_k,t_k)}^{*}\right),
\end{equation}
and the reduced vectorized state is
\begin{equation}\label{eq:vector-reduced}
 \rho_{\mathsf A}^{\rm vec}(\Ups)
 :=\Tr_{\mathsf{A}^c} |\widehat\Ups\rangle\!\rangle
 \langle\!\langle\widehat\Ups|.
\end{equation}
This is quite different from \(\Tr_{\mathsf A^c}\Ups\): the latter is another temporal operator, while Eq.~\eqref{eq:vector-reduced} is a positive density matrix containing the operator correlations between \(\mathsf A\) and
\(\mathsf A^c\).

\begin{definition}[Vectorized spatiotemporal entropy]\label{def:vec-entropy}
For a spatiotemporal state $\Upsilon$ and a spatiotemporal region $\mathsf A$, 
we define the vectorized von Neumann, R\'enyi, and Tsallis entropies, respectively, by
\begin{align}
 S_{\rm vec}(\mathsf A)_\Upsilon
 &:=-\Tr\!\left[
 \rho_{\mathsf A}^{\rm vec}\log\rho_{\mathsf A}^{\rm vec}
 \right],\\
 S_{\rm vec}^{(q)}(\mathsf A)_\Upsilon
 &:=\frac{1}{1-q}\log\Tr\!\left[
 \big(\rho_{\mathsf A}^{\rm vec}\big)^q
 \right],
 \qquad q>0,\quad q\ne1,\\
 T_{\rm vec}^{(q)}(\mathsf A)_\Upsilon
 &:=\frac{1}{q-1}\left(
 1-\Tr\!\left[
 \big(\rho_{\mathsf A}^{\rm vec}\big)^q
 \right]\right),
 \qquad q>0,\quad q\ne1.
\end{align}
For \(m\) spacetime regions, the total mutual entropy
\begin{equation}\label{eq:vec-total-correlation}
 I_{\rm vec}(\mathsf A_1:\cdots:\mathsf A_m)
 :=\sum_{j=1}^mS(\rho_{\mathsf A_j}^{\rm vec})
 -S(\rho_{\mathsf A_1\cdots\mathsf A_m}^{\rm vec})
 =D\!\left(\rho_{\mathsf A_1\cdots\mathsf A_m}^{\rm vec}
 \middle\|\bigotimes_j\rho_{\mathsf A_j}^{\rm vec}\right)\ge0
\end{equation}
vanishes exactly when the chosen regions are mutually product.
\end{definition}

Split \(\boldsymbol\mu=(\boldsymbol\mu_{\mathsf A},
\boldsymbol\mu_{\mathsf A^c})\) and flatten the normalized spatiotemporal Bloch tensor,
\begin{equation}\label{eq:coefficient-flattening}
 C^{(\mathsf A)}_{\boldsymbol\mu_{\mathsf A},
 \boldsymbol\mu_{\mathsf A^c}}
 :=\frac{T^{\boldsymbol\mu_{\mathsf A},
 \boldsymbol\mu_{\mathsf A^c}}}
 {\sqrt{\sum_{\boldsymbol\mu}|T^{\boldsymbol\mu}|^2}} .
\end{equation}
For every nonzero spatiotemporal operator \(\Ups\),
\begin{equation}\label{eq:vec-gram}
 \rho_{\mathsf A}^{\rm vec}
 =C^{(\mathsf A)}C^{(\mathsf A)\dagger},\qquad
 \rho_{\mathsf A^c}^{\rm vec}
 =C^{(\mathsf A)\dagger}C^{(\mathsf A)}.
\end{equation}
The common nonzero eigenvalues are
\begin{equation}\label{eq:operator-schmidt-prob}
 p_\alpha=s_\alpha(C^{(\mathsf A)})^2
 =\frac{\lambda_\alpha^2}{\|\Ups\|_2^2},
\end{equation}
where $s_\alpha(C^{(\mathsf A)})$ are singular values, \(\lambda_\alpha\) are the operator Schmidt values of \(\Ups\) \footnote{In the temporal setting, for a cut separating the future block $\mathsf{F}=\{t_n,\dots,t_{k+1}\}$ from the past block $\mathsf{P}=\{t_k,\dots,t_0\}$, write
\begin{equation}
   \Ups_{t_n\cdots t_0}=\sum_{\alpha}\lambda_\alpha\;\mathcal{F}_\alpha\otimes\mathcal{P}_\alpha,
   \qquad \lambda_1\ge\lambda_2\ge\cdots\ge0,
\end{equation}
with $\{\mathcal{F}_\alpha\}$, $\{\mathcal{P}_\alpha\}$ orthonormal in the Hilbert--Schmidt inner product. The temporal Schmidt rank is $\chi_T(k)=\#\{\alpha:\lambda_\alpha>0\}$, and the temporal entanglement entropy is
\begin{equation}\label{eq:STdef}
   S_T(\sP):=S_{\rm vec}(\mathsf P)_\Ups
   =-\sum_\alpha p_\alpha\log p_\alpha,\qquad
   p_\alpha=\frac{\lambda_\alpha^2}{\sum_\beta\lambda_\beta^2}.
\end{equation}
R\'enyi and Tsallis entropies are defined in the obvious way.
}.
Hence \(S_{\rm vec}(\mathsf A)=S_{\rm vec}(\mathsf A^c)\), and the entropy
is obtained directly by an SVD of a flattening of the measured Bloch tensor,
without diagonalizing the generally non-Hermitian \(\Ups\).

Notice that the superdensity operator, which is also constructed from the doubled Wightman correlation tensor $T^{\bm\mu;\bm\nu}$, admits a related vectorization interpretation \cite{jia2024spatiotemporal}. In this construction, however, vectorization is applied only to the basis observables: the observables associated with the left half are vectorized into kets, while those associated with the right half are vectorized into bras. Since $T^{\bm\mu;\bm\nu}$ is positive semidefinite with respect to the left index $\bm\mu$ and the right index $\bm\nu$, the resulting operator is positive semidefinite and, after normalization, defines a density operator known as the superdensity operator \cite{cotler2018superdensity}. A detailed comparison between the superdensity operator and our vectorization formalism, together with their connections, is provided in Section~\ref{app:superDens}.

Let the physical Hilbert-space dimension of region \(\mathsf A\) be
\(d_{\mathsf A}\), so its operator-space dimension is \(d_{\mathsf A}^2\).
Then:
\begin{enumerate}
\item
\(0\le S_{\rm vec}(\mathsf A)_\Ups
\le\log\chi_{\mathsf A}
\le2\log\min\{d_{\mathsf A},d_{\mathsf A^c}\}\).
\item \(S_{\rm vec}(\mathsf A)_\Ups=0\) if and only if
\(\Ups=X_{\mathsf A}\otimes Y_{\mathsf A^c}\). 
\item The entropy is unchanged by nonzero scalar rescaling of \(\Ups\), by
local changes of Hilbert--Schmidt basis, and by local left--right unitary
actions
\[
 \Ups\mapsto(U_{\mathsf A}\otimes U_{\mathsf A^c})\,
 \Ups\,(V_{\mathsf A}\otimes V_{\mathsf A^c}).
\]
\item It is additive:
\[
 S_{\rm vec}(\mathsf A_1\mathsf A_2)_{\Ups_1\otimes\Ups_2}
 =S_{\rm vec}(\mathsf A_1)_{\Ups_1}
 +S_{\rm vec}(\mathsf A_2)_{\Ups_2}.
\]
\item Adding a region \(\mathsf R\) changes entropy by at most $\log d_{\mathsf R}$,
\begin{equation}\label{eq:local-entropy-growth}
 \left|S_{\rm vec}(\mathsf A\mathsf R)-S_{\rm vec}(\mathsf A)\right|
 \le2\log d_{\mathsf R}.
\end{equation}
\item If
\( \delta
  :=
  \frac12
 \|
  |\widehat{\Ups}\rangle\!\rangle
  \langle\!\langle\widehat{\Ups}|
  -
  |\widehat{\widetilde\Ups}\rangle\!\rangle
  \langle\!\langle\widehat{\widetilde\Ups}|
  \|_1
  \le
  1-d_{\mathsf A}^{-2}\), then
\begin{equation}\label{eq:vec-continuity}
 |S_{\rm vec}(\mathsf A)_\Ups-
 S_{\rm vec}(\mathsf A)_{\widetilde\Ups}|
 \le\delta\log(d_{\mathsf A}^2-1)+h_2(\delta).
\end{equation}
Here \(h_2(x)=-x\log x-(1-x)\log(1-x)\).
\end{enumerate}

Items  1-4 are pure-state entropy properties applied to \(|\widehat\Ups\rangle\!\rangle\).
Item 5 is derived by $S_{\rm vec}(\sA)+S_{\rm vec}(\sR)\geq S_{\rm vec}(\sA\sR) \geq |S_{\rm vec}(\sA)-S_{\rm vec}(\sR)|$ 
and \(S(\rho_{\mathsf R}^{\rm vec})\le\log d_{\mathsf R}^2\).

Item 6 can be proven as follows.
Define the reduced vectorized states
\begin{equation}
 \rho_{\mathsf A}^{\Ups}
 :=
 \Tr_{{\mathsf A}^c}
 \left(
 |\widehat{\Ups}\rangle\!\rangle
 \langle\!\langle\widehat{\Ups}|
 \right),
 \qquad
 \rho_{\mathsf A}^{\widetilde\Ups}
 :=
 \Tr_{{\mathsf A}^c}
 \left(
 |\widehat{\widetilde\Ups}\rangle\!\rangle
 \langle\!\langle\widehat{\widetilde\Ups}|
 \right).
\end{equation}
By contractivity of the trace norm under partial trace,
\begin{equation}
 \delta_{\mathsf A}
 :=
 \frac12
 \left\|
 \rho_{\mathsf A}^{\Ups}
 -
 \rho_{\mathsf A}^{\widetilde\Ups}
 \right\|_1
 \le \delta.
\end{equation}
Since the vectorized Hilbert space associated with $\mathsf A$ has
dimension $d_{\mathsf A}^2$, the Audenaert--Fannes inequality \cite{Fannes1973entroy,Audenaert2007entropy} gives
\begin{equation}
 \left|
 S_{\rm vec}(\rho_{\mathsf A}^{\Ups})
 -
 S_{\rm vec}(\rho_{\mathsf A}^{\widetilde\Ups})
 \right|
 \le
 \delta_{\mathsf A}\log(d_{\mathsf A}^2-1)
 +
 h_2(\delta_{\mathsf A}).
\end{equation}
The function
\begin{equation}
 f_D(x):=x\log(D-1)+h_2(x)
\end{equation}
is monotonically increasing for
$0\le x\le1-D^{-1}$, since
\begin{equation}
 f_D'(x)
 =
 \log\frac{(D-1)(1-x)}{x}
 \ge0
\end{equation}
on this interval. Therefore, using
$\delta_{\mathsf A}\le\delta\le1-d_{\mathsf A}^{-2}$ with
$D=d_{\mathsf A}^2$, we obtain Eq.~\eqref{eq:vec-continuity}.

The recursion of Section~\ref{sec:recursion} and its MPO representation immediately yield bounds on the temporal Schmidt rank $\chi_T$ and temporal entanglement entropy $S$. Let $\Ups$ be a temporal state; we then have the following:
\begin{enumerate}
\item If the process is memoryless with system dimension $d_S$, then for every temporal cut $\sP|\sF$ (past and future),
$\chi_T(\sP)\le d_S^{2}$ and hence $S_{T}(\sP)\le 2\log d_S$, independently of the number of time steps $n$.
\item If the process has memory generated by a dilation with environment dimension $d_E$, then $\chi_T(\sP)\le (d_Sd_E)^{2}$ and $S_{T}(\sP)\le2\log d_S+2\log d_E$. 
This implies that temporal entanglement measures quantum memory. The excess $S_{\rm vec}(\sP)- 2\log d_S$ is nonzero only when the process is non-Markovian, and is bounded by the entropy of the environmental degrees of freedom that actually mediate correlations across the cut.
\item For the $2m$-fold state (if we regard as a $2m$-fold $n$ step temporal process), $\chi_T(\sP)\le d_S^{2m}$ and $S_{T}(\sP)\le 2m\log d_S$.
\end{enumerate}

Therefore, the temporal state can be viewed as analogous to a gapped one-dimensional spatial state, suggesting the absence of long-range entanglement in the purely temporal setting without memory.
Note that the doubled Kirkwood--Dirac/Margenau--Hill temporal state has the same bond bound $d_S^2$ as the single-branch states, even though its local slot dimension is $d_S^2$; this is clear from the MPO representations. This constitutes the state-level statement that the Keldysh doubling costs nothing in temporal complexity for a Markovian process, and explains why influence-matrix methods can work with the doubled contour at no extra bond cost.

\subsection{Two-time vectorized temporal entropy}

For a two-time process, the temporal operator entanglement can be computed more explicitly. 
For initial state $\varrho$ evolution $\cE$ and probability measure $\mathbb{P}(s)$,
consider the temporal state defined in Eq.~\eqref{eq:TwoTimePs}.
Every member has the physical marginals
$\Tr_{t_1}\Ups^{\mathbb P}=\varrho$ and
$\Tr_{t_0}\Ups^{\mathbb P}=\cE(\varrho)$.
Define the map (which is generally not a CPTP map, but only a tool for constructing temporal states)
\begin{equation}\label{eq:modular-mixture-map}
 \mathcal M_{\mathbb P,\varrho}(X)
 :=\int_{[0,1]}\varrho^sX\varrho^{1-s}\,\mathbb P(s) ds,
 \qquad
 \mathcal M_{0,\varrho}=R_\varrho,
 \quad \mathcal M_{1,\varrho}=L_\varrho,
\end{equation}
where the displayed endpoint identities specify the convention when
$\varrho$ is not faithful.  Then
\begin{equation}\label{eq:general-mixture-choi}
 \Ups^{\mathbb P}_{t_1t_0}
 =J[\cE\circ\mathcal M_{\mathbb P,\varrho}],
\end{equation}
and its operator Schmidt values across $t_1|t_0$ are exactly the singular
values of
\begin{equation}\label{eq:general-mixture-superoperator}
 \widehat\cE\widehat{\mathcal M}_{\mathbb P,\varrho}
\end{equation}
with respect to the Hilbert--Schmidt inner product \footnote{This follows from the fact that $\widehat{\cE}\widehat{\mathcal M}_{\mathbb P,\varrho}$ gives the coefficient matrix of the temporal state with respect to the Hilbert--Schmidt basis $\{\sigma_{\mu}\otimes\sigma_{\nu}\}_{\mu,\nu}$.}.  Consequently,
\begin{equation}\label{eq:general-mixture-rank}
 \chi_T(\cE,\varrho,\mathbb P)
 =\rank\!\left(
 \widehat\cE\,\widehat{\mathcal M}_{\mathbb P,\varrho}
 \right).
\end{equation}
More explicitly, in eigenbasis of initial state
$\varrho=\sum_a\lambda_a\ket{a}\!\bra{a}$,
$\mathcal M_{\mathbb P,\varrho}$ reads
\begin{align}
 \mathcal M_{\mathbb P,\varrho}(\ket{a}\!\bra{b})
 &=m^{\mathbb P}_{ab}\ket{a}\!\bra{b},\label{eq:modular-weights}\\
 m^{\mathbb P}_{ab}
 &=p_0\lambda_b+p_1\lambda_a
 +\int_{(0,1)}\lambda_a^s\lambda_b^{1-s}\,\mathbb P(s) ds,\nonumber
\end{align}
where $p_0=\mathbb P(\{0\})$ and $p_1=\mathbb P(\{1\})$.
We see $\widehat{\mathcal M}_{\mathbb P,\varrho}$ is diagonal in the basis $||a\rangle \langle b|\rrangle$.
If $\varrho$ is invertible, every $m^{\mathbb P}_{ab}$ is strictly positive,
so $\mathcal M_{\mathbb P,\varrho}$ is invertible and
\begin{equation}\label{eq:faithful-general-rank}
 \chi_T(\cE,\varrho,\mathbb P)=\rank\widehat\cE
 \qquad\text{for every }\mathbb P.
\end{equation}
If $\varrho$ has support projector $P$ and rank $q<d$, the visible operator
space is exactly $\operatorname{range}\mathcal M_{\mathbb P,\varrho}:=\{ \cM_{\mathbb P,\varrho}(X)|X\in \mathbf{B}(\cH)\}$:
\begin{equation}\label{eq:visible-operator-spaces}
 \operatorname{range}\mathcal M_{\mathbb P,\varrho}
 =\begin{cases}
 P\BB(\cH)P, & p_0=p_1=0,\\
 \BB(\cH)P, & p_0>0,\ p_1=0,\\
 P\BB(\cH), & p_0=0,\ p_1>0,\\
 P\BB(\cH)+\BB(\cH)P, & p_0>0,\ p_1>0.
 \end{cases}
\end{equation}
Its dimension is respectively $q^2$, $qd$, $qd$, or $2qd-q^2$.
Thus the invisible dynamics is determined by
$\ker\mathcal M_{\mathbb P,\varrho}$, not merely by the orthogonal complement
of $\operatorname{supp}\varrho$. See~Section~\ref{app:MrhoP} in Appendix for more details.

Several consequences follow immediately (we use $S_T$ to denote vectorized entanglement entropy):
\begin{enumerate}

\item
For invertible $\varrho$, the vectorized temporal entanglement entropy $S_T=0$ for one, and hence every, choice of
$\mathbb P$ if and only if $\widehat\cE$ has rank one.  For a
quantum channel this is equivalent to a replacement channel,
\begin{equation}
 \cE(\bullet)=\Tr[\bullet]\,\sigma.
\end{equation}
For  $\varrho$ that is not of full rank, the vectorized temporal entanglement entropy $S_T=0$ if and only if
$\cE$ is replacement on the visible space in
Eq.~\eqref{eq:visible-operator-spaces}, namely
$\cE(X)=\Tr[X]\,\sigma$ for every
$X\in\operatorname{range}\mathcal M_{\mathbb P,\varrho}$.

\item If $\cE=\mathcal U$ is unitary channel, then $\widehat{\mathcal U}$ is a
Hilbert--Schmidt unitary and the operator Schmidt values are simply the
numbers $m^{\mathbb P}_{ab}$ in Eq.~\eqref{eq:modular-weights}.  Hence the temporal entanglement entropy is given by
\begin{equation}\label{eq:unitary-general-spectrum}
 p_{ab}=\frac{(m^{\mathbb P}_{ab})^2}
 {\sum_{c,d}(m^{\mathbb P}_{cd})^2},
 \qquad
 S_T=-\sum_{a,b}p_{ab}\log p_{ab}.
\end{equation}
Thus, for unitary evolution, the temporal entanglement spectrum is controlled entirely by the initial state and the ordering measure.
For a point measure $\mathbb P=\delta_s$, we have
$m_{ab}^{(s)}=\lambda_a^s\lambda_b^{1-s}$.  Hence
$(m_{ab}^{(s)})^2=\lambda_a^{2s}\lambda_b^{2(1-s)}$, and their
normalization factor factorizes:
\begin{equation}
 \sum_{a,b}(m_{ab}^{(s)})^2
 =
 \left(\sum_a\lambda_a^{2s}\right)
 \left(\sum_b\lambda_b^{2(1-s)}\right)
 =
 Z_sZ_{1-s},
\end{equation}
where
\begin{equation}
 \pi_a^{(\alpha)}
 :=
 \frac{\lambda_a^{2\alpha}}{Z_\alpha},
 \qquad
 Z_\alpha:=\sum_a\lambda_a^{2\alpha}.
\end{equation}
The normalized temporal Schmidt spectrum therefore factorizes:
\begin{equation}\label{eq:unitary-s-spectrum}
 p_{ab}
 =
 \pi_a^{(s)}\pi_b^{(1-s)}.
\end{equation}
By
additivity of the Shannon entropy,
\begin{align}
 S_T(\Ups^{(s)})
 &=
 -\sum_{a,b}
 \pi_a^{(s)}\pi_b^{(1-s)}
 \log\!\left[
 \pi_a^{(s)}\pi_b^{(1-s)}
 \right]\nonumber\\
 &=
 H(\pi^{(s)})+H(\pi^{(1-s)}).
 \label{eq:unitary-s-entropy}
\end{align}

For $0<s<1$, both $\pi^{(s)}$ and $\pi^{(1-s)}$ are supported on
$\operatorname{supp}\varrho$.  If $q=\rank\varrho$, each distribution has
$q$ nonzero components, and hence the temporal operator-Schmidt rank is
$\chi_T=q^2$.
At the endpoints, we use the one-sided convention
$\varrho^0=\mathds{1}$, including on $\ker\varrho$.  Consequently,
$\pi_a^{(0)}=1/d$ for all $a$, whereas $\pi^{(1)}$ has support of size
$q$.  It follows that $\chi_T=dq$ at both $s=0$ and $s=1$.  In
particular, if $\varrho$ is of full rank $q=d$, then
$\chi_T=d^2$ for every $s\in[0,1]$.
The symmetric ordering $s=\tfrac12$ is especially simple.  Since
$Z_{1/2}=\sum_a\lambda_a=1$, one has
$\pi_a^{(1/2)}=\lambda_a$, and therefore
\begin{equation}
 S_T(\Ups^{(1/2)})
 =
 2H(\{\lambda_a\})
 =
 2S(\varrho).
\end{equation}
\item The channel entanglement is a special case of temporal state entanglement. For an arbitrary channel $\cE$ and $\varrho=\I/d$, the temporal
Schmidt spectrum is the singular-value spectrum of $\widehat\cE/d$, so
$S_T$ is exactly the channel operator entanglement.

\item Postcomposition cannot increase the temporal Schmidt rank.  For every
$\varrho$ and $\mathbb P$,
\begin{equation}\label{eq:general-composition-rank}
 \chi_T(\cE'\circ\cE,\varrho,\mathbb P)
 \le\min\!\left\{
 \rank\widehat\cE',\,
 \chi_T(\cE,\varrho,\mathbb P)
 \right\}.
\end{equation}
For invertible $\varrho$ this becomes
$\chi_T(\cE'\circ\cE)\le
\min\{\chi_T(\cE'),\chi_T(\cE)\}$.
 A coherent unitary channel
preserves every operator direction visible through
$\mathcal M_{\mathbb P,\varrho}$ and saturates the corresponding Schmidt-rank
bound, but it is maximally entangling only when the weights are
uniform, as for $\varrho=\I/d$.  At the opposite extreme, a replacement
channel makes every member of the two-time family a product.  More general
noise can only reduce the temporal Schmidt rank under postcomposition, although it need not
decrease the normalized temporal entanglement entropy. 
\end{enumerate}

\begin{example}[Pure closed-system two-time state]
For the temporal state $\Ups^{\mathbb P}$ discussed in Example~\ref{exp:pure}, suppose $|0\rangle =|\psi\rangle$ is initial state, its nonzero operator Schmidt values $m_{ab}^{\mathbb{P}}$ are
\begin{equation}
m_{00}^{\mathbb{P}}= 1;\qquad m_{a0}^{\mathbb{P}}=
 \underbrace{p_0,\ldots,p_0}_{d-1}, a\neq 0; \qquad \qquad m_{0,b}^{\mathbb{P}}=
 \underbrace{p_1,\ldots,p_1}_{d-1}, b\neq 0.
\end{equation}
Writing $D=1+(d-1)(p_0^2+p_1^2)$, the temporal entanglement entropy is given by
\begin{equation}
 S_T=\log D-\frac{2(d-1)}{D}
 \left(p_0^2\log p_0+p_1^2\log p_1\right),
\end{equation}
with $0\log0:=0$. Thus, for a pure initial state evolving unitarily, either endpoint
prescription, $\mathbb P=\delta_0$ or $\mathbb P=\delta_1$, gives
$\chi_T=d$ and $S_T=\log d$. By contrast, every
$\mathbb P=\delta_r$ with $0<r<1$ gives $\chi_T=1$ and $S_T=0$.
\end{example}

%==================================================================
\section{Spatiotemporal entanglement, entropy and mutual entropy: positivity, separability, and spectral structure}
\label{sec:entropy}

In contrast to the operational approach adopted in Section~\ref{sec:STentropy}, since a spatial state, represented by a density operator, can be regarded as a spatiotemporal state restricted to a fixed time slice, we seek to define notions of entanglement, entropy, and related quantities for spatiotemporal states in such a way that, upon restriction to a fixed time slice, they naturally reduce to the corresponding standard definitions for density operators.
A spatiotemporal operator is normalized, $\Tr\Ups=1$, but need not be positive, Hermitian, or even normal. Consequently, there is no unique extension of the standard notions of entropy and entanglement that simultaneously preserves all of their familiar properties and interpretations. We first consider their extension to Margenau--Hill spatiotemporal states, whose Hermiticity makes them more amenable to analysis, and then turn to general Kirkwood--Dirac spatiotemporal states.

\subsection{Hermitian spatiotemporal states: positivity and temporal separability}
\label{sec:mh-separability}

For any spatiotemporal state, Hermitianization yields a Hermitian spatiotemporal
state. Since such states give rise to real-valued spatiotemporal
quasiprobability distributions through the spatiotemporal Born rule, we refer
to them as Margenau--Hill spatiotemporal states\footnote{This is a more general definition than that given in Section~\ref{sec:tempstate}.}.

Let $\Ups_{\sA\sB}$ denote the Margenau--Hill spatiotemporal operator associated
with a bipartition of the sampled spacetime slots into $\sA|\sB$.
Hermiticity and unit trace imply that $\Ups_{\sA\sB}$ belongs to the
affine space
\begin{equation}
 \operatorname{Herm}_1(\sA\sB)
 :=
 \{X=X^\dagger:\Tr X=1\},
\end{equation}
but not necessarily to the set of density operators
$\operatorname{Pos}_1(\sA\sB)$. This distinction is essential.

\begin{definition}[Spatiotemporal separability]
\label{def:mh-separable}
A Hermitian Margenau--Hill spatiotemporal operator is said to be
spatiotemporally separable across the bipartition $\sA|\sB$ if it admits a
decomposition
\begin{equation}\label{eq:mh-separable}
 \Ups_{\sA\sB}
 =
 \sum_j p_j\,\rho_j^{\sA}\otimes\tau_j^{\sB},
 \qquad
 p_j\geq0,\quad \sum_j p_j=1,
 \qquad
 \rho_j^{\sA},\tau_j^{\sB}\geq0,
 \quad
 \Tr\rho_j^{\sA}=\Tr\tau_j^{\sB}=1.
\end{equation}
The resulting compact convex set is denoted by
$\Sep(\sA{:}\sB)$.
The definition extends straightforwardly to a multipartition
$\sA_1|\sA_2|\cdots|\sA_n$. In this case, different notions of
spatiotemporal separability form a natural hierarchy, ranging from full
separability across all spacetime regions to partial separability, where
entanglement may persist within subsets of the partition.
\end{definition}

This is precisely the standard density-operator notion of separability. An important advantage of this definition is that it naturally reduces to the usual notion of separability for spatial states, without requiring any additional assumptions or modifications.

Every $\Ups\in\operatorname{Herm}_1(\sA\sB)$ lies in exactly one of
the following regions:
\begin{align}
 &\Ups\in\Sep(\sA{:}\sB),
 &&\text{positive and spatiotemporally separable},\nonumber\\
 &\Ups\in\operatorname{Pos}_1(\sA\sB)\setminus
   \Sep(\sA{:}\sB),
 &&\text{positive and spatiotemporally entangled},\nonumber\\
 &\Ups\notin\operatorname{Pos}_1(\sA\sB),
 &&\text{spatiotemporally nonpositive}.\label{eq:mh-three-regions}
\end{align}
In the first two regions temporality $\mathfrak{T}(\Ups)=0$; in the third, $\mathfrak{T}(\Ups)>0$.

\begin{figure}[h]
    \centering
    \includegraphics[width=0.35\linewidth]{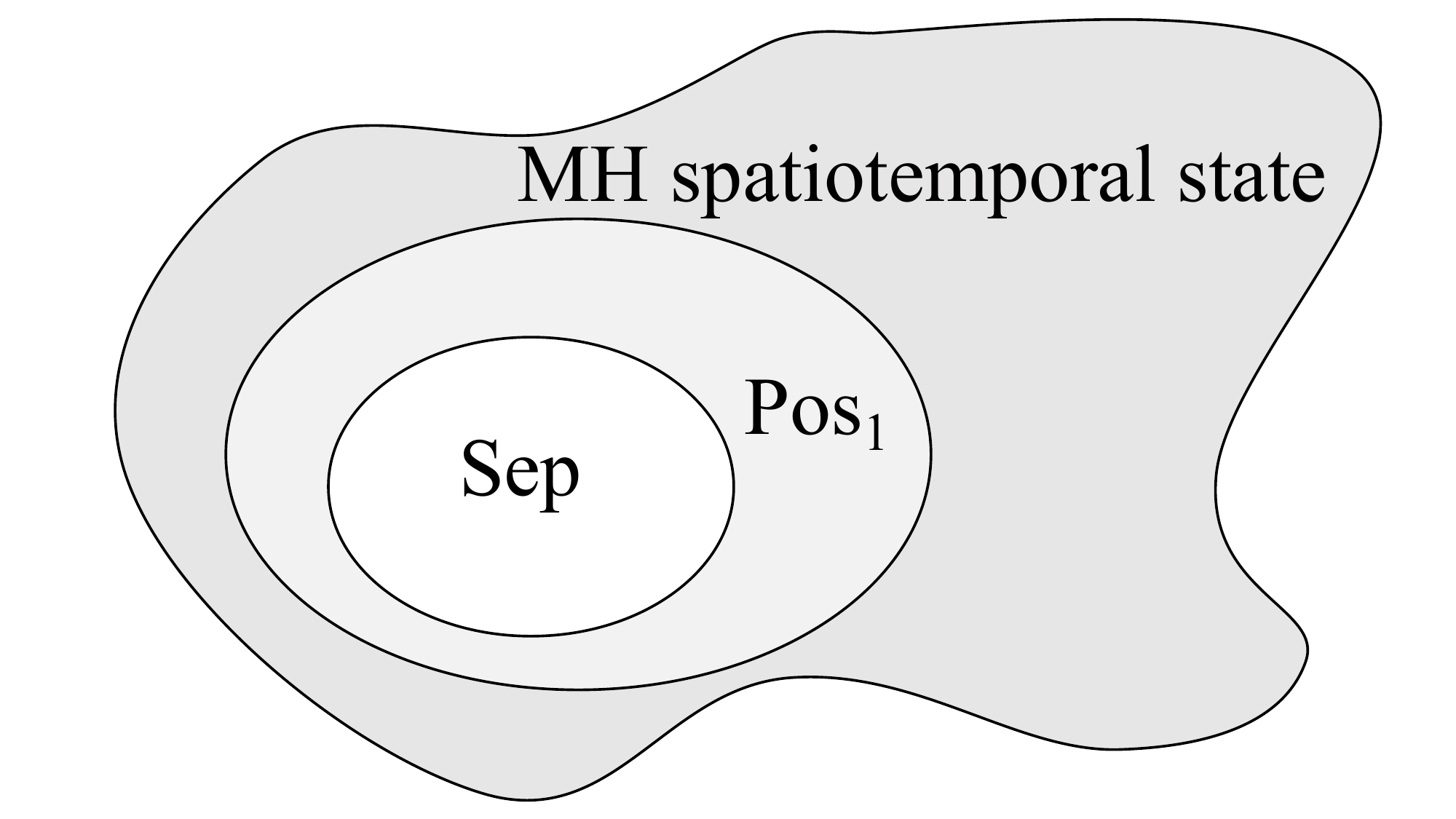}
   \caption{Classification of Hermitian spatiotemporal states, where spatiotemporal entanglement is characterized by the deviation from the set of separable states.}
    \label{Fig:SEPclassification}
\end{figure}

\subsubsection{Spatiotemporal entanglement robustness beyond the positive cone}
\label{sec:STentRobust}
Although conventional entanglement measures are defined for positive operators, the separable cone provides a natural extension of the notion of entanglement to arbitrary Hermitian spatiotemporal operators.

\begin{definition}[Spatiotemporal entanglement robustness]
\label{def:sep-base-norm}
For a spatiotemporal state $\Ups\in\operatorname{Herm}_1(\sA\sB)$, we define the separable base norm by \footnote{It can be proven to satisfy the axioms of a norm over $\operatorname{Herm}(\sA\sB)$, see Section~\ref{app:proofofTandR} in Appendix.}
\begin{equation}
\label{eq:sep-base-norm}
 \|\Ups\|_{\rm SEP}
 :=
 \inf\left\{
 \sum_j |c_j| \,:\,
 \Ups=\sum_j c_j\,\rho_j^{\sA}\otimes\tau_j^{\sB},
 c_j\in\mathbb{R},
 \rho_j^{\sA}\in\operatorname{Pos}_1(\sA),
 \tau_j^{\sB}\in\operatorname{Pos}_1(\sB)
 \right\},
\end{equation}
where the coefficients $c_j$ form a quasiprobability distribution; the separable base norm thus corresponds to the optimal quasiprobability mixture of product states with minimal negativity.
The entanglement robustness is then defined by
\begin{equation}
\label{eq:signed-sep-robustness}
 \mathfrak{R}_{\rm SEP}(\Ups)
 :=
 \frac{\|\Ups\|_{\rm SEP}-1}{2},
\end{equation}
which quantifies the negativity of the optimal coefficients $c_j$.
Equivalently, $\mathfrak{R}_{\rm SEP}(\Ups)$ can be expressed as \cite{Vidal1999robustness}
 \begin{equation}
  \mathfrak{R}_{\rm SEP}(\Ups) =  \inf_{\alpha} \{A=\frac{\Ups+\alpha Y}{1+\alpha} \in \Sep(\sA{:}\sB), Y\in \Sep(\sA{:}\sB)\}.
 \end{equation}
It is clear that $\mathfrak{R}_{\rm SEP}(\Ups)\ge0$, with equality if and only if
$\Ups\in\Sep(\sA{:}\sB)$.
 The spatiotemporal entanglement is then given by $\mathfrak{R}_{\rm SEP}(\Ups)\geq 0$.
\end{definition}

Thus, Margenau--Hill spatiotemporal states can be classified into the following three categories (see Figure~\ref{Fig:SEPclassification}):
\begin{itemize}
\item \emph{Positive and spatiotemporally separable}:
$\fT(\Ups)=0$ and $\mathfrak{R}_{\rm SEP}(\Ups)=0$.
\item \emph{Positive and spatiotemporally entangled}: 
$\fT(\Ups)=0$ and $\mathfrak{R}_{\rm SEP}(\Ups)>0$.
\item \emph{Nonpositive and spatiotemporally entangled}: 
$\fT(\Ups)>0$ and $\mathfrak{R}_{\rm SEP}(\Ups)>0$.
\end{itemize}
For every finite-dimensional Margenau--Hill spatiotemporal state $\Ups$, the spatiotemporal entanglement robustness inherits many of the standard properties of entanglement robustness for density operators. We do not review these properties here and instead refer to Ref.~\cite{Vidal1999robustness} and the general framework of the resource theory of entanglement.

\begin{proposition}\label{prop:TandR}
     In the spatiotemporal setting, we emphasize the following properties:
\begin{enumerate}
\item The partial-transpose negativity, based on the Peres--Horodecki PPT
criterion, is a standard diagnostic of spatial entanglement. For a
nonpositive spatiotemporal operator, however, the partial transpose may be
positive, so partial-transpose negativity alone does not detect all departures
from the separable-state cone. Temporality and partial-transpose negativity
capture logically distinct obstructions. Neither condition implies the other.
We define the extended partial-transpose negativity
\begin{equation}\label{eq:extended-pt-negativity}
 \fN_{\rm PT}(\Ups):=\frac{\|\Ups^{T_{\sB}}\|_1-1}{2}.
\end{equation}
The spatiotemporal entanglement robustness obeys the computable lower bound
\begin{equation}\label{eq:robustness-lower-bounds}
 \max\left\{
 \mathfrak{T}(\Ups),
 \fN_{\rm PT}(\Ups)
 \right\}
 \leq \mathfrak{R}_{\rm SEP}(\Ups).
\end{equation}
Accordingly, the gap
$\mathfrak{R}_{\rm SEP}(\Ups)-\mathfrak{T}(\Ups)$
measures the nonseparability beyond the minimal contribution enforced by
nonpositivity. It therefore provides a useful diagnostic of spatiotemporal
entanglement, although the gap itself need not be monotonic.
\item Let $\Ups=\Ups_+-\Ups_-$ be the Hahn-Jordan decomposition and set
$r:=\Tr \Ups_-=\mathfrak{T}(\Ups)$.  Then
\begin{equation}\label{eq:temporality-saturation}
 \mathfrak{R}_{\rm SEP}(\Ups)=\mathfrak{T}(\Ups)
\end{equation}
if and only if $\Ups_+/(1+r)\in\Sep(\sA{:}\sB)$ and, when $r>0$,
$\Ups_-/r\in\Sep(\sA{:}\sB)$.  Hence the strict gap
$\mathfrak{R}_{\rm SEP}(\Ups)-\mathfrak{T}(\Ups)$ is present precisely when the
positive or negative Hahn-Jordan sector cannot itself be normalized to a separable
state.
\item The separable base norm has the witness dual
\begin{equation}\label{eq:sep-base-dual}
 \|\Ups\|_{\rm SEP}
 =\sup_{W\in \operatorname{Herm}}\left\{\Tr(W\Ups):
 \big|\Tr[W(\rho_{\sA}\otimes\tau_{\sB})]\big|\le1
 \ \text{for all }\rho_{\sA}\in\operatorname{Pos}_1(\sA),\
 \tau_{\sB}\in\operatorname{Pos}_1(\sB)\right\}.
\end{equation}
Every feasible $W$ can be regarded as a spatiotemporal entanglement witness.
\end{enumerate} 
\end{proposition}

The proofs are given in Section~\ref{app:proofofTandR} in Appendix.

For Hermitian spatiotemporal states, the entropy introduced in Ref.~\cite{jia2023quantumspace} based on singular-value decomposition (SVD) applies directly, and all of its established properties carry over without modification.

%%%%%%%%%%%%%
\subsection{Spectral and SVD entropies}
We now turn to the general Kirkwood-Dirac spatiotemporal state which may be non-Hermitian and even nonnormal.
There are three natural, but inequivalent, extensions of the von Neumann
entropy.  They retain, respectively, the eigenvalue spectrum, the normalized
singular-value spectrum, and the scale of the magnitude operator.  Keeping the
three notions separate is important whenever $\Ups$ is nonpositive or
nonnormal.

We will mainly use the spectral value $\lambda_j$ and singular values $s_j$ of spatiotemporal state to define entropies, thus we first establish the following result:

\begin{proposition}[Unit-disk spectral bound of spatiotemporal state]\label{prop:temporal-spectral-disk}
Let $\Ups$ be a finite-dimensional spatiotemporal state obtained from density
operators and CPTP maps by the left, right, mixed-branch, doubled, or convexly
mixed recursions of Section~\ref{sec:recursion}.  Then
\begin{equation}\label{eq:temporal-contraction}
 \|\Ups\|_\infty\le1,
 \qquad s_j(\Ups)\in[0,1],
 \qquad |\lambda_j(\Ups)|\le1.
\end{equation}
In particular, the spectrum of a general Kirkwood--Dirac state lies in the closed unit
disk.  If $\Ups$ is Hermitian, as for a Margenau--Hill or
L\"uders--von Neumann state, then $\lambda_j\in[-1,1]$.
\end{proposition}

\begin{proof}
Notice
$\|J[\cE]\|_\infty\le1$ for any CPTP map $\cE$ \cite{fullwood2026entropypseudodensitymatrix}.  Also $\|\varrho_{t_0}\|_\infty\le1$ for every density
operator.  Each padded temporal link in Section~\ref{sec:recursion} is an ordinary product of these contractions, so submultiplicativity gives
$\|\Ups\|_\infty\le1$.  The same dilation argument applies to the doubled
link, whose extra cap is itself a swap.  Taking adjoints preserves the norm,
and a convex mixture of contractions remains a contraction.  Finally,
$s_j(\Ups)\le\|\Ups\|_\infty$ and every eigenvalue obeys
$|\lambda_j(\Ups)|\le\|\Ups\|_\infty$.  Hermiticity makes the eigenvalues
real, yielding the interval $[-1,1]$.
\end{proof}

For general spatiotemporal state, Weyl's majorant inequality gives
\begin{equation}\label{eq:spectral-singular-majorization}
 \sum_{j=1}^k|\lambda_j|^\downarrow
 \le\sum_{j=1}^k s_j^\downarrow\quad(1\le k\le D),
 \qquad
 \sum_j|\lambda_j|\le\|\Ups\|_1,
\end{equation}
so the singular spectrum controls the absolute eigenvalue weight, with equality holds for normal $\Ups$.  
Define $Z_\lambda(\Ups):=\sum_j|\lambda_j|$, where eigenvalues are counted with algebraic multiplicity.
For a spatiotemporal state of rank $R$, Eq.~\eqref{eq:spectral-singular-majorization} implies
\begin{equation}\label{eq:norm-hierarchy}
 1\le Z_\lambda(\Ups)\le \|\Ups\|_1 \le R\le D.
\end{equation}
Moreover, the temporality has a decomposition
\begin{equation}\label{eq:temporality-decomposition}
 \mathfrak T(\Ups)
 =\underbrace{\frac{Z_\lambda(\Ups)-1}{2}}_{\displaystyle
 \mathfrak T_{\rm eig}(\Ups)}
 +\underbrace{\frac{\|\Ups\|_1-Z_\lambda(\Ups)}{2}}_{\displaystyle
 \mathfrak T_{\rm nn}(\Ups)}.
\end{equation}
Here $\mathfrak T_{\rm eig}$ is determined entirely by the complex spectrum,
whereas $\mathfrak T_{\rm nn}$ is a nonnormality contribution. 
Indeed, we have (See Section~\ref{app:normal} in Appendix for a proof),
\begin{equation}\label{eq:STstateNormal}
 \mathfrak T_{\rm nn}(\Ups)=0
 \quad\Longleftrightarrow\quad
 \Ups\ \text{is normal}.
\end{equation}
For an Margenau–Hill spatiotemporal state, $\mathfrak T_{\rm nn}(\Ups)=0$ and $\mathfrak T(\Ups)=\mathfrak T_{\rm eig}(\Ups)=\sum_{\lambda_j<0}|\lambda_j|$.

\begin{proposition}\label{prop:pair-complex-eigenvalue}
    Let $\varrho_{t_0}=\varrho_{t_0}^\dagger$ and consider a Hermitian-preserved channel $\mathcal{E}$. For the sequential left/right, mixed-branch, and double Kirkwood-Dirac temporal state constructions defined before, the characteristic polynomial has real coefficients. Hence nonreal eigenvalues occur in complex-conjugate pairs, with equal algebraic multiplicities, and the spectral extended Fullwood-Parzygnat entropy, which is defined in \eqref{def:EFP} is real.
     
    \end{proposition}
\begin{proof}
    We write the characteristic polynomial of a matrix $M$ as
\begin{align*}
    \chi_M(\lambda)=\det(\lambda\mathds1-M).
\end{align*}
The proof uses the following facts:
\begin{enumerate}
    \item $\chi_{AB}=\chi_{BA}$ for square matrices $A$ and $B$ of the same size.
    \item $\chi_{M^\dagger}=\chi_M$ if and only if $\chi_M$ has real coefficients.
\end{enumerate}

For the right Krikwood-Dirac temporal state , we have
\begin{align}
    \rgt\Upsilon_{t_n,\dots,t_0}&=J[\cE_n]\star\cdots\star\varrho_{t_0}\notag\\
    &=A_nA_{n-1}\cdots A_0,
\end{align}
where, $A_i,i\neq0$ denotes $J[\cE_i]$ acting on the time slots $t_i,t_{i-1}$, tensored with the identity on all remaining slots, and $ A_0=\mathds1_n\otimes\cdots\otimes \mathds1_1\otimes\varrho_{t_0}$. These factors satisfy $A_i^\dagger=A_i$ and $[A_i,A_j]=0,$ when $|i-j|\ge 2$.
The Hermitian conjugate of right Krikwood-Dirac temporal state  is
\begin{align}
    \rgt\Upsilon_{t_n,\dots,t_0}^{\dagger}&=A_0A_1\cdots A_n.
\end{align}
Since
\begin{align}
    \chi_{A_0A_1\cdots A_n}&=\chi_{A_2A_0A_1A_3\cdots A_n}\text{ by using cyclic property and commutation}\notag\\
    &=\begin{cases}
        \chi_{A_n\cdots A_0\cdots A_{n-1}}, \text{if $n$ is even}\\
        \chi_{A_{n-1}\cdots A_0\cdots A_{n}}, \text{if $n$ is odd }
    \end{cases}\notag\\
    &=\chi_{A_n\cdots A_0},
\end{align}
and by fact~2, this characteristic polynomial has real coefficients.
The same argument applies to the left and mixed-branch Krikwood-Dirac temporal state , since each is a product of the same Hermitian factors in a different order. Moreover, the left-, right-, and mixed-branch Kirkwood-Dirac temporal states have the same spectrum, since they share the same characteristic polynomial.

For the doubled  Krikwood-Dirac temporal state, we write
\begin{align}
    \dbl\Upsilon_{t_n,\dots,t_0}&=\JJ[\cE_n]\star\cdots\star(\varrho_{t_0}\otimes\mathds1)\mathbb{S}\notag\\
    &=H_JH_S,
\end{align}
where
\begin{align}
    H_J&=\sum_{i_1,j_1,\dots,i_n,j_n}\left(\cE_n(|i_n\rangle\langle j_n|)\otimes\mathds1\right)_{t_n}\notag\\
    &\quad\otimes\left(\cE_{n-1}(|i_{n-1}\rangle\langle j_{n-1}|)\otimes|j_n\rangle\langle i_n|\right)_{t_{n-1}}\otimes\cdots\notag\\
    &\quad\otimes\left(\varrho_{t_0}\otimes|j_1\rangle\langle i_1|\right)_{t_0},\\
    H_S&=\mathbb{S}_{t_n}\otimes\cdots\otimes\mathbb{S}_{t_0}.
\end{align}
Using $\cE_k(X)^\dagger=\cE_k(X^\dagger)$ and relabeling $i_k\leftrightarrow j_k$ for every $k$, we obtain $H_J^\dagger=H_J$. Each swap operator is Hermitian, so $H_S^\dagger=H_S$ as well. Hence
\begin{align*}
    \left(\dbl\Upsilon_{t_n,\dots,t_0}\right)^\dagger
    =H_SH_J.
\end{align*}
Thus, the characteristic polynomial of the doubled  Krikwood-Dirac temporal state also has real coefficients.
\end{proof}

\subsubsection{Spectral entropy}

Let $\lambda_j$ be the eigenvalues of $\Ups$, counted with algebraic
multiplicity, and write
\begin{equation}
 \bm\lambda(\Ups)=(\lambda_1,\ldots,\lambda_D).
\end{equation}
The spectrum gives rise to three related objects.  They should not be
conflated: only the third is the entropy of a probability distribution.

\paragraph{Complex spectral entropy.}
Fix a branch cut of the logarithm that avoids the nonzero spectrum and use the
continuous convention $0\log0=0$.  Whenever
$C_{\rm spec}^{(q)}(\Ups)\ne0$, define
\begin{align}
 C_{\rm spec}^{(q)}(\Ups)&=\Tr\Ups^q=\sum_j\lambda_j^q,\\
 S_{\rm spec}^{(q)}(\Ups)&=\frac{1}{1-q}\log C_{\rm spec}^{(q)}(\Ups),\label{eq:spectral-renyi}\\
 T_{\rm spec}^{(q)}(\Ups)&=\frac{1}{q-1}\left[1-C_{\rm spec}^{(q)}(\Ups)\right],\label{eq:spectral-tsallis}\\
 S_{\rm spec}(\Ups)&=-\Tr(\Ups\log\Ups)=-\sum_j\lambda_j\log\lambda_j .
 \label{eq:spectral-vn}
\end{align}
For integer $q$, $C_{\rm spec}^{(q)}$ is branch independent.  For noninteger
$q$, and for the logarithmic quantities, a branch must be specified.  The
definitions extend to defective invertible operators through the holomorphic
functional calculus; a defective Jordan block at zero requires a separate
regularization.  Spectral entropies are similarity invariant and, with
compatible logarithm branches, additive under tensor products.  They are
generally complex and need not be monotone under quantum channels.  They should
therefore be regarded as spectral diagnostics, analogous to pseudo entropy for
a transition matrix \cite{Doi2023Pseudoentropy,Doi2023Timelike}.

\paragraph{Unnormalized spectral-magnitude entropy.}
Consider next the nonnegative spectral-magnitude vector
\begin{equation}
 \bm\lambda^{|\bullet|}(\Ups)
 :=(|\lambda_1|,\ldots,|\lambda_D|)
\end{equation}
and its moments
\begin{equation}\label{eq:spectral-magnitude-moments}
 C_{|\lambda|}^{(q)}(\Ups):=\sum_j|\lambda_j|^q.
\end{equation}
The direct von Neumann-type functional associated with this unnormalized
vector is
\begin{equation}\label{eq:spectral-magnitude-entropy}
 S_{|\lambda|}(\Ups)
 :=-\sum_j|\lambda_j|\log|\lambda_j|.
\end{equation}
For the dynamically admissible states of
Proposition~\ref{prop:temporal-spectral-disk}, $0\le|\lambda_j|\le1$, and hence
$S_{|\lambda|}$ is real and nonnegative.  More precisely,
\begin{equation}
 0\le S_{|\lambda|}(\Ups)\le \frac{D}{e}.
\end{equation}
It is similarity invariant and unchanged by taking the adjoint.  It is not a
probability entropy, because in general
$\sum_j|\lambda_j|\ne1$, and it is not additive.  If
\begin{equation}\label{eq:spectral-modulus}
 Z_\lambda(\Ups):=\sum_j|\lambda_j|,
\end{equation}
then
\begin{equation}\label{eq:raw-spectral-magnitude-product}
 S_{|\lambda|}(X\otimes Y)
 =Z_\lambda(Y)S_{|\lambda|}(X)
  +Z_\lambda(X)S_{|\lambda|}(Y).
\end{equation}
For a unit-trace operator, the triangle inequality gives
$1\le Z_\lambda(\Ups)$; equality does not imply positivity unless additional
conditions, such as Hermiticity, are imposed.

\paragraph{Normalized spectral-magnitude entropy.}
The Euclidean normalization
$\bm\lambda^{|\bullet|}/\|\bm\lambda^{|\bullet|}\|_2$ does \emph{not} define
a probability distribution.  The correct normalization is the $\ell_1$ norm:
\begin{equation}\label{eq:normalized-spectral-magnitude}
 p_j^{|\lambda|}(\Ups)
 :=\frac{|\lambda_j|}{Z_\lambda(\Ups)},
 \qquad \sum_jp_j^{|\lambda|}=1.
\end{equation}
We define the normalized spectral-magnitude entropy and its R\'enyi and Tsallis
extensions by
\begin{align}
 \widehat S_{|\lambda|}(\Ups)
 &:=-\sum_jp_j^{|\lambda|}\log p_j^{|\lambda|},
 \label{eq:normalized-spectral-magnitude-entropy}\\
 \widehat S_{|\lambda|}^{(q)}(\Ups)
 &:=\frac{1}{1-q}\log\sum_j(p_j^{|\lambda|})^q,
 \label{eq:normalized-spectral-magnitude-renyi}\\
 \widehat T_{|\lambda|}^{(q)}(\Ups)
 &:=\frac{1}{q-1}\left[1-\sum_j(p_j^{|\lambda|})^q\right].
 \label{eq:normalized-spectral-magnitude-tsallis}
\end{align}
If $r=\#\{j:\lambda_j\ne0\}$, then
\begin{equation}
 0\le\widehat S_{|\lambda|}(\Ups)\le\log r,
\end{equation}
with equality on the right exactly when the nonzero eigenvalues have equal
modulus.  This entropy is real, similarity invariant, invariant under
adjunction, invariant under nonzero scalar rescaling, and additive under tensor
products.  The raw and normalized quantities satisfy
\begin{equation}\label{eq:raw-normalized-spectral-relation}
 S_{|\lambda|}(\Ups)
 =Z_\lambda(\Ups)
 \left[\widehat S_{|\lambda|}(\Ups)-\log Z_\lambda(\Ups)\right].
\end{equation}
When $\Ups=\rho\ge0$ and $\Tr\rho=1$, one has $Z_\lambda=1$ and all three
von Neumann-type expressions reduce to the ordinary entropy $S(\rho)$.

The normalized spectral-magnitude entropy is nevertheless not the same as the
SVD entropy introduced below.  For a nonnormal operator, the moduli
$|\lambda_j|$ are generally different from its singular values $s_j$.  They
coincide, including multiplicities, when $\Ups$ is normal.

\paragraph{Mutual-information diagnostics.}
For a bipartition $\sA|\sB$, one may define the spectral mutual entropy
\begin{equation}\label{eq:spectral-mi}
 I_{\rm spec}(\sA:\sB)_\Ups
 :=S_{\rm spec}(\Ups_{\sA})+S_{\rm spec}(\Ups_{\sB})
 -S_{\rm spec}(\Ups).
\end{equation}
Likewise, one may introduce the normalized spectral-magnitude diagnostic
\begin{equation}\label{eq:normalized-spectral-magnitude-mi}
 \widehat I_{|\lambda|}(\sA:\sB)_\Ups
 :=\widehat S_{|\lambda|}(\Ups_{\sA})
  +\widehat S_{|\lambda|}(\Ups_{\sB})
  -\widehat S_{|\lambda|}(\Ups).
\end{equation}
Neither expression is a mutual information in the information-theoretic sense.
The spectrum of a partial trace is not a marginal of the global spectrum;
therefore subadditivity and data processing do not apply.  In particular,
$I_{\rm spec}$ can be complex or negative, while
$\widehat I_{|\lambda|}$ is real but can still be negative.  For example,
\begin{equation}
 \Ups=\frac12\diag(1,1,1,-1)
\end{equation}
has $\Ups_{\sA}=\Ups_{\sB}=\diag(1,0)$.  Consequently
$\widehat S_{|\lambda|}(\Ups_{\sA})=
\widehat S_{|\lambda|}(\Ups_{\sB})=0$ and
$\widehat S_{|\lambda|}(\Ups)=\log4$, giving
$\widehat I_{|\lambda|}=-\log4$.

%===================
\subsubsection{SVD entropy}
\label{sec:svd-entropies}

For a nonnormal Kirkwood--Dirac spatiotemporal state, absolute eigenvalues $|\lambda_j|$ need not equal the singular values $s_j$.
The SVD therefore provides a genuinely different way to define entropy.
We call
\begin{equation}\label{eq:svd-entropy}
 S_{\rm SVD}(\Ups)
 :=-\Tr(|\Ups|\log|\Ups|)
 =-\sum_js_j\log s_j
\end{equation}
the \emph{SVD entropy}, with $0\log0=0$.  
We can also introduce the normalized spatiotemporal state $\omega_\Ups:=\frac{|\Ups|}{\|\Ups\|_1}$, from which we obtain the probability distribution
\begin{equation}\label{eq:normalized-svd-probability}
 p_j^{\rm SVD}:=\frac{s_j}{\|\Ups\|_1},
\end{equation}
and defines the \emph{normalized SVD entropy}
\begin{equation}\label{eq:normalized-svd-entropy}
 \widehat S_{\rm SVD}(\Ups)
 :=S(\omega_\Ups) = -\Tr ( \omega_{\Ups} \log \omega_{\Ups})
 =-\sum_jp_j^{\rm SVD}\log p_j^{\rm SVD}.
\end{equation}
Their exact relation is
\begin{equation}\label{eq:svd-normalized-svd-relation}
 S_{\rm SVD}(\Ups)
 =\|\Ups\|_1\left[\widehat S_{\rm SVD}(\Ups)-\log \|\Ups\|_1\right].
\end{equation}
For \(q>0\), \(q\neq1\), define
\begin{align}
    C_{\rm SVD}^{(q)}(\Ups)
    &:=\Tr|\Ups|^q
      =\sum_j s_j^q,
    \label{eq:svd-moments}\\
    \widehat C_{\rm SVD}^{(q)}(\Ups)
    &:=C_{\rm SVD}^{(q)}(\omega_\Ups)
      =\Tr\omega_\Ups^q
      =\frac{C_{\rm SVD}^{(q)}(\Ups)}{\|\Ups\|_1^q},
    \label{eq:normalized-svd-moments}\\
    S_{\rm SVD}^{(q)}(\Ups)
    &:=\frac{t_\Ups}{1-q}
      \log\frac{C_{\rm SVD}^{(q)}(\Ups)}{t_\Ups},
    \label{eq:svd-renyi}\\
    \widehat S_{\rm SVD}^{(q)}(\Ups)
    &:=\frac{1}{1-q}
      \log\widehat C_{\rm SVD}^{(q)}(\Ups),
    \label{eq:normalized-svd-renyi}\\
    T_{\rm SVD}^{(q)}(\Ups)
    &:=\frac{t_\Ups-C_{\rm SVD}^{(q)}(\Ups)}{q-1},
    \label{eq:svd-tsallis}\\
    \widehat T_{\rm SVD}^{(q)}(\Ups)
    &:=\frac{1-\widehat C_{\rm SVD}^{(q)}(\Ups)}{q-1}.
    \label{eq:normalized-svd-tsallis}
\end{align}
In $q\to 1$ limit, we obtain
\begin{align}
    \lim_{q\to1}S_{\rm SVD}^{(q)}(\Ups)
    &=
    \lim_{q\to1}T_{\rm SVD}^{(q)}(\Ups)
    =
    -\Tr\!\left(|\Ups|\log|\Ups|\right)
    =S_{\rm SVD}(\Ups),\\
    \lim_{q\to1}\widehat S_{\rm SVD}^{(q)}(\Ups)
    &=
    \lim_{q\to1}\widehat T_{\rm SVD}^{(q)}(\Ups)
    =
    -\Tr\!\left(\omega_\Ups\log\omega_\Ups\right)
    =\widehat S_{\rm SVD}(\Ups).
\end{align}
Although the distinction between these two types of entropy may appear insignificant at the level of the full spatiotemporal state, it becomes essential when defining mutual entropies and other related information-theoretic quantities. In particular, for reduced spatiotemporal states, whether normalization is performed before or after taking the reduction can lead to qualitatively different entropic properties. Thus, the choice of normalization prescription must be treated carefully when extending these entropy measures to multipartite spatiotemporal settings.
For a normal spatiotemporal state, such as a Margenau--Hill state, the singular values coincide with the absolute values of its eigenvalues. Hence, the SVD and spectral-magnitude entropies coincide.

\begin{remark}
    Besides the SVD entropy, one can also consider the \emph{Fullwood-Parzygnat (FP) entropy}, originally defined for Hermitian states over time as
\begin{align}
    S_{FP}(\Upsilon)=-\Tr\big(\Upsilon\log|\Upsilon|\big)=-\sum_i\lambda_i\log|\lambda_i| .
\end{align}
For Hermitian $\Upsilon$ this is the real part of complex spectral entropy, $S_{FP}(\Upsilon)=\mathrm{Re}\,[-\Tr(\Upsilon\log\Upsilon)]$, the
imaginary part being $\mp\pi\sum_{\lambda_i<0}|\lambda_i|$. 
For non-Hermitian $\Upsilon$ , we take the spectral form
\begin{align}
    S_{EFP}(\Upsilon):=-\sum_i\lambda_i\log|\lambda_i| \label{def:EFP}
\end{align}
as the \emph{extended Fullwood-Parzygnat}(EFP) entropy , which is in general complex. 
The difference between EFP entropy and SVD entropy is: the SVD entropy is real and non-negative for an arbitrary operator, while the EFP entropy can be complex for non-Hermitian operator. The two coincide, and reduce to the von Neumann entropy, if and only if $\Upsilon\geq0$.
\end{remark}

\subsubsection{SVD  mutual entropy}
\label{sec:svd-mutual-entropies}

Let $\Ups_{\sA\sB}$ be a bipartite spatiotemporal state, with reduced
operators
  $  \Ups_{\sA}:=\Tr_{\sB}\Ups_{\sA\sB}$ and
    $\Ups_{\sB}:=\Tr_{\sA}\Ups_{\sA\sB}$.
There are two inequivalent normalization prescriptions, which must be distinguished carefully.

\paragraph{Reduction followed by normalization.}
The direct SVD mutual-entropy differences are
\begin{align}
    I_{\rm SVD}^{\rm red}(\sA:\sB)_\Ups
    &:=
    S_{\rm SVD}(\Ups_{\sA})
    +S_{\rm SVD}(\Ups_{\sB})
    -S_{\rm SVD}(\Ups_{\sA\sB}),
    \label{eq:raw-svd-mutual-entropy}\\
    \widehat I_{\rm SVD}^{\rm red}(\sA:\sB)_\Ups
    &:=
    \widehat S_{\rm SVD}(\Ups_{\sA})
    +\widehat S_{\rm SVD}(\Ups_{\sB})
    -\widehat S_{\rm SVD}(\Ups_{\sA\sB}).
    \label{eq:reduced-normalized-svd-mutual-entropy}
\end{align}
Their R\'enyi analogues are
\begin{align}
    I_{\rm SVD}^{(q),{\rm red}}(\sA:\sB)_\Ups
    &:=
    S_{\rm SVD}^{(q)}(\Ups_{\sA})
    +S_{\rm SVD}^{(q)}(\Ups_{\sB})
    -S_{\rm SVD}^{(q)}(\Ups_{\sA\sB}),
    \label{eq:raw-svd-renyi-mutual-entropy}\\
    \widehat I_{\rm SVD}^{(q),{\rm red}}(\sA:\sB)_\Ups
    &:=
    \widehat S_{\rm SVD}^{(q)}(\Ups_{\sA})
    +\widehat S_{\rm SVD}^{(q)}(\Ups_{\sB})
    -\widehat S_{\rm SVD}^{(q)}(\Ups_{\sA\sB}).
    \label{eq:normalized-svd-renyi-mutual-entropy}
\end{align}
Analogous Tsallis differences may be defined by replacing
$S_{\rm SVD}^{(q)}$ and $\widehat S_{\rm SVD}^{(q)}$ with
$T_{\rm SVD}^{(q)}$ and $\widehat T_{\rm SVD}^{(q)}$, respectively.

These quantities are useful spectral diagnostics, but they are not mutual
informations in the information-theoretic sense. In particular, neither
positivity nor data processing follows from their definitions. The reason is
that
\begin{equation}\label{eq:magnitude-reduction-noncommuting}
    \frac{|\Ups_{\sA}|}{\|\Ups_{\sA}\|_1}
    \neq
    \Tr_{\sB}\!\left(
    \frac{|\Ups_{\sA\sB}|}{\|\Ups_{\sA\sB}\|_1}
    \right)
\end{equation}
in general. Thus the normalized magnitude of a reduced temporal operator is
not generally a marginal of the normalized global magnitude operator.

The unnormalized quantity has an additional difficulty. Since
\begin{equation}\label{eq:raw-svd-product-law}
    S_{\rm SVD}(X\otimes Y)
    =
    \|Y\|_1 S_{\rm SVD}(X)
    +\|X\|_1 S_{\rm SVD}(Y),
\end{equation}
it is not additive. For unit-trace product operators
$\Ups_{\sA\sB}=X_{\sA}\otimes Y_{\sB}$, one obtains
\begin{equation}
    I_{\rm SVD}^{\rm red}(\sA:\sB)_{X\otimes Y}
    =
    \bigl(1-\|Y\|_1\bigr)S_{\rm SVD}(X)
    +
    \bigl(1-\|X\|_1\bigr)S_{\rm SVD}(Y),
\end{equation}
which need not vanish. By contrast,
$\widehat I_{\rm SVD}^{\rm red}$ vanishes on product operators because the
normalized SVD entropy is additive, but it can still be negative for
correlated operators.

For example, consider the unit-trace Hermitian contraction
\begin{equation}\label{eq:negative-normalized-svd-mi-example}
    \Ups_{\sA\sB}
    =\frac{1}{2}\operatorname{diag}(1,1,1,-1).
\end{equation}
Its reduced operators are
\begin{equation}
    \Ups_{\sA}=\Ups_{\sB}=\operatorname{diag}(1,0),
\end{equation}
and hence
\begin{equation}
    \widehat S_{\rm SVD}(\Ups_{\sA})
    =
    \widehat S_{\rm SVD}(\Ups_{\sB})=0.
\end{equation}
On the other hand, all four normalized singular values of
$\Ups_{\sA\sB}$ are equal to $1/4$, so that
\begin{equation}
    \widehat S_{\rm SVD}(\Ups_{\sA\sB})=\log4.
\end{equation}
Consequently,
\begin{equation}
    \widehat I_{\rm SVD}^{\rm red}(\sA:\sB)_\Ups=-\log4.
\end{equation}

\paragraph{Global normalization followed by reduction.}
A genuine normalized SVD mutual information is obtained by first defining
the single global density operator
\begin{equation}
    \omega_\Ups
    :=\frac{|\Ups_{\sA\sB}|}{\|\Ups_{\sA\sB}\|_1},
\end{equation}
and only afterward taking its marginals:
\begin{equation}\label{eq:global-magnitude-marginals}
    \omega_{\Ups,\sA}:=\Tr_{\sB}\omega_\Ups,
    \qquad
    \omega_{\Ups,\sB}:=\Tr_{\sA}\omega_\Ups.
\end{equation}
Define the regional normalized SVD entropies by
\begin{equation}\label{eq:regional-normalized-svd-entropy}
    \widehat S_{\rm SVD}(\sA)_\Ups
    :=S(\omega_{\Ups,\sA}),
    \qquad
    \widehat S_{\rm SVD}(\sB)_\Ups
    :=S(\omega_{\Ups,\sB}),
\end{equation}
and
\begin{equation}
    \widehat S_{\rm SVD}(\sA\sB)_\Ups
    :=S(\omega_\Ups)
    =\widehat S_{\rm SVD}(\Ups_{\sA\sB}).
\end{equation}
The normalized SVD mutual information is then
\begin{align}
    \widehat I_{\rm SVD}(\sA:\sB)_\Ups
    &:=
    \widehat S_{\rm SVD}(\sA)_\Ups
    +\widehat S_{\rm SVD}(\sB)_\Ups
    -\widehat S_{\rm SVD}(\sA\sB)_\Ups
    \label{eq:normalized-svd-mutual-information}\\
    &=
    S(\omega_{\Ups,\sA})
    +S(\omega_{\Ups,\sB})
    -S(\omega_\Ups)
    \nonumber\\
    &=
    D\!\left(
        \omega_\Ups
        \middle\|
        \omega_{\Ups,\sA}\otimes\omega_{\Ups,\sB}
    \right).
    \label{eq:normalized-svd-relative-entropy}
\end{align}
Because this is the ordinary quantum mutual information of the density
operator $\omega_\Ups$, it satisfies
\begin{equation}
    0\leq
    \widehat I_{\rm SVD}(\sA:\sB)_\Ups
    \leq
    2\min\{\log d_{\sA},\log d_{\sB}\}.
\end{equation}
Moreover,
\begin{equation}
    \widehat I_{\rm SVD}(\sA:\sB)_\Ups=0
    \quad\Longleftrightarrow\quad
    \omega_\Ups
    =
    \omega_{\Ups,\sA}\otimes\omega_{\Ups,\sB}.
\end{equation}
It is invariant under nonzero rescaling of $\Ups$ and under local
left--right unitary transformations
\begin{equation}
    \Ups\longmapsto
    (U_{\sA}\otimes U_{\sB})\,
    \Ups\,
    (V_{\sA}\otimes V_{\sB}),
\end{equation}
and it obeys data processing under local CPTP maps acting on
$\omega_\Ups$. Notice, however, that such a CPTP transformation of
$\omega_\Ups$ need not arise from the same linear transformation of the
original spatiotemporal operator $\Ups$, because the map
\begin{equation}
    \Ups\longmapsto\frac{|\Ups|}{\|\Ups\|_1}
\end{equation}
is nonlinear.

For three regions, one may similarly define the normalized SVD conditional
mutual information
\begin{align}
    \widehat I_{\rm SVD}(\sA:\sC|\sB)_\Ups
    &:=
    S(\omega_{\Ups,\sA\sB})
    +S(\omega_{\Ups,\sB\sC})
    -S(\omega_{\Ups,\sB})
    -S(\omega_{\Ups,\sA\sB\sC})
    \nonumber\\
    &\geq0,
    \label{eq:normalized-svd-conditional-mi}
\end{align}
where the inequality follows from strong subadditivity.

If $\Ups$ is a positive density operator, then $\omega_\Ups=\Ups$, and
$\widehat I_{\rm SVD}$ reduces to the ordinary quantum mutual information.
If $\Ups$ is normal, as for a Margenau--Hill state, then
\begin{equation}
    |\Ups|
    =
    \sum_j|\lambda_j|
    |j\rangle\langle j|,
\end{equation}
so the normalized SVD mutual information coincides with the mutual
information of the normalized spectral-magnitude state. Nevertheless,
Eq.~\eqref{eq:magnitude-reduction-noncommuting} generally remains strict.
Thus, even for Margenau--Hill states, global normalization followed by
reduction must be distinguished from reduction followed by normalization.

\subsection{Examples}\label{sec:example}
To compare different definitions of entropy, we consider a single qubit at $n+1$ time points $t_0,\dots,t_n$ ($n\ge 1$), with the same Pauli channel acting between consecutive times. Let $\{\hat\sigma_\mu\}_{\mu=0}^{3}$ denote the normalized Pauli basis. The initial state is written in its Bloch representation as
\begin{align}
     \varrho_{t_0}=\frac{1}{\sqrt2}\sum_{\mu}r_\mu\hat\sigma_\mu,\quad r_0=1,\quad |\vec r|\le 1,
\end{align}
and we denote its purity by $P_0:=\Tr\bigl(\varrho_{t_0}^2\bigr)=\tfrac12\bigl(1+|\vec r|^2\bigr)$. We consider the Pauli channel
\begin{align}
    \mathcal E(X)=(2\lambda-1)\,X+2(1-\lambda)\sum_{\mu=0}^{3}\omega_\mu\,\hat\sigma_\mu X\hat\sigma_\mu,\qquad \omega_0=1,\quad \omega_x+\omega_y+\omega_z=1 .
\end{align}
It is diagonal in the Pauli basis, $\mathcal E(\hat\sigma_\mu)=h_\mu\hat\sigma_\mu$, with
\begin{align}
    h_0=1,\qquad h_i=2\lambda-1+2\omega_i(1-\lambda),\quad i=x,y,z .
\end{align}
The Jamio\l{}kowski operator and its doubled version are
\begin{align}
    J[\mathcal E]&=\sum_{\mu}h_\mu\,\hat\sigma_\mu\otimes\hat\sigma_\mu,\\
     \mathbb J[\mathcal E_k]&=\sum_{\mu,\nu}h_\mu\,\bigl(\hat\sigma_\mu\hat\sigma_\nu\otimes\hat\sigma_\nu\bigr)_{t_k}\otimes\bigl(\I\otimes\hat\sigma_\mu\bigr)_{t_{k-1}}.
\end{align}
Evaluating the determinants of these two operators reveals that $\lambda = \frac{1}{2}$ is a singular point for arbitrary $\omega_i$.
For convenience, we define
\begin{align}
      Q:=\sum_{\mu}h_\mu^{2},\qquad q_\mu:=\frac{h_\mu^{2}}{Q}.
\end{align}
\begin{proposition}\label{prop:Ivec for doubled KD state}
Let $\Un$ denote the doubled Kirkwood-Dirac operator for the $n$-step process described above, with $n\geq 1$. Consider a temporal subsystem $A\subseteq\{t_0,\ldots,t_n\}$, where each selected time site includes both copies in the doubled construction. The set of bonds $\{t_i,t_{i+1}\}$ crossing the partition $A|A^c$ is
\begin{align}
    \partial A&:= \left\{\ell\in\{1,\ldots,n\}:t_{\ell}\in A,t_{\ell}\notin A\right\},& d_A &:= |\partial A|,
\end{align}
Thus a bond belongs to $\partial A$ precisely when exactly one of its endpoints belongs to $A$. The normalized squared operator-Schmidt coefficients of $\Un$ across $A|A^c$ are
\begin{align}
    p_{\boldsymbol\mu_{\partial A}}&= \prod_{\ell\in\partial A}q_{\mu_\ell}.
\end{align}
Consequently, the vectorized spatiotemporal entropy which is defined in\eqref{eq:STdef}  is independent of the initial Bloch vector .
\end{proposition}
The proof is given in Section \ref{app:Ivec} in Appendix.
\begin{remark}
  The vectorized spatiotemporal entropy of right, left, and mixed-branch Kirkwood–Dirac temporal states is not generally independent of the initial state. This can already be seen in a two-time process, for which the mixed-branch construction reduces to either the right or the left construction. Consider the right temporal state
\begin{align}
    \Upsr_{t_1t_0}=\sum_{\mu}\hs_{\mu}\otimes h_{\mu} \hs_{\mu}\varrho_{t_0},
\end{align}
where $\hs_{\mu}$ form an orthonormal basis with respect to the Hilbert–Schmidt inner product. The Gram matrix is
\begin{align}
    \Sigma_{\mu\nu}&=\Tr\!\left[(h_{\mu}\hs_{\mu}\varrho_{t_0})^\dagger(h_{\nu}\hs_{\nu}\varrho_{t_0})\right]\notag\\
    &=\frac{1}{2}\sum_{\alpha,\beta}h_{\mu}h_{\nu}r_{\alpha}r_{\beta} \Tr(\hs_{\alpha}\hs_{\mu}\hs_{\nu}\hs_{\beta}).
\end{align}
Upon diagonalization, the spectrum of $\Sigma$ depends explicitly on both $h$ and $r$; consequently, $S_{\mathrm{vec}}$ will depend on the initial Bloch-vector parameters.
\end{remark}
In what follows, we examine several entropy prescriptions for the right, mixed-branch, and doubled Kirkwood–Dirac temporal states. We consider three-time processes generated by either depolarizing or dephasing channels. The initial Bloch-vector components are fixed at $r_x=0.3$ and $r_y=0.4$, while \(r_z\) varies over the physical range $|r_z|\leq\sqrt{3}/2$. The numerical analysis below focuses on the doubled construction; the corresponding results for the right and mixed-branch constructions are presented in Appendix~\ref{app:KD-comparison}. The entropy prescriptions considered are the complex spectral entropy, the unnormalized spectral-magnitude entropy,  the unnormalized and normalized singular-value entropies, and the extended Fullwood–Parzygnat entropy. We evaluate the corresponding mutual information and additionally consider the mutual information defined through the vectorized spatiotemporal entropy.
%%%%%%%%
\subsubsection{Depolarizing channel}
The depolarizing channel is given by
\begin{align}
    \cE_\lambda(X)=\frac{4\lambda-1}{3}X+\frac{2(1-\lambda)}{3}\Tr(X)\I.
\end{align}
At \(\lambda=1\), the channel reduces to the identity. At \(\lambda=\tfrac14\), it becomes completely depolarizing, mapping every input density matrix to \(\I/2\), the doubled Kirkwood-Dirac temporal state can be simplified into product state
\begin{align}
   \Upsd_{t_2t_1t_0}&=\frac{\mathbb{S}}{2}\otimes\frac{\mathbb{S}}{2}\otimes(\varrho_{t_0}\otimes\I)\mathbb{S}.
\end{align}
The entropy figures of doubled Kirkwood-Dirac temporal state are showed in Figure~\ref{fig:doubled-entropy-depolarizing}.
\begin{figure}
    \centering
    \includegraphics[width=1\linewidth]{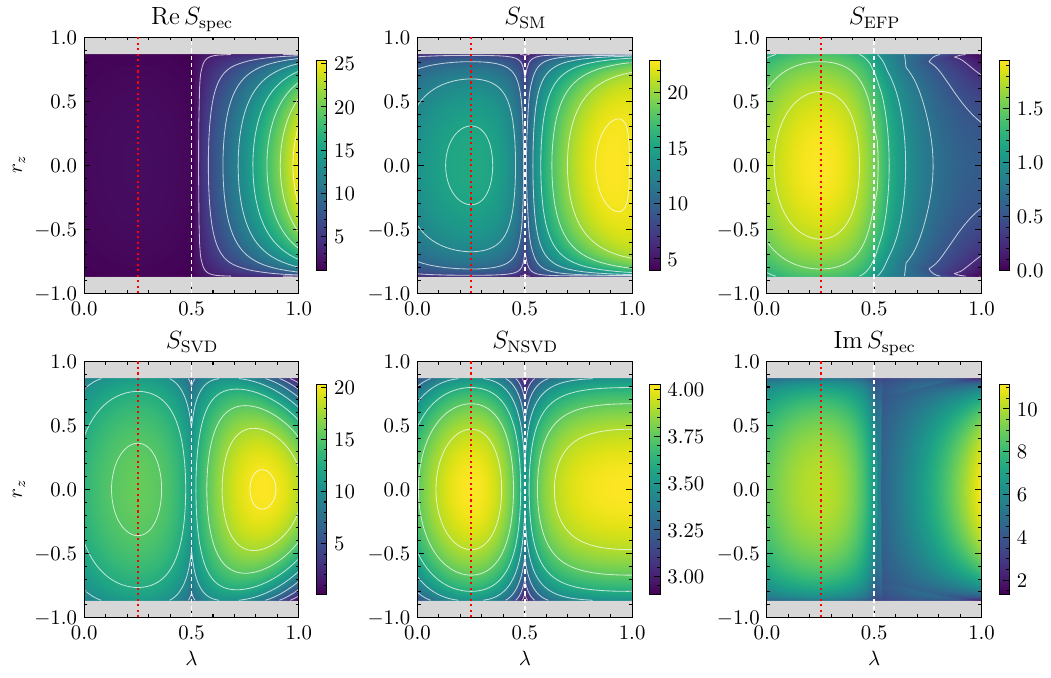}
    \caption{Entropy maps for the three-time doubled Kirkwood-Dirac temporal state generated by depolarizing channels, with initial Bloch-vector components $r_x=0.3$ and $r_y=0.4$. The panels show the real part of the complex spectral entropy, the extended Fullwood-Parzygnat entropy, the unnormalized and normalized singular-value entropies, and the imaginary part of the complex spectral entropy. Gray regions lie outside the physical range $|r_z|\leq\sqrt{3}/2$. The red dotted and white dashed lines mark $\lambda=1/4$ and $\lambda=1/2$, respectively. A difference between $\operatorname{Re}S_{\mathrm{spec}}$ and $S_{\mathrm{EFP}}$ signals non-real eigenvalues, while a nonzero $\operatorname{Im}S_{\mathrm{spec}}$ indicates negative real or non-real eigenvalues.}
    \label{fig:doubled-entropy-depolarizing}
\end{figure}
The entropy maps show that none of the entropy prescriptions considered vanish at the completely depolarizing point, $\lambda=1/4$. To examine correlations between different temporal subsystems more directly, we next compare the mutual-information quantities constructed from the entropy prescriptions above, as well as from the vectorized spatiotemporal entropy. The following figures present the results for the temporal bipartition $A=t_1t_0$ and $B=t_2$ which is given in Figure~\ref{fig:doubled-MI-depolarizing}.
\begin{figure}
    \centering
    \includegraphics[width=1\linewidth]{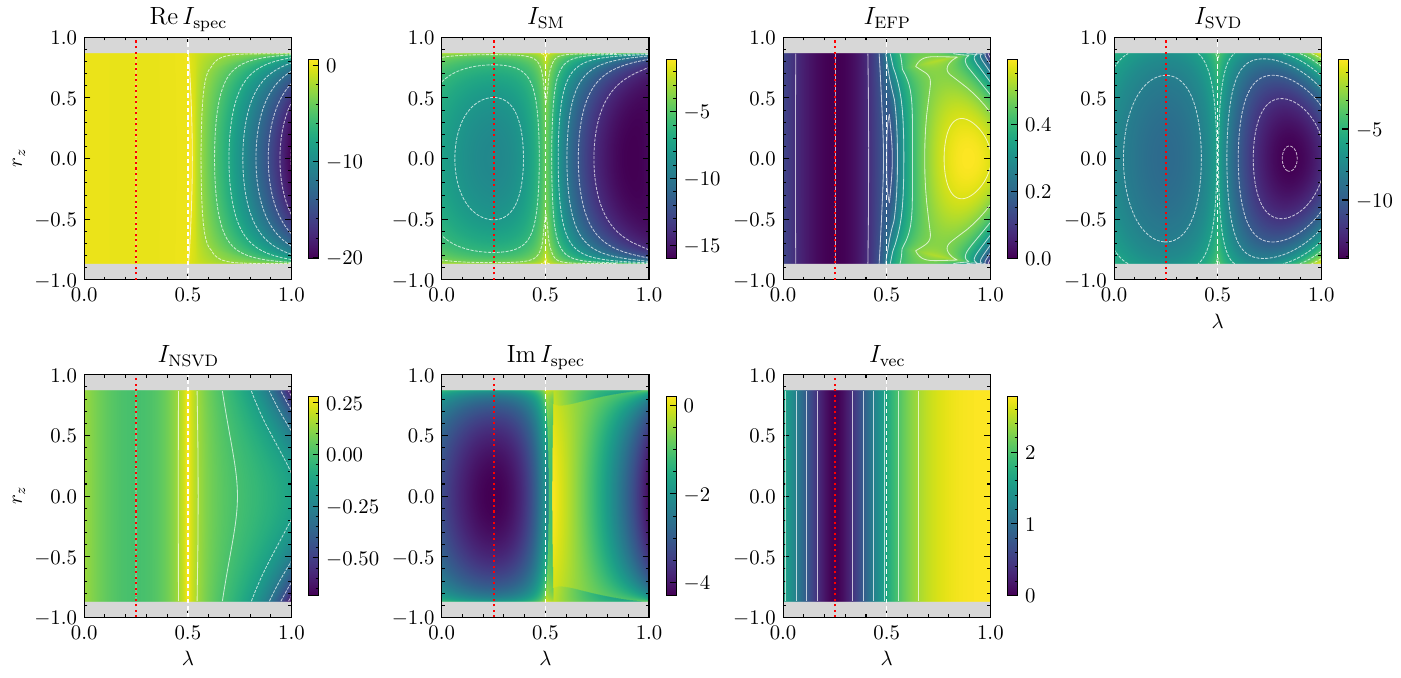}
    \caption{Mutual-information maps for the three-time doubled Kirkwood-Dirac temporal state generated by depolarizing channels. At the completely depolarizing point, $\lambda=1/4$, $\operatorname{Re}I_{\mathrm{spec}}$, $I_{\mathrm{EFP}}$, $I_{\mathrm{NSVD}}$, and $I_{\mathrm{vec}}$ vanish, whereas $\operatorname{Im}I_{\mathrm{spec}}$, $I_{\mathrm{SM}}$ and $I_{\mathrm{SVD}}$ remain nonzero. These residual values can arise even for a product operator, owing respectively to the branch dependence of the complex spectral entropy and the nonadditivity of the unnormalized spectral magnitude and singular-value entropies. They therefore do not by themselves indicate residual temporal correlations. }
    \label{fig:doubled-MI-depolarizing}
\end{figure}

%%%%%%%%%
\subsubsection{Dephasing channel}
 The dephasing channel is given by
\begin{align}
    \cE(X)
    &=\lambda X+2(1-\lambda)\hs_zX\hs_z,
\end{align}
For $X=\varrho_{t_0}$, the channel acts as
\begin{align}
    \cE(\varrho_{t_0})
    =\frac{1}{\sqrt{2}}
    \left[
        \hs_0+r_z\hs_z
        +(2\lambda-1)(r_x\hs_x+r_y\hs_y)
    \right].
\end{align}
At $\lambda=1/2$ and $r_z=0$,  the channel coefficients satisfy
\begin{align}
    h_0=h_z=1,
    \qquad h_x=h_y=0,
\end{align}
and the output is maximally mixed and therefore coincides with the output of the completely depolarizing channel. Nevertheless, the dephasing channel preserves the $z$ component, so this agreement of single-time states does not imply the absence of temporal correlations. To examine these correlations, we consider the doubled Kirkwood-Dirac temporal state
\begin{align}
    \Upsd_{t_2t_1t_0}
    &=
    \sum_{\mu_2,\mu_1\in\{0,z\}}
    F_{\mu_2}\otimes K_{\mu_2\mu_1}\otimes G_{\mu_1}
    \notag\\
    &=
    \frac{\mathbb S}{2}\otimes
    \left(
        \frac{\mathbb S}{\sqrt{2}}\otimes G_0
        +F_z\otimes G_z
    \right)
    \notag\\
    &\quad+
    F_z\otimes
    \left(
        K_{z0}\otimes G_0
        +K_{zz}\otimes G_z
    \right).
\end{align}
For $r_z=0$, tracing over the respective complementary subsystems gives
\begin{align}
    \Upsd_{t_2}
    &=\Tr_A\!\left(\Upsd_{t_2t_1t_0}\right)
      =\frac{\mathbb S}{2},\\
    \Upsd_{t_1t_0}
    &=\Tr_{A^c}\!\left(\Upsd_{t_2t_1t_0}\right)
      =\frac{\mathbb S}{\sqrt{2}}\otimes G_0
       +F_z\otimes G_z.
\end{align}
Consequently,
\begin{align}
    \Upsd_{t_2t_1t_0}
    -\Upsd_{t_2}\otimes\Upsd_{t_1t_0}
    =
    F_z\otimes
    \left(
        K_{z0}\otimes G_0
        +K_{zz}\otimes G_z
    \right)
    \neq0.
\end{align}
Thus, the doubled temporal state does not factorize across $t_2\,|\,t_1t_0$, demonstrating that temporal correlations persist even though the single-time physical output is maximally mixed.
\begin{figure}
    \centering
    \includegraphics[width=1\linewidth]{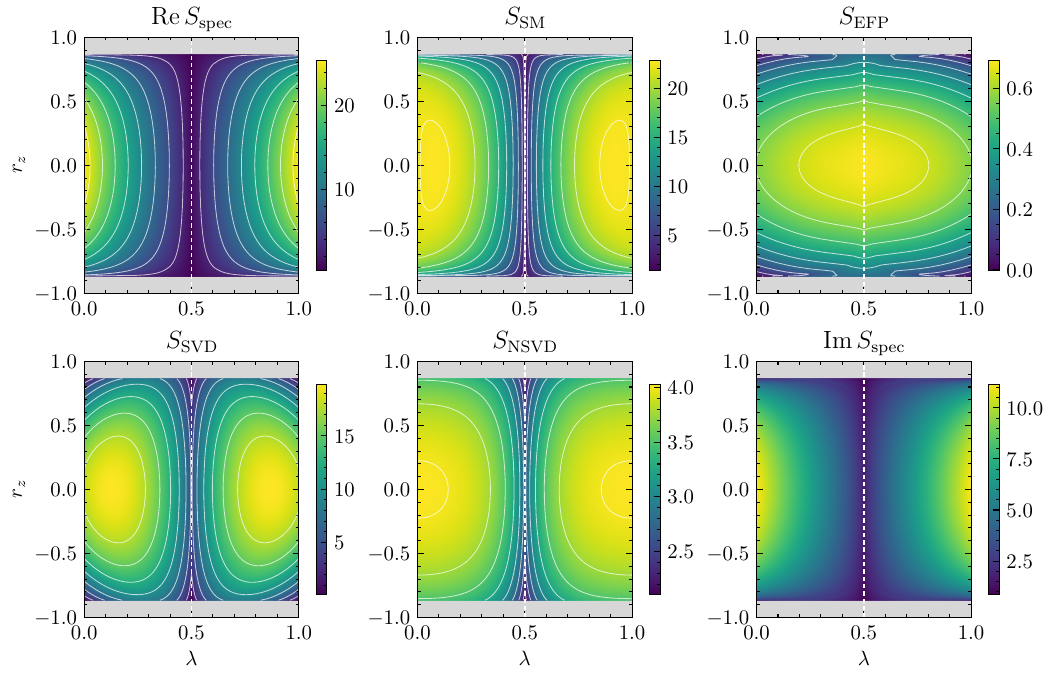}
    \caption{The channel preserves the longitudinal component $r_z$ and multiplies the transverse components by $2\lambda-1$. Complete dephasing occurs at $\lambda=1/2$, since the off-diagonal elements vanish, marked by the white dashed line. The extended Fullwood-Parzygnat entropy exhibits a distinct contour pattern: for fixed $r_z$, it is maximal at complete dephasing and decreases toward the unitary endpoints $\lambda=0,1$, whereas the other panels generally display central minima. Its comparatively small scale reflects cancellations in the signed spectral sum $S_{\mathrm{EFP}}=-\sum_j \lambda_j\ln|\lambda_j|$, including exact cancellation between opposite eigenvalue pairs.}
    \label{fig:doubled-entropy-dephasing}
\end{figure}
The corresponding entropy and mutual information figures are shown in Figure~\ref{fig:doubled-entropy-dephasing} and Figure~\ref{fig:doubled-MI-dephasing}.
\begin{figure}
    \centering
    \includegraphics[width=1\linewidth]{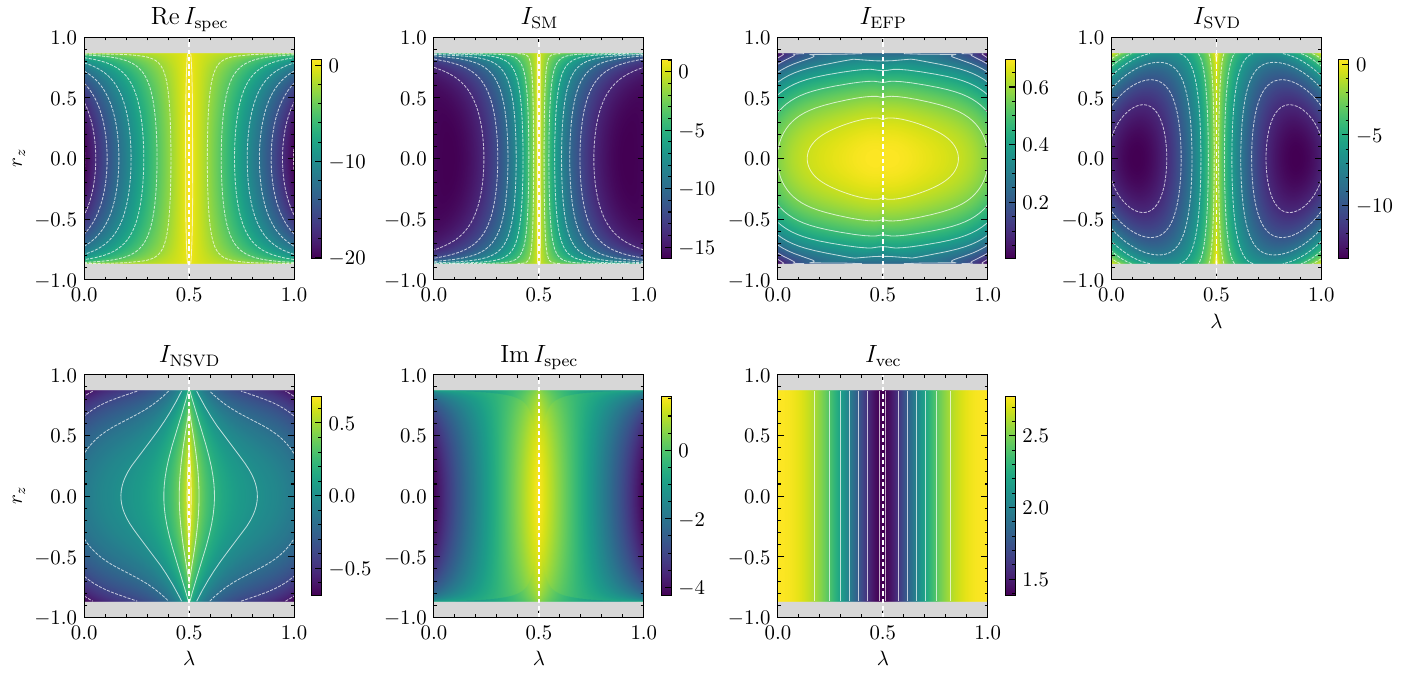}
    \caption{At the complete-dephasing point $\lambda=1/2$ and $r_z=0$, all plotted quantities remain nonzero. In particular, $I_{\mathrm{EFP}}$ and $I_{\mathrm{NSVD}}$ attain their maxima, whereas others attain a nonzero minimum. The different extrema illustrate the dependence on the entropy prescription. Temporal correlations persist because dephasing preserves the $z$ populations even when transverse coherences are completely suppressed. }
    \label{fig:doubled-MI-dephasing}
\end{figure}

%%%%%%%%%
\section{Conclusion and outlook}

We have developed an operator framework in which Wightman
correlators, spatiotemporal Kirkwood--Dirac quasiprobabilities, and spatiotemporal states are three representations of the same data.  The padded link product gives recursive constructions for left, right, mixed-branch, doubled, and many-fold states while retaining every sampled time as an explicit tensor factor.

Two complementary notions of temporal entanglement emerge.  Vectorization turns any spatiotemporal operator into a pure state on a doubled operator Hilbert space, so standard entanglement entropy, mutual information apply directly.  The intrinsic spectral approach instead exposes the nonpositive and nonnormal structure of the operator itself.  For memoryless spatiotemporal state we proved that the
operator norm is at most one and hence that every eigenvalue lies in the unit disk.  The temporality separates exactly into a contribution visible in the eigenvalue moduli and a nonnormal contribution, the latter vanishing if and only if the temporal state is normal.  This yields a precise distinction between complex spectral entropy, spectral-magnitude entropy, SVD entropy, and their normalized counterparts.  On the Margenau--Hill case the spectral-magnitude and SVD hierarchies coincide; for general Kirkwood--Dirac states they do not. 

Several directions remain open. Extending the framework to infinite-dimensional systems is an important direction for future work, particularly for applications to quantum field theory. 
The spatiotemporal framework may also provide a useful approach to quantum causal inference. Moreover, given the notion of quantum states in the temporal direction, it is natural to ask whether the concept of quantum phases can also be extended to time. In particular, it would be interesting to investigate whether time crystals can be characterized within the spatiotemporal-state framework. We leave these questions for future work.

\begin{acknowledgments}
Z.J. is supported by the National Natural Science Foundation of China (Grant
No. 12605015), Start-up Grant of Central South University (Grant No. 502045031) and by the Frontier Interdisciplinary Direction ``Quantum Technologies'' Project of Central South University (Grant No. 506010805). Z. J. would like to thank Dagomir Kaszlikowski, Kavan Modi and Kelvin Onggadinata for discussions on related topics. 
\end{acknowledgments}

\appendix

%%%%============

\section{Some proofs and technical results}
\label{app:proofs}

In this section, we provide detailed proofs for several results and some technical discussion.
Recall that duality identity for the Schatten 1-norm reads
\begin{equation}\label{eq:DualHolder}
    \|X\|_1 = \sup_{\|Z\|_{\infty}\leq 1} |\Tr Z^{\dagger} X|
    =\sup_{\|Z\|_{\infty}\leq 1} \operatorname{Re}\Tr Z^{\dagger} X,
\end{equation}
we will use it in several proofs.

\subsection{Proof of Proposition~\ref{prop:temporality-properties}}
\begin{proof}
1-4 are straightforward.

5. For both Hermitian and non-Hermitian \(\Ups\), by the reverse triangle inequality, for every positive semidefinite \(\rho\) with \(\Tr\rho=1\),
\begin{equation}
  \|\Ups-\rho\|_1 \ge \|\Ups\|_1-\|\rho\|_1 = \|\Ups\|_1-1 =: 2\mathfrak{T}(\Ups).  
\end{equation}
Thus
\begin{equation}
2\mathfrak{T}(\Ups) \le \inf_{\rho\ge0,\,\Tr\rho=1}\|\Ups-\rho\|_1.  
\end{equation}

For Hermitian \(\Ups\) with Hahn--Jordan decomposition $\Ups=\Ups_+-\Ups_-$, we have
\(\Tr \Ups=\Tr\Ups_+-\Tr\Ups_-=1\), and
\(\|\Ups\|_1-1=2\Tr\Ups_-=:2r\). Then $\Tr\Ups_+=1+r$ and $\|\Ups\|_1
 =\Tr\Ups_++\Tr\Ups_-
 =1+2r$.
Define the density operator
\begin{equation}
 \rho_*:=\frac{\Ups_+}{\Tr\Ups_+}
 =\frac{\Ups_+}{1+r}.
\end{equation}
Since the supports of $\Ups_+$ and $\Ups_-$ are orthogonal,
\begin{equation}
 \Ups-\rho_*
 =
 \frac{r}{1+r}\Ups_+-\Ups_-
\end{equation}
is itself a Hahn--Jordan decomposition. Consequently,
\begin{align}
 \|\Ups-\rho_*\|_1
 &=
 \frac{r}{1+r}\Tr\Ups_+
 +\Tr\Ups_-
 =2r
 =2\mathfrak{T}(\Ups).
\end{align}
The infimum is therefore attained by
$\rho_*=\Ups_+/\Tr\Ups_+$, which proves
Eq.~\eqref{eq:temporality-jordan}.

6. Let $\Phi$ be a CPTP map. Its adjoint $\Phi^\dagger$ is unital and
completely positive (UCP), and is therefore contractive with respect to the
operator norm:
\begin{equation}\label{eq:UPCnorm}
    \|\Phi^\dagger(X)\|_{\infty}\leq \|X\|_{\infty}.
\end{equation}
To see this, we use Choi's generalization \cite{choi1980some} of Kadison's
inequality \cite{kadison1952generalized}, which states that for any unital positive map
$\Psi$,
\begin{equation}
    \Psi(A^\dagger A)\geq \Psi(A^\dagger)\Psi(A).
\end{equation}
Let $\|X\|_\infty=\lambda$. Then
\begin{equation}
    X^\dagger X\leq \lambda^2\mathds{1}.
\end{equation}
By the positivity and unitality of $\Phi^\dagger$,
\begin{equation}
    \Phi^\dagger(X^\dagger X)
    \leq \Phi^\dagger(\lambda^2\mathds{1})
    =\lambda^2\mathds{1}.
\end{equation}
Applying Kadison's inequality to $\Phi^\dagger$, we obtain
\begin{equation}
    \Phi^\dagger(X)^\dagger\Phi^\dagger(X)
    \leq \Phi^\dagger(X^\dagger X)
    \leq \lambda^2\mathds{1}.
\end{equation}
Taking the operator norm gives
\begin{equation}
    \|\Phi^\dagger(X)\|_\infty^2
    =\|\Phi^\dagger(X)^\dagger\Phi^\dagger(X)\|_\infty
    \leq \lambda^2
    =\|X\|_\infty^2,
\end{equation}
which proves Eq.~\eqref{eq:UPCnorm}.

We now use this result to bound $\|\Phi(X)\|_1$. By the duality between the
trace norm and the operator norm,
\begin{align}
    \|\Phi(X)\|_1
    &=\max_{\|Z\|_\infty\leq 1}
    \left|\Tr\!\left[Z^\dagger\Phi(X)\right]\right|\\
    &=\max_{\|Z\|_\infty\leq 1}
    \left|\Tr\!\left[\Phi^\dagger(Z^\dagger)X\right]\right|\\
    &=\max_{\|Z\|_\infty\leq 1}
    \left|\Tr\!\left[\Phi^\dagger(Z)^\dagger X\right]\right|.
\end{align}
For every $Z$ satisfying $\|Z\|_\infty\leq1$,
Eq.~\eqref{eq:UPCnorm} implies
\begin{equation}
    \|\Phi^\dagger(Z)\|_\infty\leq1.
\end{equation}
Therefore,
\begin{equation}
    \left\{\Phi^\dagger(Z):\|Z\|_\infty\leq1\right\}
    \subseteq
    \left\{W:\|W\|_\infty\leq1\right\}.
\end{equation}
It follows that
\begin{align}
    \|\Phi(X)\|_1
    &\leq
    \max_{\|W\|_\infty\leq1}
    \left|\Tr\!\left[W^\dagger X\right]\right|\\
    &=\|X\|_1.
\end{align}
Hence every CPTP map is contractive with respect to the trace norm:
\begin{equation}
    \|\Phi(X)\|_1\leq\|X\|_1.
\end{equation}
Therefore, the temporality is non-increasing under CPTP maps.

7. This is clear from the fact that norm of all frames are bounded by 1.
\end{proof}

\subsection{Proof of Eq.~\eqref{eq:kms-petz-reflection}}

Since $\cE_t(\varrho)=\varrho$, the Petz recovery map associated with
$\varrho$ and $\cE_t$ is
\begin{align}
\mathcal R_{\varrho,\cE_t}(X)
&=
\varrho^{1/2}
\cE_t^\dagger\!\left(
\varrho^{-1/2}X\varrho^{-1/2}
\right)
\varrho^{1/2}.
\end{align}
For the unitary channel
\begin{equation}
    \cE_t(X)=U_t X U_t^\dagger,
\end{equation}
its Hilbert--Schmidt adjoint is
\begin{equation}
    \cE_t^\dagger(X)=U_t^\dagger XU_t.
\end{equation}
Hence
\begin{align}
\mathcal R_{\varrho,\cE_t}(X)
&=
\varrho^{1/2}
U_t^\dagger
\varrho^{-1/2}
X
\varrho^{-1/2}
U_t
\varrho^{1/2}.
\end{align}
Using $[U_t,\varrho]=0$, and therefore
$[U_t,\varrho^\alpha]=0$ for every real $\alpha$, we obtain
\begin{align}
\mathcal R_{\varrho,\cE_t}(X)
&=
U_t^\dagger
\varrho^{1/2}\varrho^{-1/2}
X
\varrho^{-1/2}\varrho^{1/2}
U_t
\nonumber\\
&=
U_t^\dagger XU_t
=
\cE_{-t}(X).
\end{align}

We now consider the temporal operator
\begin{equation}
\Ups^{(s)}[\varrho,\cE_t]
=
(U_t\otimes I)
(\varrho^{1-s}\otimes\varrho^s)
\mathbb S
(U_t^\dagger\otimes I).
\end{equation}
Using
\begin{equation}
    \mathbb S(A\otimes B)\mathbb S
    =
    B\otimes A,
\end{equation}
together with
\begin{equation}
    \mathbb S(U_t\otimes I)
    =
    (I\otimes U_t)\mathbb S,
\end{equation}
we find
\begin{align}
\mathbb S\,
\Ups^{(s)}[\varrho,\cE_t]\,
\mathbb S
&=
\mathbb S
(U_t\otimes I)
(\varrho^{1-s}\otimes\varrho^s)
\mathbb S
(U_t^\dagger\otimes I)
\mathbb S
\nonumber\\
&=
(I\otimes U_t)
(\varrho^s\otimes\varrho^{1-s})
(U_t^\dagger\otimes I)
\mathbb S
\nonumber\\
&=
(\varrho^sU_t^\dagger
 \otimes
 U_t\varrho^{1-s})
\mathbb S.
\end{align}
Since $[U_t,\varrho]=0$, this becomes
\begin{equation}\label{eq:swap-reflected-form}
\mathbb S\,
\Ups^{(s)}[\varrho,\cE_t]\,
\mathbb S
=
(U_t^\dagger\varrho^s
 \otimes
 \varrho^{1-s}U_t)
\mathbb S.
\end{equation}
On the other hand,
\begin{align}
\Ups^{(1-s)}[\varrho,\cE_{-t}]
&=
(U_t^\dagger\otimes I)
(\varrho^s\otimes\varrho^{1-s})
\mathbb S
(U_t\otimes I)
\nonumber\\
&=
(U_t^\dagger\otimes I)
(\varrho^s\otimes\varrho^{1-s})
(I\otimes U_t)
\mathbb S
\nonumber\\
&=
(U_t^\dagger\varrho^s
 \otimes
 \varrho^{1-s}U_t)
\mathbb S.
\end{align}
Comparing with Eq.~\eqref{eq:swap-reflected-form}, we conclude that
\begin{equation}
    \mathbb S\,
    \Ups^{(s)}[\varrho,\cE_t]\,
    \mathbb S
    =
    \Ups^{(1-s)}[\varrho,\cE_{-t}].
\end{equation}

\subsection{Visible space of $\cM_{\varrho,\mathbb{P}}$ in Eq.~\eqref{eq:visible-operator-spaces}}
\label{app:MrhoP}
We now give more details about two-time temporal state and its entanglement entropy.
Let
\begin{equation}
  \varrho=\sum_{a=1}^{q}\lambda_a\ket{a}\!\bra{a},
  \qquad
  \lambda_a>0,
\end{equation}
and let \(P=\sum_{a=1}^{q}\ket{a}\!\bra{a}\) be the support
projector of \(\varrho\).  We also write
$  Q:=\mathds{1}-P$.
Recall that
\begin{equation}
  \cM_{\mathbb P,\varrho}(X)
  :=
  \int_{[0,1]}
  \varrho^sX\varrho^{1-s}\,\mathbb P(s)ds,
\end{equation}
with
\begin{equation}
  p_0:=\mathbb P(\{0\}),
  \qquad
  p_1:=\mathbb P(\{1\}),
\end{equation}
and endpoint conventions
\begin{equation}
  \cM_{0,\varrho}(X)=X\varrho,
  \qquad
  \cM_{1,\varrho}(X)=\varrho X.
\end{equation}
Extend the eigenbasis of \(\varrho\) to an orthonormal basis of
\(\cH\), setting \(\lambda_a=0\) for \(a>q\).  The matrix units are
eigenoperators of $\cM_{\varrho,\mathbb{P}}$:
\begin{equation}\label{eq:modular-map-matrix-units}
  \cM_{\mathbb P,\varrho}(\ket a\!\bra b)
  =
  m_{ab}^{\mathbb P}\ket a\!\bra b,
\end{equation}
where
\begin{equation}\label{eq:modular-map-weights-nonfaithful}
  m_{ab}^{\mathbb P}
  =
  p_0\lambda_b+p_1\lambda_a
  +
  \int_{(0,1)}
  \lambda_a^s\lambda_b^{1-s}\,d\mathbb P(s).
\end{equation}
Consequently, \(\ket a\!\bra b\) belongs to
\(\operatorname{range}\cM_{\mathbb P,\varrho}\) precisely when
\(m_{ab}^{\mathbb P}>0\).
There are four possibilities:
\begin{equation}\label{eq:visible-matrix-unit-table}
\begin{array}{c|c|c}
 a & b & m_{ab}^{\mathbb P}>0 \\ \hline
 a\in\operatorname{supp}\varrho
   & b\in\operatorname{supp}\varrho
   & \text{always} \\[2pt]
 a\in\ker\varrho
   & b\in\operatorname{supp}\varrho
   & p_0>0 \\[2pt]
 a\in\operatorname{supp}\varrho
   & b\in\ker\varrho
   & p_1>0 \\[2pt]
 a\in\ker\varrho
   & b\in\ker\varrho
   & \text{never}
\end{array}
\end{equation}
Indeed, when \(a,b\in\operatorname{supp}\varrho\), both
\(\lambda_a\) and \(\lambda_b\) are strictly positive, so
\(m_{ab}^{\mathbb P}>0\).  When
\(a\in\ker\varrho\) and
\(b\in\operatorname{supp}\varrho\), the only possible nonzero
contribution in Eq.~\eqref{eq:modular-map-weights-nonfaithful} is
\(p_0\lambda_b\).  Similarly, when
\(a\in\operatorname{supp}\varrho\) and
\(b\in\ker\varrho\), the only possible nonzero contribution is
\(p_1\lambda_a\).  If \(a,b\in\ker\varrho\), every term vanishes.

Since \(\cM_{\mathbb P,\varrho}\) is diagonal in the matrix-unit
basis, its range is the span of the matrix units with nonzero weights.
It follows that
\begin{equation}\label{eq:visible-operator-spacesV}
  \operatorname{range}\cM_{\mathbb P,\varrho}
  =
  \begin{cases}
    P\BB(\cH)P,
      &p_0=p_1=0,\\[2pt]
    \BB(\cH)P,
      &p_0>0,\ p_1=0,\\[2pt]
    P\BB(\cH),
      &p_0=0,\ p_1>0,\\[2pt]
    P\BB(\cH)+\BB(\cH)P,
      &p_0>0,\ p_1>0.
  \end{cases}
\end{equation}
To compute the dimensions, use the orthogonal decomposition
\begin{equation}
  \cH=P\cH\oplus Q\cH,
  \qquad
  \dim P\cH=q,
  \qquad
  \dim Q\cH=d-q.
\end{equation}
Every operator \(X\in\BB(\cH)\) has the block form
\begin{equation}
  X=
  \begin{pmatrix}
    X_{PP} & X_{PQ}\\
    X_{QP} & X_{QQ}
  \end{pmatrix}.
\end{equation}
First,
\begin{equation}
  PXP=
  \begin{pmatrix}
    X_{PP} & 0\\
    0 & 0
  \end{pmatrix}.
\end{equation}
Since \(X_{PP}\) is a \(q\times q\) matrix,
\begin{equation}
  \dim P\BB(\cH)P=q^2.
\end{equation}
Next,
\begin{equation}
  XP=
  \begin{pmatrix}
    X_{PP} & 0\\
    X_{QP} & 0
  \end{pmatrix}.
\end{equation}
Such an operator has \(d\) rows and \(q\) unrestricted columns, and
therefore
\begin{equation}
  \dim\BB(\cH)P=dq.
\end{equation}
Similarly,
\begin{equation}
  PX=
  \begin{pmatrix}
    X_{PP} & X_{PQ}\\
    0 & 0
  \end{pmatrix}
\end{equation}
has \(q\) unrestricted rows and \(d\) columns, so
\begin{equation}
  \dim P\BB(\cH)=qd.
\end{equation}
Finally,
\begin{equation}
  P\BB(\cH)+\BB(\cH)P
  =
  \left\{
  \begin{pmatrix}
    X_{PP} & X_{PQ}\\
    X_{QP} & 0
  \end{pmatrix}
  \right\}.
\end{equation}
The three unrestricted blocks have dimensions
\begin{equation}
  \dim X_{PP}=q^2,
  \qquad
  \dim X_{PQ}=q(d-q),
  \qquad
  \dim X_{QP}=(d-q)q.
\end{equation}
Hence
\begin{align}
  \dim\!\left(
    P\BB(\cH)+\BB(\cH)P
  \right)
  &=
  q^2+q(d-q)+(d-q)q\\
  &=
  2qd-q^2.
\end{align}
Equivalently,
\begin{align}
  \dim\!\left(
    P\BB(\cH)+\BB(\cH)P
  \right)
  &=
  \dim P\BB(\cH)
  +\dim\BB(\cH)P \nonumber\\
  &\quad
  -\dim\!\left(
    P\BB(\cH)\cap\BB(\cH)P
  \right)\\
  &=
  qd+qd-q^2,
\end{align}
because
\begin{equation}
  P\BB(\cH)\cap\BB(\cH)P
  =
  P\BB(\cH)P.
\end{equation}
Therefore,
\begin{equation}\label{eq:visible-space-dimensions}
  \dim\operatorname{range}\cM_{\mathbb P,\varrho}
  =
  \begin{cases}
    q^2,
      &p_0=p_1=0,\\[2pt]
    qd,
      &p_0>0,\ p_1=0,\\[2pt]
    qd,
      &p_0=0,\ p_1>0,\\[2pt]
    2qd-q^2,
      &p_0>0,\ p_1>0.
  \end{cases}
\end{equation}
The corresponding kernels are
\begin{equation}\label{eq:invisible-operator-spaces}
  \ker\cM_{\mathbb P,\varrho}
  =
  \begin{cases}
    \{X\in\BB(\cH):PXP=0\},
      &p_0=p_1=0,\\[2pt]
    \BB(\cH)Q,
      &p_0>0,\ p_1=0,\\[2pt]
    Q\BB(\cH),
      &p_0=0,\ p_1>0,\\[2pt]
    Q\BB(\cH)Q,
      &p_0>0,\ p_1>0.
  \end{cases}
\end{equation}

The temporal operator depends on the dynamics only through the
composition
\begin{equation}
  \cE\circ\cM_{\mathbb P,\varrho}.
\end{equation}
Thus two maps \(\cE\) and \(\cF\) generate the same two-time
temporal operator if and only if
\begin{equation}
  \cE(X)=\cF(X)
  \qquad
  \text{for every }
  X\in\operatorname{range}\cM_{\mathbb P,\varrho}.
\end{equation}
Accordingly, the invisible dynamics is determined by the operator
directions annihilated by
\(\cM_{\mathbb P,\varrho}\), rather than merely by the vector-space
orthogonal complement of \(\operatorname{supp}\varrho\).  In
particular, depending on the endpoint weights \(p_0\) and \(p_1\), the
cross-coherence sectors
\begin{equation}
  P\BB(\cH)Q
  \qquad\text{and}\qquad
  Q\BB(\cH)P
\end{equation}
may remain visible.

\subsection{Proof of Proposition~\ref{prop:TandR}}
\label{app:proofofTandR}

\begin{proposition}
For real vector space $\operatorname{Herm}(\sA\sB)$,  $\|\cdot\|_{\rm SEP}$ in Eq.~\eqref{eq:sep-base-norm} is a norm.
\end{proposition}

\begin{proof}

\emph{Nonnegativity and definiteness.}
Nonnegativity follows immediately from the definition.
For any admissible decomposition
\begin{equation}
X=\sum_jc_j\,\rho_j^{\sA}\otimes\tau_j^{\sB},
\end{equation}
the trace-norm triangle inequality gives
\begin{equation}
\begin{split}
\|X\|_1
&\leq
\sum_j|c_j|\,
\|\rho_j^{\sA}\otimes\tau_j^{\sB}\|_1
=
\sum_j|c_j|.
\end{split}
\end{equation}
Taking the infimum
over all admissible decompositions yields
\begin{equation}\label{eq:trace-below-sep}
\|X\|_1\leq\|X\|_{\rm SEP}.
\end{equation}
Consequently, if $\|X\|_{\rm SEP}=0$, then $\|X\|_1=0$, which implies
$X=0$. Conversely, the zero operator has norm zero by taking all
coefficients to vanish. Thus
$\|X\|_{\rm SEP}=0
\quad\Longleftrightarrow\quad
X=0$.

\emph{Absolute homogeneity.}
Let $\alpha\in\mathbb R$. Scaling any admissible decomposition of $X$ by
$\alpha$ gives an admissible decomposition of $\alpha X$, and therefore
\begin{equation}
\|\alpha X\|_{\rm SEP}
\leq
|\alpha|\,\|X\|_{\rm SEP}.
\end{equation}
If $\alpha\neq0$, applying the same inequality to
$X=\alpha^{-1}(\alpha X)$ gives
\begin{equation}
\|X\|_{\rm SEP}
\leq
\frac{1}{|\alpha|}\|\alpha X\|_{\rm SEP}.
\end{equation}
Hence
\begin{equation}
\|\alpha X\|_{\rm SEP}
=
|\alpha|\,\|X\|_{\rm SEP}.
\end{equation}
The result is immediate when $\alpha=0$.

\emph{Triangle inequality.}
Let $X,Y\in\operatorname{Herm}(\sA\sB)$ and let
$\varepsilon>0$. By the definition of the infimum, there exist admissible
decompositions
\begin{equation}
X=\sum_jc_j\omega_j,
\qquad
Y=\sum_kd_k\eta_k,
\end{equation}
where $\omega_j,\eta_k$ are product density operators, such that
\begin{equation}
\sum_j|c_j|
\leq
\|X\|_{\rm SEP}+\frac{\varepsilon}{2},
\qquad
\sum_k|d_k|
\leq
\|Y\|_{\rm SEP}+\frac{\varepsilon}{2}.
\end{equation}
Concatenating the two decompositions gives
\begin{equation}
X+Y
=
\sum_jc_j\omega_j+\sum_kd_k\eta_k.
\end{equation}
Therefore,
\begin{equation}
\begin{split}
\|X+Y\|_{\rm SEP}
&\leq
\sum_j|c_j|+\sum_k|d_k|\\
&\leq
\|X\|_{\rm SEP}
+\|Y\|_{\rm SEP}
+\varepsilon.
\end{split}
\end{equation}
Since $\varepsilon>0$ is arbitrary,
\begin{equation}
\|X+Y\|_{\rm SEP}
\leq
\|X\|_{\rm SEP}
+\|Y\|_{\rm SEP}.
\end{equation}

Thus $\|\cdot\|_{\rm SEP}$ is finite, positive definite, absolutely
homogeneous, and subadditive. It is therefore a norm on the real vector space
$\operatorname{Herm}(\sA\sB)$.
\end{proof}

The following is the proof of Proposition~\ref{prop:TandR}.
\begin{proof}
1. Let
   \begin{equation}
   \Upsilon=(1+s)\sigma_+-s\sigma_-,
   \qquad
   \sigma_+,\sigma_-\in\Sep(\sA{:}\sB),
   \end{equation}
   be any feasible separable quasiprobabilistic mixture. Since $\|\sigma_\pm\|_1=1$, the triangle inequality therefore
   gives
   \begin{equation}
   \|\Upsilon\|_1
   \leq
   (1+s)\|\sigma_+\|_1+s\|\sigma_-\|_1
   =1+2s.
   \end{equation}
   Consequently,
   \begin{equation}
   \mathfrak T(\Upsilon)
   =
   \frac{\|\Upsilon\|_1-1}{2}
   \leq s.
   \end{equation}
   Moreover, partial transpose maps every separable state to another positive
   trace-one operator. Hence $\|\sigma_\pm^{T_{\sB}}\|_1=1$, and linearity of
   the partial transpose gives
   \begin{equation}
   \Upsilon^{T_{\sB}}
   =
   (1+s)\sigma_+^{T_{\sB}}
   -s\sigma_-^{T_{\sB}}.
   \end{equation}
   Another application of the triangle inequality yields
   \begin{equation}
   \|\Upsilon^{T_{\sB}}\|_1
   \leq
   (1+s)\|\sigma_+^{T_{\sB}}\|_1
   +s\|\sigma_-^{T_{\sB}}\|_1
   =1+2s.
   \end{equation}
   Therefore,
   \begin{equation}
   \fN_{\rm PT}(\Upsilon)
   =
   \frac{\|\Upsilon^{T_{\sB}}\|_1-1}{2}
   \leq s.
   \end{equation}
   Since these inequalities hold for every feasible quasiprobabilistic mixture, minimizing
   over $s$ proves Eq.~\eqref{eq:robustness-lower-bounds}

2. From 1, we have
   \begin{equation}
   \mathfrak R_{\rm SEP}(\Ups)\geq r=\mathfrak T(\Ups).
   \end{equation}
   Recall that $\mathfrak R_{\rm SEP}(\Ups)$ is the minimum $s\geq0$ for which
   there exist $\sigma_\pm\in\Sep(\sA{:}\sB)$ such that
   \begin{equation}\label{eq:sep-pseudomixture-proof}
   \Ups=(1+s)\sigma_+-s\sigma_-.
   \end{equation}
   Suppose first that
 $  \frac{\Ups_+}{1+r}\in\Sep(\sA{:}\sB) $
   and, when $r>0$,
   $\frac{\Ups_-}{r}\in\Sep(\sA{:}\sB)$.
   The Hahn--Jordan decomposition then gives the feasible separable
   quasiprobabilistic mixture
   \begin{equation}
   \Ups
   =
   (1+r)\frac{\Ups_+}{1+r}
   -r\frac{\Ups_-}{r}.
   \end{equation}
   Hence $\mathfrak R_{\rm SEP}(\Ups)\leq r$. Together with the preceding lower
   bound, this proves
   $\mathfrak R_{\rm SEP}(\Ups)
   =r=\mathfrak T(\Ups)$.

Conversely, suppose that
$\mathfrak R_{\rm SEP}(\Ups)=\fT(\Ups)=r$.
Since the optimization is attained in finite dimensions, there exist
$\sigma_\pm\in\Sep(\sA{:}\sB)$ such that
\begin{equation}
\Ups=(1+r)\sigma_+-r\sigma_-.
\end{equation}
Define
$A:=(1+r)\sigma_+$, $B:=r\sigma_-$.
Then $A,B\geq0$, $\Ups=A-B$, and
\begin{equation}
\|A-B\|_1
=
\|\Ups\|_1
=
1+2r
=
\Tr A+\Tr B.
\end{equation}
Equality in the trace-norm triangle inequality for positive operators holds
if and only if their supports are orthogonal. Thus $AB=0$.
From Eq.~\eqref{eq:DualHolder}, for Hermitian operator $X$, we have
\begin{equation}
\|X\|_1
=
\sup_{H\in\operatorname{Herm},\,\|H\|_\infty\leq1}
\Tr HX.
\end{equation}
Thus there is a Hermitian operator $-\I\leq H\leq\I$ such that
\begin{equation}
\|A-B\|_1=\Tr[H(A-B)].
\end{equation}
Equality with $\Tr A+\Tr B$ implies (recall the inner products of positive
operators are positive)
\begin{equation}
\Tr[(\I-H)A]=0,
\qquad
\Tr[(\I+H)B]=0.
\end{equation}
Notice $\Tr[(\I-H)A]=\|(\I-H)^{1/2}A^{1/2}\|_2^2=0$, $(\I-H)^{1/2}A^{1/2}=0$, $\operatorname{supp} A \subseteq \operatorname{ker} (\I-H)$; similarly, $\operatorname{supp} B \subseteq \operatorname{ker} (\I+H)$.
Thus $H=\I$ on $\operatorname{supp}A$ and $H=-\I$ on
$\operatorname{supp}B$, so these supports are orthogonal.
Therefore $A-B$ is itself the Hahn--Jordan decomposition of $\Ups$.
Uniqueness of that decomposition gives
$A=\Ups_+$, $B=\Ups_-$.
Consequently,
\begin{equation}
\sigma_+=\frac{\Ups_+}{1+r}\in\Sep(\sA{:}\sB)
\end{equation}
and, when $r>0$,
\begin{equation}
\sigma_-=\frac{\Ups_-}{r}\in\Sep(\sA{:}\sB).
\end{equation}
This proves the equivalence.

3.
Let
\begin{equation}
\cP
:=
\left\{
\rho_{\sA}\otimes\tau_{\sB}:
\rho_{\sA}\in\operatorname{Pos}_1(\sA),\
\tau_{\sB}\in\operatorname{Pos}_1(\sB)
\right\}.
\end{equation}
By the definition of the separable base norm, its closed unit ball is
\begin{equation}
\cC
= \left\{
X:\|X\|_{\rm SEP}\leq1
\right\} =
\operatorname{conv}
\left(
\cP\cup-\cP
\right).
\end{equation}
Indeed, every element of $\cC$ has a product-state decomposition whose
total absolute coefficient is at most one. Conversely, every decomposition
with total absolute coefficient at most one gives an element of $\cC$,
because $\cC$ is convex and contains the zero operator.

The dual unit ball is therefore the set
\begin{equation}
\cC^{\vee}
=
\left\{
W\in \operatorname{Herm}(\sA\sB):
\left|
\Tr\!\left[W(\rho_{\sA}\otimes\tau_{\sB})\right]
\right|
\leq1
\quad\text{for all }\rho_{\sA},\tau_{\sB}
\right\}.
\end{equation}
Finite-dimensional norm duality now gives
\begin{equation}
\|\Ups\|_{\rm SEP}
=
\max_{W\in\cC^{\vee}}\Tr(W\Ups),
\end{equation}
which is  Eq.~\eqref{eq:sep-base-dual}.
\end{proof}

\subsection{Proof of Eq.~\eqref{eq:STstateNormal}}
\label{app:normal}
\begin{proof}
We need to prove
\begin{equation}
    \|\Ups\|_1=\sum_j|\lambda_j|
    \quad\Longleftrightarrow\quad
    \Ups\ \text{is normal}.
\end{equation}
When $\Ups$ is normal, it is clear that the equality holds.
For the reverse direction,  take a Schur decomposition
\begin{equation}
    \Ups=UTU^\dagger,
\end{equation}
where \(U\) is unitary and \(T\) is upper triangular, with
\(T_{jj}=\lambda_j\). By unitary invariance of the trace norm $\|\Ups\|_1=\|T\|_1$.
Define the diagonal unitary
\begin{equation}
    D:=\operatorname{diag}(e^{i\theta_1},\ldots,e^{i\theta_D}),
    \qquad
    e^{i\theta_j}:=
    \begin{cases}
        \lambda_j/|\lambda_j|, & \lambda_j\neq0,\\
        1, & \lambda_j=0.
    \end{cases}
\end{equation}
Then
\begin{equation}
    \operatorname{Re}\Tr(D^\dagger T)
    =\sum_j\operatorname{Re}
      \bigl(e^{-i\theta_j}\lambda_j\bigr)
    =\sum_j|\lambda_j|
    =Z_\lambda(\Ups).
\end{equation}
Using trace-norm duality in Eq.~\eqref{eq:DualHolder},
and observing that \(\|D\|_{\infty}=1\), we obtain
\begin{equation}
    \|T\|_1\geq
    \operatorname{Re}\Tr(D^\dagger T)
    =Z_\lambda(\Ups).
\end{equation}
Suppose now that equality holds: $\|T\|_1=Z_\lambda(\Ups)$. Set $ A:=D^\dagger T.$
Since left multiplication by a unitary preserves singular values,
\begin{equation}
    \|A\|_1=\|T\|_1
    =\operatorname{Re}\Tr A.
\end{equation}
We claim that this equality implies \(A\geq0\). To see this, write the
polar decomposition
\begin{equation}
    A=V|A|,
    \qquad
    |A|=\sum_k s_k|k\rangle\langle k|,
\end{equation}
where \(V\) is a partial isometry and \(s_k\geq0\). Then
\begin{align}
    \operatorname{Re}\Tr A
    &=\operatorname{Re}\Tr(V|A|)
    =\sum_k s_k\,
      \operatorname{Re}\langle k|V|k\rangle
    \leq\sum_k s_k
    =\|A\|_1.
\end{align}
Equality requires
\(\operatorname{Re}\langle k|V|k\rangle=1\) for every \(k\) with
\(s_k>0\). Since \(\|V\|_\infty\leq1\), this is possible only if
\begin{equation}
    V|k\rangle=|k\rangle
    \qquad\text{whenever }s_k>0.
\end{equation}
Thus \(V\) acts as the identity on the support of \(|A|\), and hence
\begin{equation}
    A=V|A|=|A|\geq0.
\end{equation}
On the other hand, \(A=D^\dagger T\) is upper triangular. Since a positive
operator is Hermitian, \(A\) is both Hermitian and upper triangular.
Every upper-triangular Hermitian matrix is diagonal: its entries below the
diagonal vanish by triangularity, and Hermiticity then forces all entries
above the diagonal to vanish. Therefore \(A\) is diagonal. Since \(D\) is
also diagonal,
\begin{equation}
    T=DA
\end{equation}
is diagonal. It follows that
\begin{equation}
    \Ups=UTU^\dagger
\end{equation}
is unitarily diagonalizable and therefore normal.
\end{proof}

\subsection{Proof of Proposition \ref{prop:Ivec for doubled KD state}}\label{app:Ivec}
\begin{proof}
For simplicity, we introduce
\begin{align}
    F_\mu
    &:= \sum_{\nu=0}^{3}
    \hat{\sigma}_\mu\hat{\sigma}_\nu
    \otimes\hat{\sigma}_\nu, \\
    G_\mu
    &:= (\mathds{1}\otimes\hat{\sigma}_\mu)
    \overleftrightarrow{\Upsilon}_{t_0}, \\
    K_{\alpha\beta}
    &:= \sum_{\nu=0}^{3}
    \hat{\sigma}_\beta\hat{\sigma}_\nu
    \otimes\hat{\sigma}_\alpha\hat{\sigma}_\nu.
\end{align}
The doubled Kirkwood-Dirac operator can then be expressed as
\begin{align}
    \Un&= \sum_{\mu_1,\ldots,\mu_n}\left(\prod_{\ell=1}^{n}h_{\mu_\ell}\right)F_{\mu_n}\otimes K_{\mu_n\mu_{n-1}}\otimes\cdots\otimes K_{\mu_2\mu_1}\otimes G_{\mu_1}.
\end{align}

First consider an initial interval $A=\{t_0,\ldots,t_k\}$, with
$0\leq k<n$. Grouping the factors on either side of the bond $k+1$ gives
\begin{align}
    \Un&= \sum_{\mu=0}^{3}h_\mu F_\mu^{(k)}\otimes G_\mu^{(k)},
\end{align}
where
\begin{align}
    F_\mu^{(k)}&= \sum_{\mu_{k+2},\ldots,\mu_n}\left(\prod_{\ell=k+2}^{n}h_{\mu_\ell}\right)F_{\mu_n}\otimes K_{\mu_n\mu_{n-1}}\otimes\cdots\otimes K_{\mu_{k+2}\mu}, \\
    G_\mu^{(k)}&= \sum_{\mu_1,\ldots,\mu_k}\left(\prod_{\ell=1}^{k}h_{\mu_\ell}\right)K_{\mu\mu_k}\otimes\cdots\otimes
    K_{\mu_2\mu_1}\otimes G_{\mu_1}.
\end{align}
At the endpoints, these expressions are understood as
\begin{align}
    F_\mu^{(n-1)} &= F_\mu,
    \qquad G_\mu^{(0)} = G_\mu.
\end{align}

To obtain the operator-Schmidt coefficients, consider the Gram matrices
\begin{align}
    \Gamma_{\mu\nu}^{(k)}&:= \operatorname{Tr}\left[F_\mu^{(k)\dagger}F_\nu^{(k)}\right],
    \qquad\Sigma_{\mu\nu}^{(k)}:= \operatorname{Tr}\!\left[G_\mu^{(k)\dagger}G_\nu^{(k)}\right].
\end{align}
The local operators satisfy
\begin{align}
    \operatorname{Tr}(F_\mu^\dagger F_\nu)
    &= 2\delta_{\mu\nu}, \\
    \operatorname{Tr}(K_{\alpha\beta}^\dagger K_{\gamma\delta})
    &= \delta_{\alpha\gamma}\delta_{\beta\delta}, \\
    \operatorname{Tr}(G_\mu^\dagger G_\nu)
    &= P_0\delta_{\mu\nu}.
\end{align}
 Using these identities to contract the internal indices yields
\begin{align}
    \Gamma_{\mu\nu}^{(k)}
    &= 2Q^{\,n-k-1}\delta_{\mu\nu},
    \qquad
    \Sigma_{\mu\nu}^{(k)}
    = P_0Q^k\delta_{\mu\nu}.
\end{align}
Thus both block families are orthogonal and can be normalized directly:
\begin{align}
    \widehat F_\mu^{(k)}
    &= \frac{F_\mu^{(k)}}{\sqrt{2Q^{\,n-k-1}}},
    \qquad
    \widehat G_\mu^{(k)}
    = \frac{G_\mu^{(k)}}{\sqrt{P_0Q^k}}.
\end{align}
After absorbing the sign of each nonzero $h_\mu$ into one normalized block
operator, the operator-Schmidt coefficients are
\begin{align}
    \lambda_\mu&= \sqrt{2P_0Q^{\,n-1}}\,|h_\mu|.
\end{align}
It follows that
\begin{align}
    \|\Un\|_{2}^2
    &= \sum_{\mu=0}^{3}\lambda_\mu^2
    = 2P_0Q^n, \\
    p_\mu
    &= \frac{\lambda_\mu^2}{\sum_\nu \lambda_\nu^2}
    = \frac{h_\mu^2}{Q}
    = q_\mu.
\end{align}
This proves the result for an initial interval, for which $d_A=1$.

Now, we consider an arbitrary temporal subsystem $A$, including disconnected regions. For $X=A$ or $A^c$, let
\begin{align}
    I_X
    &:= \{\ell\in\{1,\ldots,n\}:t_{\ell-1},t_\ell\in X\},
    & i_X &:= |I_X|
\end{align}
denote its internal bonds and their number. Every bond is either internal to $A/A^c$, or crosses the partition, so
\begin{align}
    i_A+i_{A^c}+d_A &= n.
\end{align}
Fix the crossing-bond labels
\begin{align}
    \boldsymbol\mu_{\partial A}
    &= (\mu_\ell)_{\ell\in\partial A}.
\end{align}
For each $X$, regrouping the temporal tensor factors according to $A|A^c$, we obtain
\begin{align}
    \Un
    &= \sum_{\boldsymbol\mu_{\partial A}}
    \left(\prod_{\ell\in\partial A}h_{\mu_\ell}\right)
    F_{\boldsymbol\mu_{\partial A}}^{(A)}
    \otimes G_{\boldsymbol\mu_{\partial A}}^{(A^c)}.
\end{align}
The same local trace identities give
\begin{align}
    \operatorname{Tr}\!\left[
    F_{\boldsymbol\mu_{\partial A}}^{(A)\dagger}
    F_{\boldsymbol\nu_{\partial A}}^{(A)}
    \right]
    &= c_AQ^{i_A}
    \delta_{\boldsymbol\mu_{\partial A},\boldsymbol\nu_{\partial A}},\quad  \operatorname{Tr}\!\left[
    G_{\boldsymbol\mu_{\partial A}}^{(A^c)\dagger}
    G_{\boldsymbol\nu_{\partial A}}^{(A^c)}
    \right]
    = c_{A^c}Q^{i_{A^c}}
    \delta_{\boldsymbol\mu_{\partial A},\boldsymbol\nu_{\partial A}},
\end{align}
where
\begin{align}
    c_X
    &:= \begin{cases}
        1, & t_0\notin X,\ t_n\notin X, \\
        P_0, & t_0\in X,\ t_n\notin X, \\
        2, & t_0\notin X,\ t_n\in X, \\
        2P_0, & t_0\in X,\ t_n\in X.
    \end{cases}
\end{align}
In particular, $c_Ac_{A^c}=2P_0$. Here the Kronecker delta is the product of the deltas for all crossing-bond labels. Normalizing the two block families and absorbing the coefficient signs into one of them therefore gives the operator-Schmidt coefficients
\begin{align}
    \lambda_{\boldsymbol\mu_{\partial A}}
    &= \sqrt{c_Ac_{A^c}Q^{i_A+i_{A^c}}}
    \prod_{\ell\in\partial A}|h_{\mu_\ell}| \\
    &= \sqrt{2P_0Q^{\,n-d_A}}
    \prod_{\ell\in\partial A}|h_{\mu_\ell}|.
\end{align}
Dividing their squares by the squared Hilbert--Schmidt norm yields
\begin{align}
    p_{\boldsymbol\mu_{\partial A}}
    &= \frac{
    2P_0Q^{\,n-d_A}\prod_{\ell\in\partial A}h_{\mu_\ell}^2
    }{2P_0Q^n} \\
    &= \prod_{\ell\in\partial A}\frac{h_{\mu_\ell}^2}{Q}
    = \prod_{\ell\in\partial A}q_{\mu_\ell}.
\end{align}
The dependence on the initial-state purity $P_0$ therefore cancels. Since this probability distribution factorizes over the crossing bonds,
\begin{align}
    S_{\mathrm{vec}}(A)
    &= -\sum_{\boldsymbol\mu_{\partial A}}
    p_{\boldsymbol\mu_{\partial A}}
    \log p_{\boldsymbol\mu_{\partial A}} \\
    &= -d_A\sum_{\mu=0}^{3}q_\mu\log q_\mu.
\end{align}
This completes the proof.
\end{proof}

\begin{figure}[h]
    \centering
    \includegraphics[width=1\linewidth]{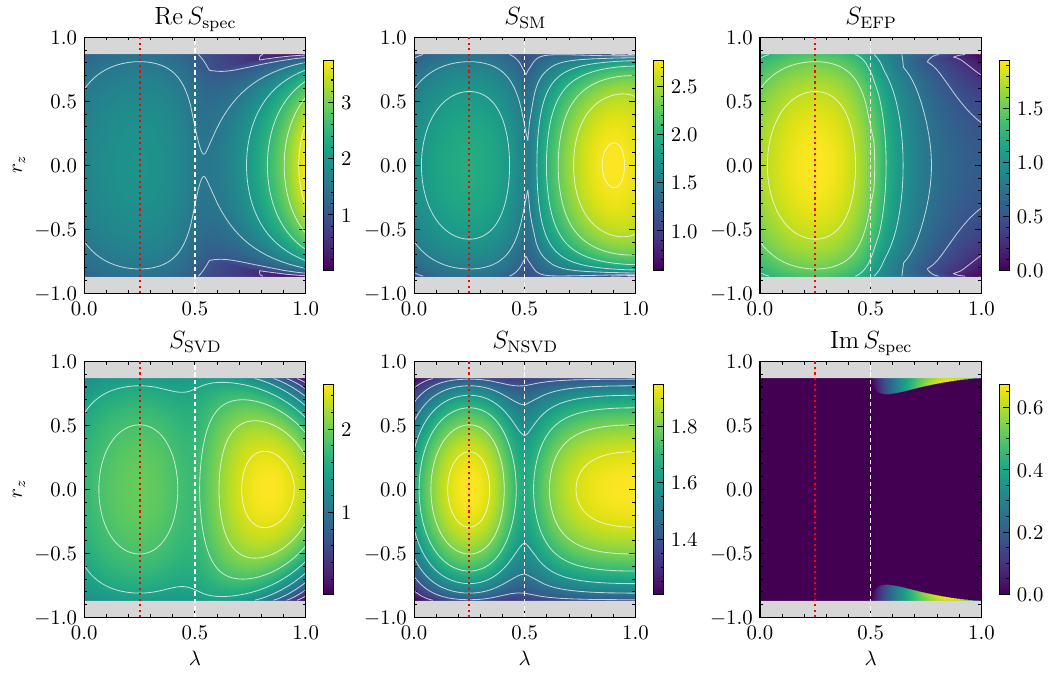}
    \caption{Entropy maps for the three-time right Kirkwood-Dirac temporal state generated by depolarizing channel. The imaginary part of the complex spectral entropy exhibits a different pattern from that of the doubled construction, reflecting their different spectral structures. For $\lambda\leq\frac12$, the right temporal state has a nonnegative spectrum, and hence $\operatorname{Im}S_{\mathrm{spec}}=0$. For $\lambda>\frac12$, nonreal eigenvalues occur in complex-conjugate pairs whose imaginary contributions to the spectral entropy cancel. Consequently, $\operatorname{Im}S_{\mathrm{spec}}$ can vanish even when nonreal eigenvalues are present. With the principal logarithm, a nonzero imaginary part in this case arises from negative real eigenvalues.}
     \label{fig:right-depolarizing-entropy}
\end{figure}

\begin{figure}[h]
    \centering
    \includegraphics[width=1\linewidth]{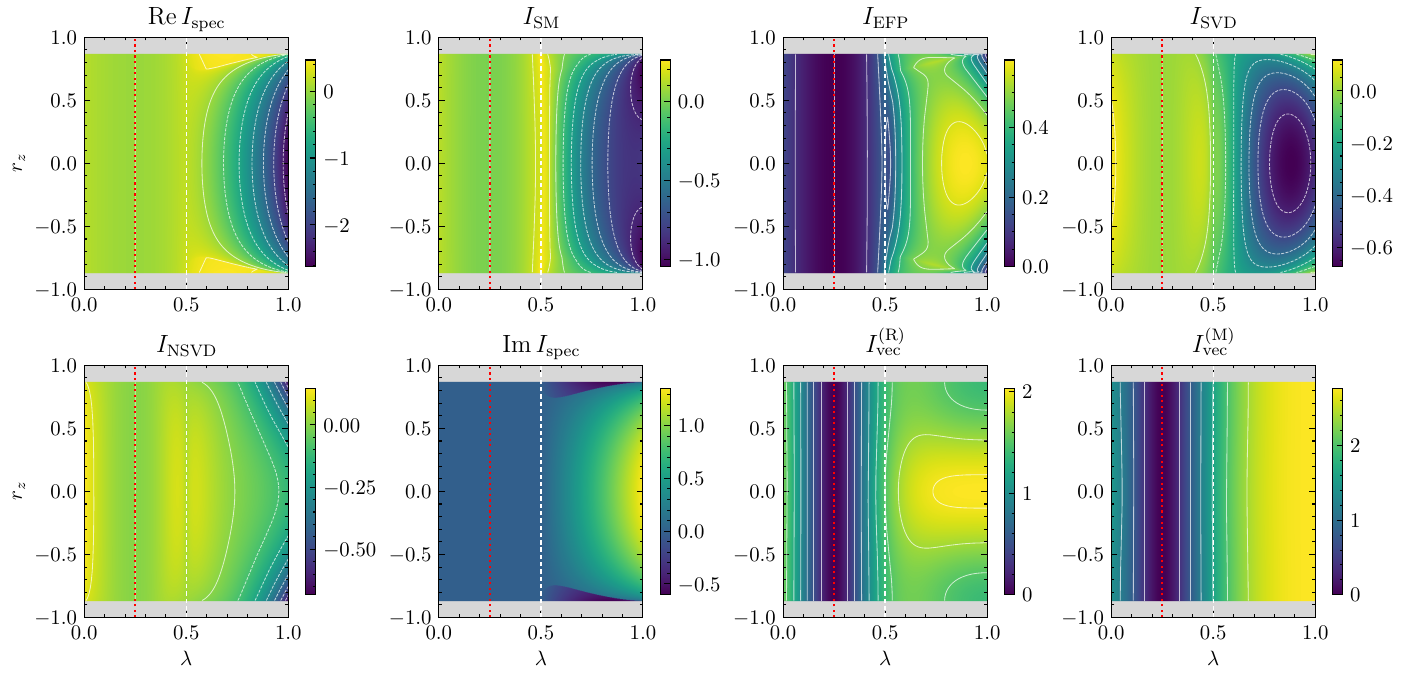}
    \caption{Mutual-information maps for three-time right- and mixed-branch Kirkwood-Dirac temporal states generated by depolarizing channel. The panel labeled $I_{\mathrm{vec}}^{(\mathrm M)}$ corresponds to the mixed-branch construction; all remaining panels correspond to the right construction. At $\lambda=\frac14$, both temporal states reduce to $\Upsilon_{t_2t_1t_0}^{(\mathrm{R})}=\Upsilon_{t_2t_1t_0}^{(\mathrm M)}
=\frac{\mathbb I}{2}\otimes\frac{\mathbb I}{2}\otimes\varrho_{t_0}$,
and all plotted mutual-information quantities vanish. In particular, the unnormalized singular-value entropy is additive for this product because its factors are positive and have unit trace norm. The approximately vertical contours of $I_{\mathrm{vec}}$ for $\lambda<\frac12$ indicate weak dependence on the initial Bloch-vector component $r_z$; this dependence is generally nonzero.}
    \label{fig:right-depolarizing-mi}
\end{figure}

\begin{figure}[h]
    \centering
    \includegraphics[width=1\linewidth]{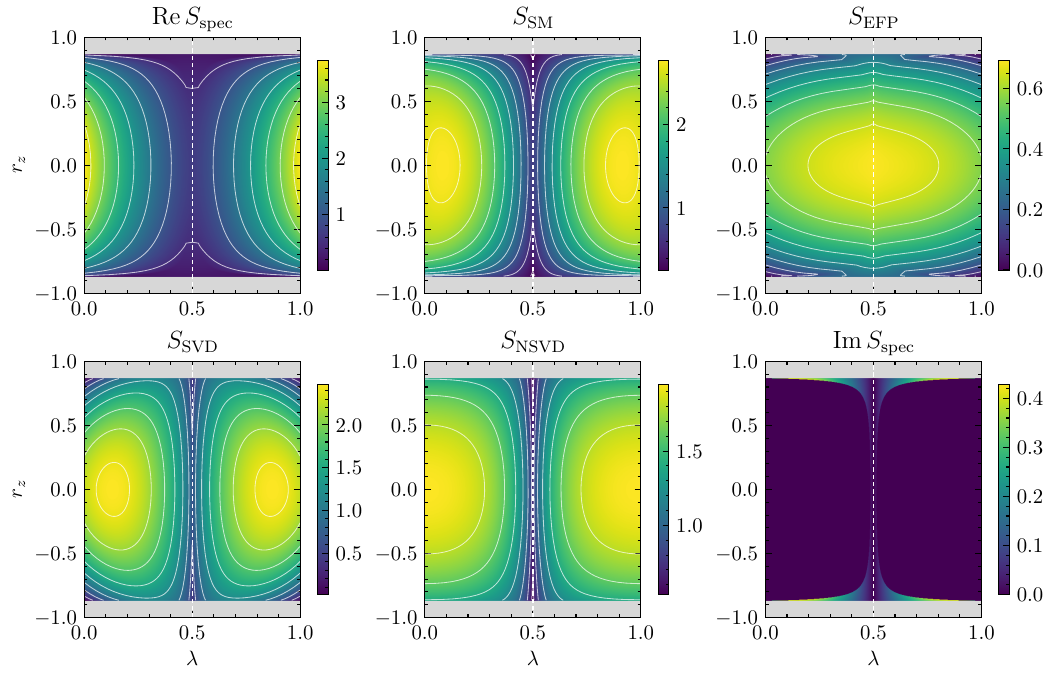}
    \caption{Entropy maps for the three-time right Kirkwood-Dirac temporal state generated by dephasing channel. The vanishing region of the complex spectral entropy coincides with the spectrum containing only positive real and strictly complex eigenvalues. }
    \label{fig:right-dephasing-entropy}
  
\end{figure}

\begin{figure}[h]
    \centering
    \includegraphics[width=1\linewidth]{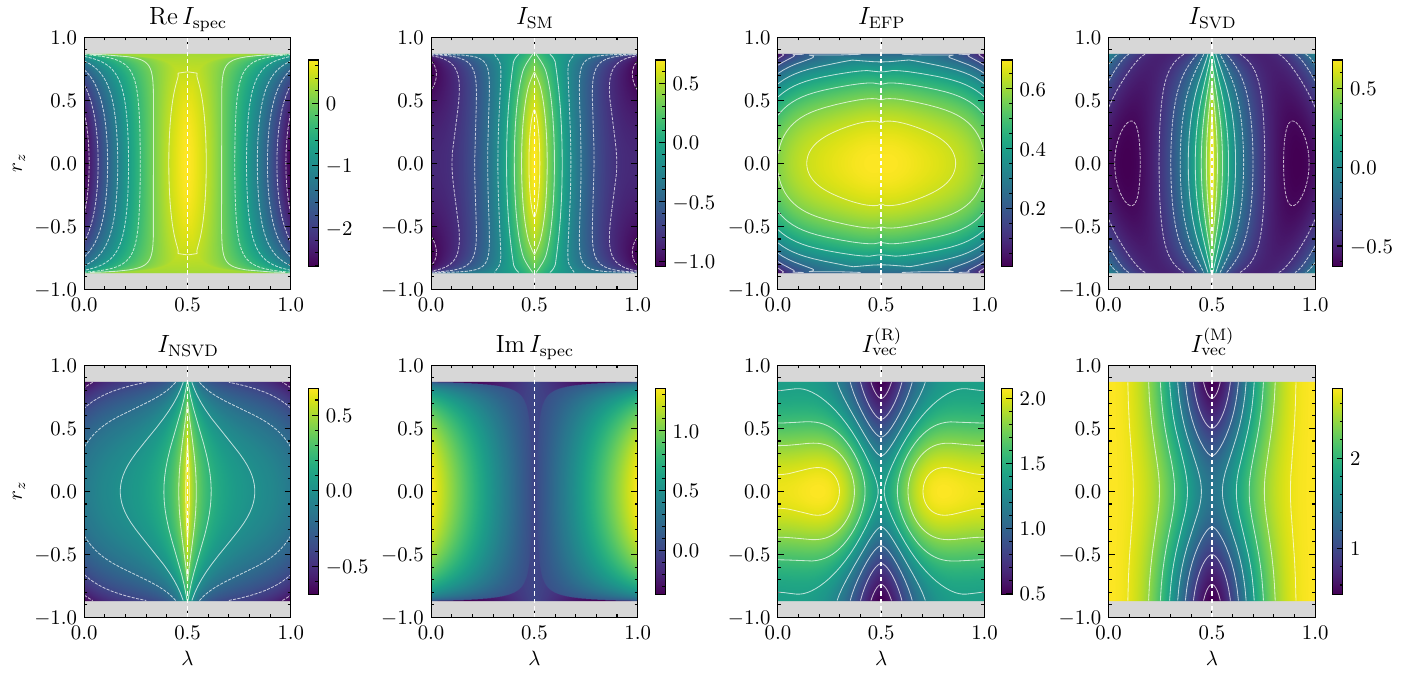}
    \caption{Mutual-information maps for three-time right- and mixed-branch Kirkwood-Dirac temporal states generated by dephasing channels. The mixed-branch contours are more nearly vertical near the endpoints $\lambda=0,1$, indicating weaker dependence on $r_z$. At both endpoints, $I_{\mathrm{vec}}^{(\mathrm M)}$ is almost independent of the initial state, while $I_{\mathrm{vec}}^{(\mathrm R)}$ retains its Bloch-vector dependence. This difference reflects the distinct operator Schmidt spectra associated with the operator ordering in the two constructions.}
  \label{fig:right-dephasing-mi}
\end{figure}

\section{Doubled Kirkwood--Dirac tensors and the superdensity operator}
\label{app:superDens}
The doubled Kirkwood--Dirac correlation function 
$T^{\boldsymbol\mu;\boldsymbol\nu}$ defines a positive semidefinite  doubled correlation tensor
\cite{cotler2018superdensity,jia2024spatiotemporal}.
Using the orthonormal kets
$|\boldsymbol\mu\rangle\!\rangle
=\bigotimes_k|\widehat\sigma_{\mu_k}\rangle\!\rangle$, the corresponding
superdensity operator is defined by \cite{cotler2018superdensity}
\begin{equation}\label{eq:superdensity-xi}
 \Xi:=\frac{1}{Z}\sum_{\boldsymbol\mu,\boldsymbol\nu}
 T^{\boldsymbol\mu;\boldsymbol\nu}
 |\boldsymbol\mu\rangle\!\rangle
 \langle\!\langle\boldsymbol\nu|,
 \qquad
 Z=\sum_{\boldsymbol\mu}
 T^{\boldsymbol\mu;\boldsymbol\mu}.
\end{equation}
Since $\Xi\geq 0$ and $\Tr\Xi=1$, the standard notions
of entropy and mutual information apply. For a spatiotemporal region
$\mathsf A$, we define
\begin{equation}
 S_\Xi(\mathsf A)
 :=S\!\left(\Tr_{\mathsf A^c}\Xi\right),
 \qquad
 I_\Xi(\mathsf A:\mathsf B)
 :=S_\Xi(\mathsf A)+S_\Xi(\mathsf B)-S_\Xi(\mathsf A\mathsf B).
\end{equation}
As $\Xi$ is a genuine density operator, these quantities satisfy the usual
entropy inequalities.

The vectorized doubled Kirkwood--Dirac spatiotemporal state discussed in Section~\ref{sec:STentropy} is indeed equivalent to the vectorization of the superdensity operator. However, the vectorized spatiotemporal state and the superdensity operator provide two conceptually distinct constructions for encoding spatiotemporal correlations:
\begin{enumerate}
\item
The vectorized spatiotemporal state is pure and can be constructed from any nonzero left, right, mixed-branch, doubled, or folded temporal operator.
In contrast, $\Xi$ is generally mixed and its construction requires a
positive doubled correlation tensor $T^{\boldsymbol\mu;\boldsymbol\nu}$. It can not be applied to single local space state, like left, right, mixed-brach spatiotemporal states.

\item
Vectorization preserves the linear relations among different temporal
branches in a transparent manner. For example, the relation
$\Upsl=\Upsr^\dagger$ is mapped to an antiunitary conjugation between the
corresponding vectorized kets. By contrast, the construction of $\Xi$
contracts the two branches into a doubled correlation tensor, and such relations are no longer manifest as simple relations between states.

\item
 In general, there is no
equality or universal ordering between $S_\Xi$, $I_\Xi$, $S_{\rm vec}$,
and $I_{\rm vec}$.
\end{enumerate}

We therefore take vectorization as the basic framework for defining
spatiotemporal entanglement. It applies uniformly to all temporal states and retains the relations among the left, right, and mixed temporal branches.
The superdensity operator provides a complementary positive construction,
particularly when the doubled correlation tensor
$T^{\boldsymbol\mu;\boldsymbol\nu}$ itself is the object of interest.

\section{Entropy and mutual information for right and mixed-branch Kirkwood-Dirac temporal state}\label{app:KD-comparison}
This appendix presents the entropy and mutual-information maps for the right and mixed-branch Kirkwood-Dirac temporal states, complementing the doubled-construction results discussed in Section \ref{sec:example}. We use the same channel parameterizations, initial-state parameters, and temporal bipartition as in the main analysis. The results are organized by channel, with the entropy maps followed by the corresponding mutual-information maps.  The mixed-branch temporal state considered here is
\begin{align}
    \Ups_{t_2t_1t_0}^{(M)}=J[\cE_2]\star\varrho_{t_0}\star J[\cE_1].
\end{align}
The right and mixed-branch Kirkwood–Dirac temporal states have identical spectra, as we showed in the proof of Proposition  \ref{prop:pair-complex-eigenvalue}. Consequently, all entropy prescriptions depending only on the eigenvalues coincide for the two constructions. However, the corresponding singular-value entropies and vectorized mutual information can distinguish the two constructions. 
For the parameter ranges considered, the entropy and mutual-information maps of the mixed-branch Kirkwood-Dirac temporal state exhibit qualitatively similar patterns to those of the right construction, except for the vectorized mutual information. Therefore, in the figures below, we present the right-state maps and additionally include the mixed-branch \(I_{\mathrm{vec}}\) map to highlight the difference in vectorized temporal correlations.  The entropies and mutual information for the depolarizing channel are shown in Figures~\ref{fig:right-depolarizing-entropy} and~\ref{fig:right-depolarizing-mi}, respectively. The corresponding results for the dephasing channel are presented in Figures~\ref{fig:right-dephasing-entropy} and~\ref{fig:right-dephasing-mi}, respectively.

\bibliographystyle{apsrev4-1-title}
\bibliography{Jiabib}
\end{document}

%% file: MathPhysDef.tex
\usepackage{MathPhystitle}
\pdfoutput = 1
\usepackage[T1]{fontenc}

\usepackage{amsmath}
\usepackage{amssymb}
\usepackage{amsthm}
\usepackage{bm}
\usepackage{braket}
\usepackage{graphicx}
\usepackage{adjustbox}

\usepackage{mathdots}

\usepackage{multirow}
\usepackage{tabu}

\usepackage{booktabs} % for better lines
\usepackage{xcolor}   % for defining custom colors
\usepackage{hhline}   % for custom horizontal lines

\usepackage{orcidlink} %%%%%%oricd link%%%%%%%

\usepackage{tikz}
\usetikzlibrary{positioning,intersections}
\usetikzlibrary{calc}
\usetikzlibrary{shapes.geometric}
\usepackage{braids}
\usepackage{tikz-cd}
\usetikzlibrary{shapes.geometric,shapes.misc}
\usetikzlibrary{arrows,matrix,calc,scopes,decorations.markings,snakes}
\usetikzlibrary{arrows.meta}
\usetikzlibrary{intersections, patterns,fit} 
\usetikzlibrary{decorations.pathreplacing,angles,quotes}
\usetikzlibrary{tqft} %for TQFT graph
\definecolor{darkblue}{RGB}{40, 85, 120} % Darker blue for better contrast

\newtheorem{theorem}{Theorem}[section]

\newtheorem{proposition}[theorem]{Proposition}
\newtheorem{corollary}[theorem]{Corollary}

\theoremstyle{definition}
\newtheorem{definition}[theorem]{Definition}

\newtheorem{example}[theorem]{Example}

\theoremstyle{remark}
\newtheorem{remark}[theorem]{Remark}
\usepackage[bbgreekl]{mathbbol} %change mathbb style
\usepackage{dsfont}
\usepackage{microtype}       % Improves typography
\DeclareMathAlphabet{\mathcal}{OMS}{cmsy}{m}{n}
\usepackage{euscript}

\newcommand{\Rbb}{\mathbb{R}}

\newcommand\Tr{\operatorname{Tr}}

\newcommand{\id}{\operatorname{id}}

\newcommand{\cC}{\text{\usefont{OMS}{cmsy}{m}{n}C}}

\newcommand{\cE}{\text{\usefont{OMS}{cmsy}{m}{n}E}}
\newcommand{\cF}{\text{\usefont{OMS}{cmsy}{m}{n}F}}

\newcommand{\cH}{\text{\usefont{OMS}{cmsy}{m}{n}H}}
\newcommand{\cK}{\text{\usefont{OMS}{cmsy}{m}{n}K}}

\newcommand{\cM}{\text{\usefont{OMS}{cmsy}{m}{n}M}}

\newcommand{\cP}{\text{\usefont{OMS}{cmsy}{m}{n}P}}

\newcommand\sA           {\mathsf{A}}
\newcommand\sB           {\mathsf{B}}
\newcommand\sC           {\mathsf{C}}

\newcommand\sF          {\mathsf{F}}

\newcommand\sP         {\mathsf{P}}

\newcommand\sR         {\mathsf{R}}

\newcommand\fA           {\mathfrak{A}}

\newcommand\fN           {\mathfrak{N}}

\newcommand\fT           {\mathfrak{T}}

\newcommand\ttB         {\mathtt{B}}

\usepackage{mathtools,amssymb,scalerel}

\newcommand\biopencrossl{%
	\mathrel{\scalerel*{>\kern-.4\LMpt\joinrel\blacktriangleleft}{x}}}
\newcommand\biopencrossr{%
	\mathrel{\scalerel*{\blacktriangleright\joinrel\kern-.4\LMpt<}{x}}}

%% file: Jiabib.bib
@misc{das2026timelikeentanglementsignaturesergodicity,
      title={Timelike Entanglement Signatures of Ergodicity and Spectral Chaos}, 
      author={Rathindra Nath Das and Arnab Kundu and Nemai Chandra Sarkar},
      year={2026},
      eprint={2601.19981},
      archivePrefix={arXiv},
      primaryClass={hep-th},
      url={https://arxiv.org/abs/2601.19981}, 
}

@misc{das2025temporalcorrelationschaosspacetime,
      title={Temporal correlations and chaos from spacetime kernel}, 
      author={Rathindra Nath Das and Arnab Kundu and Matheus H. Martins Costa and Nemai Chandra Sarkar},
      year={2025},
      eprint={2512.06078},
      archivePrefix={arXiv},
      primaryClass={hep-th},
      url={https://arxiv.org/abs/2512.06078}, 
}

@article{Guo2026Spacetime,
  author = {Guo, Wu-zhong},
  title = {Spacetime density matrix: formalism and properties},
  journal = {Journal of High Energy Physics},
  year = {2026},
  volume = {2026},
  number = {1},
  pages = {128},
  doi = {10.1007/JHEP01(2026)128},
  url = {https://doi.org/10.1007/JHEP01(2026)128},
  isbn = {1029-8479}
}

@misc{diaz2026unifyingspacetimeapproachesquantum,
      title={Unifying spacetime approaches to quantum mechanics}, 
      author={N. L. Diaz and M. Cerezo and Paolo Braccia},
      year={2026},
      eprint={2606.12539},
      archivePrefix={arXiv},
      primaryClass={quant-ph},
      url={https://arxiv.org/abs/2606.12539}, 
}

@article{Audenaert2007entropy,
doi = {10.1088/1751-8113/40/28/S18},
url = {https://doi.org/10.1088/1751-8113/40/28/S18},
year = {2007},
month = {jun},
publisher = {},
volume = {40},
number = {28},
pages = {8127},
author = {Audenaert, Koenraad M R},
title = {A sharp continuity estimate for the von Neumann entropy},
journal = {Journal of Physics A: Mathematical and Theoretical}
}

@Article{Fannes1973entroy,
author={Fannes, M.},
title={A continuity property of the entropy density for spin lattice systems},
journal={Communications in Mathematical Physics},
year={1973},
month={Dec},
day={01},
volume={31},
number={4},
pages={291-294},
issn={1432-0916},
doi={10.1007/BF01646490},
url={https://doi.org/10.1007/BF01646490}
}

@article{Rovelli1991time,
  title = {Time in quantum gravity: An hypothesis},
  author = {Rovelli, Carlo},
  journal = {Phys. Rev. D},
  volume = {43},
  issue = {2},
  pages = {442--456},
  numpages = {0},
  year = {1991},
  month = {Jan},
  publisher = {American Physical Society},
  doi = {10.1103/PhysRevD.43.442},
  url = {https://link.aps.org/doi/10.1103/PhysRevD.43.442}
}

@article{Anderson2012time,
author = {Anderson, E.},
title = {Problem of time in quantum gravity},
journal = {Annalen der Physik},
volume = {524},
number = {12},
pages = {757-786},
doi = {https://doi.org/10.1002/andp.201200147},
url = {https://onlinelibrary.wiley.com/doi/abs/10.1002/andp.201200147},
eprint = {https://onlinelibrary.wiley.com/doi/pdf/10.1002/andp.201200147},
year = {2012}
}

@book{Muga2008Time,
  editor    = {Muga, J. G. and Sala Mayato, R. and Egusquiza, {\'I}. L.},
  title     = {Time in Quantum Mechanics},
  series    = {Lecture Notes in Physics},
  publisher = {Springer Berlin Heidelberg},
  address   = {Berlin, Heidelberg},
  year      = {2008},
  edition   = {2},
  pages     = {455},
  doi       = {10.1007/978-3-540-73473-4},
  isbn      = {978-3-540-73473-4}
}

@article{Page1983evolution,
  title = {Evolution without evolution: Dynamics described by stationary observables},
  author = {Page, Don N. and Wootters, William K.},
  journal = {Phys. Rev. D},
  volume = {27},
  issue = {12},
  pages = {2885--2892},
  numpages = {0},
  year = {1983},
  month = {Jun},
  publisher = {American Physical Society},
  doi = {10.1103/PhysRevD.27.2885},
  url = {https://link.aps.org/doi/10.1103/PhysRevD.27.2885}
}

@article{lie2025probingquantumstatesspacetime,
  title = {Probing Quantum States over Spacetime through Interferometry},
  author = {Lie, Seok Hyung and Kwon, Hyukjoon},
  journal = {Phys. Rev. Lett.},
  volume = {136},
  issue = {25},
  pages = {250201},
  numpages = {10},
  year = {2026},
  month = {Jun},
  publisher = {American Physical Society},
  doi = {10.1103/t6jg-z17r},
  url = {https://link.aps.org/doi/10.1103/t6jg-z17r},
   eprint={2507.19258},
  archivePrefix={arXiv},
  primaryClass={quant-ph} 
}

@article{Doi2023Timelike,
  author = {Doi, Kazuki and Harper, Jonathan and Mollabashi, Ali and Takayanagi, Tadashi and Taki, Yusuke},
  title = {Timelike entanglement entropy},
  journal = {Journal of High Energy Physics},
  year = {2023},
  volume = {2023},
  number = {5},
  pages = {52},
  doi = {10.1007/JHEP05(2023)052},
  url = {https://doi.org/10.1007/JHEP05(2023)052},
  isbn = {1029-8479}
}

@article{Doi2023Pseudoentropy,
  title = {Pseudoentropy in $\mathrm{dS}/\mathrm{CFT}$ and Timelike Entanglement Entropy},
  author = {Doi, Kazuki and Harper, Jonathan and Mollabashi, Ali and Takayanagi, Tadashi and Taki, Yusuke},
  journal = {Phys. Rev. Lett.},
  volume = {130},
  issue = {3},
  pages = {031601},
  numpages = {6},
  year = {2023},
  month = {Jan},
  publisher = {American Physical Society},
  doi = {10.1103/PhysRevLett.130.031601},
  url = {https://link.aps.org/doi/10.1103/PhysRevLett.130.031601}
}

@misc{fullwood2026entropypseudodensitymatrix,
      title={On the entropy of a pseudo-density matrix}, 
      author={James Fullwood and Boyu Yang},
      year={2026},
      eprint={2608.28946},
      archivePrefix={arXiv},
      primaryClass={quant-ph},
      url={https://arxiv.org/abs/2608.28946}, 
}

@Article{Caputa2024SVD,
author={Caputa, Pawe{\l}
and Purkayastha, Souradeep
and Saha, Abhigyan
and Su{\l}kowski, Piotr},
title={Musings on SVD and pseudo entanglement entropies},
journal={Journal of High Energy Physics},
year={2024},
month={Nov},
day={19},
volume={2024},
number={11},
pages={103},
issn={1029-8479},
doi={10.1007/JHEP11(2024)103},
url={https://doi.org/10.1007/JHEP11(2024)103}
}

@article{Watanabe1955twostateQM,
  title = {Symmetry of Physical Laws. Part III. Prediction and Retrodiction},
  author = {Watanabe, Satosi},
  journal = {Rev. Mod. Phys.},
  volume = {27},
  issue = {2},
  pages = {179--186},
  numpages = {0},
  year = {1955},
  month = {Apr},
  publisher = {American Physical Society},
  doi = {10.1103/RevModPhys.27.179},
  url = {https://link.aps.org/doi/10.1103/RevModPhys.27.179}
}

@misc{aharonov2007twostatevectorformalismqauntum,
      title={The Two-State Vector Formalism of Qauntum Mechanics: an Updated Review}, 
      author={Yakir Aharonov and Lev Vaidman},
      year={2007},
      eprint={quant-ph/0105101},
      archivePrefix={arXiv},
      primaryClass={quant-ph},
      url={https://arxiv.org/abs/quant-ph/0105101}, 
}

@article{Braunstein2007QIblackhole,
  title = {Quantum Information Cannot Be Completely Hidden in Correlations: Implications for the Black-Hole Information Paradox},
  author = {Braunstein, Samuel L. and Pati, Arun K.},
  journal = {Phys. Rev. Lett.},
  volume = {98},
  issue = {8},
  pages = {080502},
  numpages = {4},
  year = {2007},
  month = {Feb},
  publisher = {American Physical Society},
  doi = {10.1103/PhysRevLett.98.080502},
  url = {https://link.aps.org/doi/10.1103/PhysRevLett.98.080502}
}

@Article{Muller2012GPTblackhole,
author={M{\"u}ller, Markus P.
and Oppenheim, Jonathan
and Dahlsten, Oscar C.O.},
title={The black hole information problem beyond quantum theory},
journal={Journal of High Energy Physics},
year={2012},
month={Sep},
day={25},
volume={2012},
number={9},
pages={116},
issn={1029-8479},
doi={10.1007/JHEP09(2012)116},
url={https://doi.org/10.1007/JHEP09(2012)116}
}

@article{Song2025twotime,
    author = {Song, Minjeong and Parzygnat, Arthur J.},
    title = {Bipartite quantum states admitting a causal explanation},
    journal = {AVS Quantum Science},
    volume = {7},
    number = {4},
    pages = {045002},
    year = {2025},
    month = {12},
    issn = {2639-0213},
    doi = {10.1116/5.0303300},
    url = {https://doi.org/10.1116/5.0303300},
}

@article{Lerose2021influencematrix,
  title = {Influence Matrix Approach to Many-Body Floquet Dynamics},
  author = {Lerose, Alessio and Sonner, Michael and Abanin, Dmitry A.},
  journal = {Phys. Rev. X},
  volume = {11},
  issue = {2},
  pages = {021040},
  numpages = {21},
  year = {2021},
  month = {May},
  publisher = {American Physical Society},
  doi = {10.1103/PhysRevX.11.021040},
  url = {https://link.aps.org/doi/10.1103/PhysRevX.11.021040}
}

@article{Feynman1963influenceMatrix,
title = {The theory of a general quantum system interacting with a linear dissipative system},
journal = {Annals of Physics},
volume = {24},
pages = {118-173},
year = {1963},
issn = {0003-4916},
doi = {https://doi.org/10.1016/0003-4916(63)90068-X},
url = {https://www.sciencedirect.com/science/article/pii/000349166390068X},
author = {R.P Feynman and F.L Vernon}
}

@article{Dowling2024chaosEnt,
  title = {Operational Metric for Quantum Chaos and the Corresponding Spatiotemporal-Entanglement Structure},
  author = {Dowling, Neil and Modi, Kavan},
  journal = {PRX Quantum},
  volume = {5},
  issue = {1},
  pages = {010314},
  numpages = {28},
  year = {2024},
  month = {Feb},
  publisher = {American Physical Society},
  doi = {10.1103/PRXQuantum.5.010314},
  url = {https://link.aps.org/doi/10.1103/PRXQuantum.5.010314}
}

@article{Dowling2023scrabling,
  title = {Scrambling Is Necessary but Not Sufficient for Chaos},
  author = {Dowling, Neil and Kos, Pavel and Modi, Kavan},
  journal = {Phys. Rev. Lett.},
  volume = {131},
  issue = {18},
  pages = {180403},
  numpages = {6},
  year = {2023},
  month = {Nov},
  publisher = {American Physical Society},
  doi = {10.1103/PhysRevLett.131.180403},
  url = {https://link.aps.org/doi/10.1103/PhysRevLett.131.180403}
}

@Article{Roberts2015otoc,
author={Roberts, Daniel A.
and Stanford, Douglas
and Susskind, Leonard},
title={Localized shocks},
journal={Journal of High Energy Physics},
year={2015},
month={Mar},
day={10},
volume={2015},
number={3},
pages={51},
issn={1029-8479},
doi={10.1007/JHEP03(2015)051},
url={https://doi.org/10.1007/JHEP03(2015)051}
}

@article{choi1980some,
  title={Some assorted inequalities for positive linear maps on C*-algebras},
  author={Choi, Man-Duen},
  journal={Journal of Operator Theory},
  pages={271--285},
  year={1980},
  publisher={JSTOR}
}

@article{kadison1952generalized,
  title={A generalized Schwarz inequality and algebraic invariants for operator algebras},
  author={Kadison, Richard V},
  journal={Annals of Mathematics},
  volume={56},
  number={3},
  pages={494--503},
  year={1952},
  publisher={JSTOR}
}

@misc{milekhin2025observable,
      title={Observable and computable entanglement in time}, 
      author={Alexey Milekhin and Zofia Adamska and John Preskill},
      year={2025},
      eprint={2502.12240},
      archivePrefix={arXiv},
      primaryClass={quant-ph},
      url={https://arxiv.org/abs/2502.12240}, 
}

@article{Prosen2007operatorEnt,
  title = {Operator space entanglement entropy in a transverse Ising chain},
  author = {Prosen, Toma\ifmmode \check{z}\else \v{z}\fi{} and Pi\ifmmode \check{z}\else \v{z}\fi{}orn, Iztok},
  journal = {Phys. Rev. A},
  volume = {76},
  issue = {3},
  pages = {032316},
  numpages = {5},
  year = {2007},
  month = {Sep},
  publisher = {American Physical Society},
  doi = {10.1103/PhysRevA.76.032316},
  url = {https://link.aps.org/doi/10.1103/PhysRevA.76.032316},
      eprint={0706.2480},
  archivePrefix={arXiv},
  primaryClass={quant-ph}
}

@article{Zanardi2001operatorEnt,
  title = {Entanglement of quantum evolutions},
  author = {Zanardi, Paolo},
  journal = {Phys. Rev. A},
  volume = {63},
  issue = {4},
  pages = {040304(R)},
  numpages = {4},
  year = {2001},
  month = {Mar},
  publisher = {American Physical Society},
  doi = {10.1103/PhysRevA.63.040304},
  url = {https://link.aps.org/doi/10.1103/PhysRevA.63.040304},
    eprint={quant-ph/0010074},
  archivePrefix={arXiv},
  primaryClass={quant-ph}
}

@Article{Haag1967KMS,
author={Haag, R.
and Hugenholtz, N. M.
and Winnink, M.},
title={On the equilibrium states in quantum statistical mechanics},
journal={Communications in Mathematical Physics},
year={1967},
month={Jun},
day={01},
volume={5},
number={3},
pages={215-236},
issn={1432-0916},
doi={10.1007/BF01646342},
url={https://doi.org/10.1007/BF01646342}
}

@article{martin1959KMStheory,
  title={Theory of many-particle systems. I},
  author={Martin, Paul C and Schwinger, Julian},
  journal={Physical Review},
  volume={115},
  number={6},
  pages={1342},
  year={1959},
  publisher={APS},
  url={https://journals.aps.org/pr/abstract/10.1103/PhysRev.115.1342}
}

@article{kubo1957statistical,
  title={Statistical-mechanical theory of irreversible processes. I. General theory and simple applications to magnetic and conduction problems},
  author={Kubo, Ryogo},
  journal={Journal of the physical society of Japan},
  volume={12},
  number={6},
  pages={570--586},
  year={1957},
  publisher={The Physical Society of Japan},
  url={https://journals.jps.jp/doi/10.1143/JPSJ.12.570}
}

@book{kamenev2023field,
  title={Field theory of non-equilibrium systems},
  author={Kamenev, Alex},
  year={2023},
  publisher={Cambridge University Press}
}

@article{sieberer2016keldysh,
  title={Keldysh field theory for driven open quantum systems},
  author={Sieberer, Lukas M and Buchhold, Michael and Diehl, Sebastian},
  journal={Reports on Progress in Physics},
  volume={79},
  number={9},
  pages={096001},
  year={2016},
  publisher={IOP Publishing},
  url={https://iopscience.iop.org/article/10.1088/0034-4885/79/9/096001},
    eprint={1512.00637},
  archivePrefix={arXiv},
  primaryClass={cond-mat.quant-gas}
}

@Article{Petz1986petzmap,
author={Petz, D{\'e}nes},
title={Sufficient subalgebras and the relative entropy of states of a von Neumann algebra},
journal={Communications in Mathematical Physics},
year={1986},
month={Mar},
day={01},
volume={105},
number={1},
pages={123-131},
issn={1432-0916},
doi={10.1007/BF01212345},
url={https://doi.org/10.1007/BF01212345}
}

@Article{Parzygnat2023SVD,
author={Parzygnat, Arthur J.
and Takayanagi, Tadashi
and Taki, Yusuke
and Wei, Zixia},
title={SVD entanglement entropy},
journal={Journal of High Energy Physics},
year={2023},
month={Dec},
day={18},
volume={2023},
number={12},
pages={123},
issn={1029-8479},
doi={10.1007/JHEP12(2023)123},
url={https://doi.org/10.1007/JHEP12(2023)123},
  eprint={2307.06531},
  archivePrefix={arXiv},
  primaryClass={hep-th} 
}

@Article{Glorioso2024spacetimeMutual,
author={Glorioso, Paolo
and Qi, Xiao-Liang
and Yang, Zhenbin},
title={Space-time generalization of mutual information},
journal={Journal of High Energy Physics},
year={2024},
month={May},
day={31},
volume={2024},
number={5},
pages={338},
issn={1029-8479},
doi={10.1007/JHEP05(2024)338},
url={https://doi.org/10.1007/JHEP05(2024)338},
  eprint={2401.02475},
  archivePrefix={arXiv},
  primaryClass={quant-ph} 
}

@book{weinberg1995quantum,
  title={The quantum theory of fields},
  author={Weinberg, Steven},
  volume={1},
  year={1995},
  publisher={Cambridge university press}
}

@misc{jia2026temporaltomography,
      title={Temporal State Tomography via Quantum Snapshotting the Temporal Quasiprobabilities}, 
      author={Zhian Jia},
      year={2026},
      eprint={2605.02655},
      archivePrefix={arXiv},
      primaryClass={quant-ph},
      url={https://arxiv.org/abs/2605.02655}, 
}

@article{Lie2025stateovertime,
  title = {Multipartite Quantum States over Time from Two Fundamental Assumptions},
  author = {Lie, Seok Hyung and Fullwood, James},
  journal = {Phys. Rev. Lett.},
  volume = {135},
  issue = {23},
  pages = {230204},
  numpages = {8},
  year = {2025},
  month = {Dec},
  publisher = {American Physical Society},
  doi = {10.1103/lbf3-snp8},
  url = {https://link.aps.org/doi/10.1103/lbf3-snp8}
}

@misc{Jia2025TemporalKirkwoodDirac,
   title={Temporal Kirkwood-Dirac Quasiprobability Distribution and Unification of Temporal State Formalisms through Temporal Bloch Tomography}, 
      author={Zhian Jia and Kavan Modi and Dagomir Kaszlikowski},
      year={2026},
      eprint={2601.05294},
      archivePrefix={arXiv},
      primaryClass={quant-ph},
      url={https://arxiv.org/abs/2601.05294}, 
}

@misc{fullwood2025spatiotemporalbornrule,
      title={The spatiotemporal Born rule is quasiprobabilistic}, 
      author={James Fullwood and Zhihao Ma and Zhen Wu},
      year={2025},
      eprint={2507.16919},
      archivePrefix={arXiv},
      primaryClass={quant-ph},
      url={https://arxiv.org/abs/2507.16919}, 
}

@article{margenau1961correlation,
  title={Correlation between measurements in quantum theory},
  author={Margenau, Henry and Hill, Robert Nyden},
  journal={Progress of Theoretical Physics},
  volume={26},
  number={5},
  pages={722--738},
  year={1961},
  publisher={Oxford University Press},
  url={https://academic.oup.com/ptp/article/26/5/722/1936017},
  doi={10.1143/PTP.26.722},
}

@article{Alonso2019otocKD,
  title = {Out-of-Time-Ordered-Correlator Quasiprobabilities Robustly Witness Scrambling},
  author = {Gonz\'{a}lez Alonso, Jos\'{e} Ra\'{u}l and Yunger Halpern, Nicole and Dressel, Justin},
  journal = {Phys. Rev. Lett.},
  volume = {122},
  issue = {4},
  pages = {040404},
  numpages = {7},
  year = {2019},
  month = {Feb},
  publisher = {American Physical Society},
  doi = {10.1103/PhysRevLett.122.040404},
  url = {https://link.aps.org/doi/10.1103/PhysRevLett.122.040404},
  eprint={1806.09637},
  archivePrefix={arXiv},
  primaryClass={quant-ph} 
}

@article{Halpern2018otoc,
  title = {Quasiprobability behind the out-of-time-ordered correlator},
  author = {Yunger Halpern, Nicole and Swingle, Brian and Dressel, Justin},
  journal = {Phys. Rev. A},
  volume = {97},
  issue = {4},
  pages = {042105},
  numpages = {37},
  year = {2018},
  month = {Apr},
  publisher = {American Physical Society},
  doi = {10.1103/PhysRevA.97.042105},
  url = {https://link.aps.org/doi/10.1103/PhysRevA.97.042105},
  eprint={1704.01971},
  archivePrefix={arXiv},
  primaryClass={quant-ph}  
}

@article{larkin1969otoc,
  title={Quasiclassical method in the theory of superconductivity},
  author={Larkin, Anatoly I and Ovchinnikov, Yu N},
  journal={Sov Phys JETP},
  volume={28},
  number={6},
  pages={1200--1205},
  year={1969}
}

@article{jia2024spatiotemporal,
  title={The Spatiotemporal Doubled-Density Operator: A Unified Framework for Analyzing Spatial and Temporal Quantum Processes},
  author={Jia, Zhian and Kaszlikowski, Dagomir},
  journal={Advanced Quantum Technologies},
  volume={7},
  number={11},
  pages={2400102},
  year={2024},
  publisher={Wiley Online Library},
  url={https://advanced.onlinelibrary.wiley.com/doi/10.1002/qute.202400102},
  doi = {10.1002/qute.202400102},
  eprint={2305.15649},
  archivePrefix={arXiv},
  primaryClass={quant-ph}   
}

@article{Leifer2013toward,
	title = {Towards a formulation of quantum theory as a causally neutral theory of Bayesian inference},
	author = {Leifer, M. S. and Spekkens, Robert W.},
	journal = {Phys. Rev. A},
	volume = {88},
	issue = {5},
	pages = {052130},
	numpages = {38},
	year = {2013},
	month = {Nov},
	publisher = {American Physical Society},
	doi = {10.1103/PhysRevA.88.052130},
	url = {https://link.aps.org/doi/10.1103/PhysRevA.88.052130},
  eprint={1107.5849},
  archivePrefix={arXiv},
  primaryClass={quant-ph}   
}

@article{Chiribella2009comb,
  title = {Theoretical framework for quantum networks},
  author = {Chiribella, Giulio and D'Ariano, Giacomo Mauro and Perinotti, Paolo},
  journal = {Phys. Rev. A},
  volume = {80},
  issue = {2},
  pages = {022339},
  numpages = {20},
  year = {2009},
  month = {Aug},
  publisher = {American Physical Society},
  doi = {10.1103/PhysRevA.80.022339},
  url = {https://link.aps.org/doi/10.1103/PhysRevA.80.022339},
  eprint={0904.4483},
  archivePrefix={arXiv},
  primaryClass={quant-ph}   
}

@article{Aharonov2009multi,
  title = {Multiple-time states and multiple-time measurements in quantum mechanics},
  author = {Aharonov, Yakir and Popescu, Sandu and Tollaksen, Jeff and Vaidman, Lev},
  journal = {Phys. Rev. A},
  volume = {79},
  issue = {5},
  pages = {052110},
  numpages = {16},
  year = {2009},
  month = {May},
  publisher = {American Physical Society},
  doi = {10.1103/PhysRevA.79.052110},
  url = {https://link.aps.org/doi/10.1103/PhysRevA.79.052110},
  eprint={0712.0320},
  archivePrefix={arXiv},
  primaryClass={quant-ph}    
}

@article{Parzygnat2023pdo,
  title = {From Time-Reversal Symmetry to Quantum Bayes' Rules},
  author = {Parzygnat, Arthur J. and Fullwood, James},
  journal = {PRX Quantum},
  volume = {4},
  issue = {2},
  pages = {020334},
  numpages = {26},
  year = {2023},
  month = {Jun},
  publisher = {American Physical Society},
  doi = {10.1103/PRXQuantum.4.020334},
  url = {https://link.aps.org/doi/10.1103/PRXQuantum.4.020334},
eprint={2212.08088},
  archivePrefix={arXiv},
  primaryClass={quant-ph}  
}

@Article{Liu2025PDO,
author={Liu, Xiangjing
and Qiu, Yixian
and Dahlsten, Oscar
and Vedral, Vlatko},
title={Quantum causal inference with extremely light touch},
journal={npj Quantum Information},
year={2025},
month={Mar},
day={29},
volume={11},
number={1},
pages={54},
issn={2056-6387},
doi={10.1038/s41534-024-00956-0},
url={https://doi.org/10.1038/s41534-024-00956-0},  
eprint={2303.10544},
  archivePrefix={arXiv},
  primaryClass={quant-ph}  
}

@article{ArvidssonShukur2024KDreview,
doi = {10.1088/1367-2630/ada05d},
url = {https://dx.doi.org/10.1088/1367-2630/ada05d},
year = {2024},
month = {dec},
publisher = {IOP Publishing},
volume = {26},
number = {12},
pages = {121201},
author = {Arvidsson-Shukur, David R M and Braasch Jr, William F and De Bièvre, Stephan and Dressel, Justin and Jordan, Andrew N and Langrenez, Christopher and Lostaglio, Matteo and Lundeen, Jeff S and Halpern, Nicole Yunger},
title = {Properties and applications of the Kirkwood–Dirac distribution},
journal = {New Journal of Physics},
   eprint={2403.18899},
  archivePrefix={arXiv},
  primaryClass={quant-ph}   
}

@article{Song2024causal,
  title = {Causal Classification of Spatiotemporal Quantum Correlations},
  author = {Song, Minjeong and Narasimhachar, Varun and Regula, Bartosz and Elliott, Thomas J. and Gu, Mile},
  journal = {Phys. Rev. Lett.},
  volume = {133},
  issue = {11},
  pages = {110202},
  numpages = {6},
  year = {2024},
  month = {Sep},
  publisher = {American Physical Society},
  doi = {10.1103/PhysRevLett.133.110202},
  url = {https://link.aps.org/doi/10.1103/PhysRevLett.133.110202},
   eprint={2306.09336 },
  archivePrefix={arXiv},
  primaryClass={quant-ph}   
}

@article{Milz2021quantum,
  title = {Quantum Stochastic Processes and Quantum non-Markovian Phenomena},
  author = {Milz, Simon and Modi, Kavan},
  journal = {PRX Quantum},
  volume = {2},
  issue = {3},
  pages = {030201},
  numpages = {81},
  year = {2021},
  month = {Jul},
  publisher = {American Physical Society},
  doi = {10.1103/PRXQuantum.2.030201},
  url = {https://link.aps.org/doi/10.1103/PRXQuantum.2.030201},
   eprint={2012.01894},
  archivePrefix={arXiv},
  primaryClass={quant-ph}    
}

@article {liu2023unification,
    AUTHOR = {Liu, Xiangjing and Jia, Zhian and Qiu, Yixian and Li, Fei and
              Dahlsten, Oscar},
     TITLE = {Unification of spatiotemporal quantum formalisms: mapping
              between process and pseudo-density matrices via multiple-time
              states},
   JOURNAL = {New J. Phys.},
  FJOURNAL = {New Journal of Physics},
    VOLUME = {26},
      YEAR = {2024},
    NUMBER = {March},
     PAGES = {Paper No. 033008, 15},
   MRCLASS = {81P16},
  MRNUMBER = {4730855},
       DOI = {10.1088/1367-2630/ad264c},
       URL = {https://doi.org/10.1088/1367-2630/ad264c},
    eprint={2306.05958},
  archivePrefix={arXiv},
  primaryClass={quant-ph}   
}

@article{fullwood2023quantum,
  title={Quantum dynamics as a pseudo-density matrix},
  author={Fullwood, James},
  journal={Quantum},
  volume={9},
  pages={1719},
  year={2025},
  url={https://doi.org/10.22331/q-2025-04-24-1719},
  eprint={2304.03954},
  archivePrefix={arXiv},
  primaryClass={quant-ph}      
}

@article{fullwood2022quantum,
  title={On quantum states over time},
  author={Fullwood, James and Parzygnat, Arthur J},
  journal={Proceedings of the Royal Society A},
  volume={478},
  number={2264},
  pages={20220104},
  year={2022},
  publisher={The Royal Society},
  url={https://royalsocietypublishing.org/doi/abs/10.1098/rspa.2022.0104},
    eprint={2202.03607},
  archivePrefix={arXiv},
  primaryClass={quant-ph}   
}

@article{horsman2017can,
  title={Can a quantum state over time resemble a quantum state at a single time?},
  author={Horsman, Dominic and Heunen, Chris and Pusey, Matthew F and Barrett, Jonathan and Spekkens, Robert W},
  journal={Proceedings of the Royal Society A: Mathematical, Physical and Engineering Sciences},
  volume={473},
  number={2205},
  pages={20170395},
  year={2017},
  publisher={The Royal Society Publishing},
  url={https://royalsocietypublishing.org/doi/full/10.1098/rspa.2017.0395},
  eprint={1607.03637},
  archivePrefix={arXiv},
  primaryClass={quant-ph}  
}

@article{Dirac1945on,
  title = {On the Analogy Between Classical and Quantum Mechanics},
  author = {Dirac, P. A. M.},
  journal = {Rev. Mod. Phys.},
  volume = {17},
  issue = {2-3},
  pages = {195--199},
  numpages = {0},
  year = {1945},
  month = {Apr},
  publisher = {American Physical Society},
  doi = {10.1103/RevModPhys.17.195},
  url = {https://link.aps.org/doi/10.1103/RevModPhys.17.195}
}

@article{Kirkwood1933quantum,
  title = {Quantum Statistics of Almost Classical Assemblies},
  author = {Kirkwood, John G.},
  journal = {Phys. Rev.},
  volume = {44},
  issue = {1},
  pages = {31--37},
  numpages = {0},
  year = {1933},
  month = {Jul},
  publisher = {American Physical Society},
  doi = {10.1103/PhysRev.44.31},
  url = {https://link.aps.org/doi/10.1103/PhysRev.44.31}
}

@article {jia2023quantumspace,
    AUTHOR = {Jia, Zhian and Song, Minjeong and Kaszlikowski, Dagomir},
     TITLE = {Quantum space-time marginal problem: global causal structure
              from local causal information},
   JOURNAL = {New J. Phys.},
  FJOURNAL = {New Journal of Physics},
    VOLUME = {25},
      YEAR = {2023},
    NUMBER = {December},
     PAGES = {Paper No. 123038, 23},
   MRCLASS = {81P45 (81P16 81P40)},
  MRNUMBER = {4687136},
       DOI = {10.1088/1367-2630/ad1416},
       URL = {https://doi.org/10.1088/1367-2630/ad1416},
  eprint={2303.12819},
  archivePrefix={arXiv},
  primaryClass={quant-ph} 
}

@article{davies1970operational,
  title={An operational approach to quantum probability},
  author={Davies, E Brian and Lewis, John T},
  journal={Communications in Mathematical Physics},
  volume={17},
  number={3},
  pages={239--260},
  year={1970},
  publisher={Springer},
  url={https://link.springer.com/article/10.1007/BF01647093}
}

@article{Vidal1999robustness,
  title = {Robustness of entanglement},
  author = {Vidal, Guifr\'e and Tarrach, Rolf},
  journal = {Phys. Rev. A},
  volume = {59},
  issue = {1},
  pages = {141--155},
  numpages = {0},
  year = {1999},
  month = {Jan},
  publisher = {American Physical Society},
  doi = {10.1103/PhysRevA.59.141},
  url = {https://link.aps.org/doi/10.1103/PhysRevA.59.141},
  eprint={quant-ph/9806094},
  archivePrefix={arXiv},
  primaryClass={quant-ph}     
}

@inproceedings{gutoski2007toward,
 title={Toward a general theory of quantum games},
  author={Gutoski, Gus and Watrous, John},
  booktitle={Proceedings of the thirty-ninth annual ACM symposium on Theory of computing},
  pages={565--574},
  year={2007},
  url={https://dl.acm.org/doi/10.1145/1250790.1250873},
  eprint={quant-ph/0611234},
  archivePrefix={arXiv},
  primaryClass={quant-ph}    
}

@article{zhao2018geometry,
  title = {Geometry of quantum correlations in space-time},
  author = {Zhao, Zhikuan and Pisarczyk, Robert and Thompson, Jayne and Gu, Mile and Vedral, Vlatko and Fitzsimons, Joseph F.},
  journal = {Phys. Rev. A},
  volume = {98},
  issue = {5},
  pages = {052312},
  numpages = {5},
  year = {2018},
  month = {Nov},
  publisher = {American Physical Society},
  doi = {10.1103/PhysRevA.98.052312},
  url = {https://link.aps.org/doi/10.1103/PhysRevA.98.052312},
  eprint={1711.05955},
  archivePrefix={arXiv},
  primaryClass={quant-ph}   
}

@article{cotler2018superdensity,
  title={Superdensity operators for spacetime quantum mechanics},
  author={Cotler, Jordan and Jian, Chao-Ming and Qi, Xiao-Liang and Wilczek, Frank},
  journal={Journal of High Energy Physics},
  volume={2018},
  number={9},
  pages={1--57},
  year={2018},
  publisher={Springer},
  url={https://link.springer.com/article/10.1007/JHEP09(2018)093},
  eprint={1711.03119 },
  archivePrefix={arXiv},
  primaryClass={quant-ph}       
}

@article{griffiths1984consistent,
  title={Consistent histories and the interpretation of quantum mechanics},
  author={Griffiths, Robert B},
  journal={Journal of Statistical Physics},
  volume={36},
  number={1},
  pages={219--272},
  year={1984},
  publisher={Springer},
  url={https://link.springer.com/article/10.1007/BF01015734}
}

@article{oreshkov2012quantum,
  title={Quantum correlations with no causal order},
  author={Oreshkov, Ognyan and Costa, Fabio and Brukner, {\v{C}}aslav},
  journal={Nature Communications},
  volume={3},
  number={1},
  pages={1--8},
  year={2012},
  publisher={Nature Publishing Group},
  url={https://www.nature.com/articles/ncomms2076},
  eprint={1105.4464},
  archivePrefix={arXiv},
  primaryClass={quant-ph}     
}

@article{Wigner1932on,
  title = {On the Quantum Correction For Thermodynamic Equilibrium},
  author = {Wigner, E.},
  journal = {Phys. Rev.},
  volume = {40},
  issue = {5},
  pages = {749--759},
  numpages = {0},
  year = {1932},
  month = {Jun},
  publisher = {American Physical Society},
  doi = {10.1103/PhysRev.40.749},
  url = {https://link.aps.org/doi/10.1103/PhysRev.40.749}
}

@article {wei2022antilinear,
    AUTHOR = {Wei, Lu and Jia, Zhian and Kaszlikowski, Dagomir and Tan, Sheng},
     TITLE = {Antilinear superoperator, quantum geometric invariance, and antilinear symmetry for higher-dimensional quantum systems},
   JOURNAL = {Quantum Inf. Process.},
  FJOURNAL = {Quantum Information Processing},
    VOLUME = {23},
      YEAR = {2024},
    NUMBER = {8},
     PAGES = {Paper No. 290, 35},
      ISSN = {1570-0755},
   MRCLASS = {81S22 (81P40 81Q70)},
  MRNUMBER = {4779602},
       DOI = {10.1007/s11128-024-04499-3},
       URL = {https://doi.org/10.1007/s11128-024-04499-3},
 eprint={2202.10989},
  archivePrefix={arXiv},
  primaryClass={quant-ph} 
}

@article{fitzsimons2015quantum,
  title={Quantum correlations which imply causation},
  author={Fitzsimons, Joseph F and Jones, Jonathan A and Vedral, Vlatko},
  journal={Scientific Reports},
  volume={5},
  number={1},
  pages={1--7},
  year={2015},
  publisher={Nature Publishing Group},
  url={https://www.nature.com/articles/srep18281},
  eprint={1302.2731},
  archivePrefix={arXiv},
  primaryClass={quant-ph} 
}

@article{Gell-mann1962symmetry,
  title = {Symmetries of Baryons and Mesons},
  author = {Gell-Mann, Murray},
  journal = {Phys. Rev.},
  volume = {125},
  issue = {3},
  pages = {1067--1084},
  numpages = {0},
  year = {1962},
  month = {Feb},
  publisher = {American Physical Society},
  doi = {10.1103/PhysRev.125.1067},
  url = {https://link.aps.org/doi/10.1103/PhysRev.125.1067}
}

@Book{Nielsen2010,
  title     = {Quantum computation and quantum information},
  publisher = {Cambridge university press},
  year      = {2010},
  author    = {Nielsen, Michael A and Chuang, Isaac L},
  url       = {http://www.cambridge.org/cn/academic/subjects/physics/quantum-physics-quantum-information-and-quantum-computation/quantum-computation-and-quantum-information-10th-anniversary-edition?format=PB&isbn=9781107002173#eUc7irgofUw4ZEd5.97},
}

@Article{bell1964,
  author  = {Bell, John S},
  title   = {On the Einstein Podolsky Rosen paradox},
  journal = {Physics},
  year    = {1964},
  volume  = {1},
  pages   = {195--200},
}

@Article{kochen1967problem,
  author  = {Kochen, Simon and Specker, Ernst P},
  title   = {The problem of hidden variables in quantum mechanics},
  journal = {J. Math. Mech.},
  year    = {1967},
  volume  = {17},
  pages   = {59-87},
  url={https://www.jstor.org/stable/24902153}
}

@Article{Ryu2006,
  author    = {Ryu, Shinsei and Takayanagi, Tadashi},
  title     = {Holographic Derivation of Entanglement Entropy from the anti--de Sitter Space/Conformal Field Theory Correspondence},
  journal   = {Phys. Rev. Lett.},
  year      = {2006},
  volume    = {96},
  pages     = {181602},
  month     = {May},
  doi       = {10.1103/PhysRevLett.96.181602},
  issue     = {18},
  numpages  = {4},
  publisher = {American Physical Society},
  url       = {https://link.aps.org/doi/10.1103/PhysRevLett.96.181602},
}

@article{Pollock2018processtensor,
  title = {Non-Markovian quantum processes: Complete framework and efficient characterization},
  author = {Pollock, Felix A. and Rodr\'{\i}guez-Rosario, C\'esar and Frauenheim, Thomas and Paternostro, Mauro and Modi, Kavan},
  journal = {Phys. Rev. A},
  volume = {97},
  issue = {1},
  pages = {012127},
  numpages = {13},
  year = {2018},
  month = {Jan},
  publisher = {American Physical Society},
  doi = {10.1103/PhysRevA.97.012127},
  url = {https://link.aps.org/doi/10.1103/PhysRevA.97.012127}
}

@article{Carignano:2024jxb,
    author = "Carignano, Stefano and Tagliacozzo, Luca",
    title = "{Loschmidt echo, emerging dual unitarity and scaling of generalized temporal entropies after quenches to the critical point}",
    eprint = "2405.14706",
    archivePrefix = "arXiv",
    primaryClass = "cond-mat.stat-mech",
    doi = "10.22331/q-2025-09-16-1859",
    journal = "Quantum",
    volume = "9",
    pages = "1859",
    year = "2025"
}
